\documentclass[12pt]{article} 

\usepackage[a4paper, total={7in, 8.5in}]{geometry}

\usepackage[sectionbib]{natbib}
\usepackage[font=small,labelfont=bf]{caption}
\usepackage[pdfencoding=auto,
            colorlinks = true,
            urlcolor  = blue,
            ]{hyperref}

\usepackage{comment}

\usepackage{graphicx}%
\usepackage{subcaption}
\usepackage{tikz}
\usetikzlibrary{shapes,snakes}
\usetikzlibrary{positioning}

\usepackage{multirow}%
\usepackage{amssymb,amsfonts}%
\usepackage{amsthm}%
\usepackage{mathrsfs}%
\usepackage{amsmath,mathtools}
\DeclareMathOperator{\Hessian}{Hess}
\mathtoolsset{showonlyrefs=true}

\usepackage[title]{appendix}%
\usepackage{xcolor}%
\usepackage{textcomp}%
\usepackage{manyfoot}%
\usepackage{booktabs}%

\usepackage{algorithm}
\usepackage{algpseudocode}


\usepackage{listings}%

\theoremstyle{thmstyleone}%
\newtheorem{theorem}{Theorem}
\newtheorem{proposition}[theorem]{Proposition}%

\newtheorem{assumption}{Assumption}
\newtheorem{lemma}[theorem]{Lemma} 
\newtheorem{corollary}[theorem]{Corollary}

\theoremstyle{thmstyletwo}%

\theoremstyle{thmstylethree}%

\title{Simulation Based Composite Likelihood}

\date{16th January 2025}

\author{Lorenzo Rimella\\
ESOMAS, University of Turin, Via Verdi 8, Turin, 10124, Italy\\
Statistics Initiative, Collegio Carlo Alberto, Piazza Arbarello 8, Turin, 10122, Italy\\
Chris Jewell\\
Mathematical Sciences, Lancaster University, Lancaster, LA14YF, United Kingdom\\
Paul Fearnhead\\
Mathematical Sciences, Lancaster University, Lancaster, LA14YF, United Kingdom
}

\begin{document}
 \maketitle
 

\abstract{Inference for high-dimensional hidden Markov models is challenging due to the exponential-in-dimension computational cost of calculating the likelihood. To address this issue, we introduce an innovative composite likelihood approach called ``Simulation Based Composite Likelihood'' (SimBa-CL). With SimBa-CL, we approximate the likelihood by the product of its marginals, which we estimate using Monte Carlo sampling. In a similar vein to approximate Bayesian computation (ABC), SimBa-CL requires multiple simulations from the model, but, in contrast to ABC, it provides a likelihood approximation that guides the optimization of the parameters. Leveraging automatic differentiation libraries, it is simple to calculate gradients and Hessians to not only speed up optimization but also to build approximate confidence sets. We present extensive empirical results which validate our theory and demonstrate its advantage over SMC, and apply SimBa-CL to real-world Aphtovirus data.}


\section{Introduction}

Discrete-state hidden Markov models (HMMs) are common in many applications, such as epidemics \cite[]{keeling2011modeling}, systems biology \cite[]{wilkinson2018stochastic} and ecology \cite[]{glennie2023hidden}. Increasingly there is interest in individual-based models \cite[e.g.][]{rimella2022inference}, in which the HMM explicitly describes the state of each individual agent in a population. For example, an individual-based epidemic model may characterize each person in a population as having a latent state, being either susceptible, infected, or recovered.  This state is typically observed noisily, with a sample of individuals being detected as infected with a possibly imperfect diagnostic test \cite[e.g.][]{cocker2023drum}.  Thus, whilst there may be only a small number of states for each individual, this corresponds to a latent state-space that grows exponentially with the number of individuals.

In theory, likelihood calculations for such discrete-state HMMs are tractable using the forward-backward recursions \cite[]{scott2002bayesian}. However, the computational cost of these recursions is at least linear in the size of the state-space of the HMM: this means that they are infeasible for individual-based models with moderate or larger population sizes. This has led to a range of approximate inference methods. These include Monte Carlo methods such as Markov chain Monte Carlo (MCMC) and sequential Monte Carlo (SMC). Whilst such methods can work well, often they scale poorly with the population size -- which may lead to poor mixing of MCMC algorithms or large Monte Carlo variance of the weights in SMC. An alternative approach is approximate Bayesian computation (ABC), where one simulates from the model for different parameter values, and then approximates the posterior for the parameter based on how similar each simulated dataset is to the true data. Such a method needs informative, low-dimensional, summary statistics to be available so that one can accurately measure how close a simulated dataset is to the true data. Furthermore, ABC can struggle with complex models with many parameters, as the number of summary statistics needs to increase with the number of parameters \cite[]{fearnhead2012constructing}. 

In this paper, we consider individual-based HMMs where we have individual-level observations. We present a computationally efficient method for inference that is based on the simple observation: if we fix the state of all members of the population except one, then we can analytically calculate the conditional likelihood of that one individual using forward-backward recursions. This idea has been used before within MCMC algorithms that update the state of each individual in turn conditional on the states of the other individuals \cite[]{fintzi2017efficient,touloupou2020scalable}. Here we use it in a different way. By simulating multiple realizations of the states of the other individuals we can average the conditional likelihood to obtain a Monte Carlo estimate of the likelihood of the data for a given individual. We then sum the log of these estimated likelihoods over individuals to obtain a composite log-likelihood \cite[]{varin2008composite} that can be maximized using, for example, stochastic gradient ascent, to estimate the parameters.  

We introduce the general class of models we consider in Section \ref{sec:model}. We then show how to obtain a Monte Carlo estimate of the likelihood for the observations associated with a single individual, which can be used as the basis of a composite likelihood for our model. The calculation of the likelihood for each individual involves accounting for feedback between the state of the individual in question, and the probability distribution of future states of the rest of the population. A computationally more efficient method can be obtained by ignoring this feedback -- and we present theory that bounds the error of this approach, and shows that it can decay to zero as the population size tends to infinity. Then in Section \ref{sec:var_sens_matrix} we show how we can get confidence regions around estimators based on maximizing our composite likelihood. We then demonstrate the efficiency for individual-based epidemic models both on simulated data and on data from the 2001 UK foot and mouth outbreak.

\section{Model} \label{sec:model}

\subsection{Notation}

Given the integer $t \in \mathbb{N}$, we denote the set of integers from $1$ to $t$ as $[t]$, and we use $[0:t]$ if we want to include $0$. Additionally, we use $[t]$ as shorthand for indexing, for instance $x_{[t]}$ denotes the collection $x_1,\dots, x_t$. If $x$ is an $N$-dimensionsal vector, then given an index $n \in [N]$, we use $x^n$ to denote the $n$th component of $x$ and $x^{\setminus n}$ to denote the $\left (N-1\right )$-dimensional vector obtained by removing the $n$th component from $x$. If required, we augment the superscript notation and use $x^{(i)}$ to refer to the vector $x^{(i)}$ with components $x^{(i),n}$. 
For a finite and discrete set $\mathcal{S}$, we represent the cardinality of $\mathcal{S}$ as $\textbf{card}\left (S\right )$, and we use the shorthand $\sum_{x}$ to express the sum over all elements of $\mathcal{S}$. We use bold font to denote random variables and regular font for deterministic quantities. For the underlying probability measure, we commonly use $p$ and, for the sake of clarity, we focus on its functional form, for instance, we use $p\left (x_t|x_{t-1},\theta\right )$ for the probability of $\mathbf{x}_t=x_t$ given $\mathbf{x}_{t-1}=x_{t-1}$ and the parameters $\theta$.

\begin{figure*}[httb!]
\centering
\begin{subfigure}{0.325\textwidth}
\resizebox{0.9\textwidth}{!}{
\begin{tikzpicture}[-latex, auto, node distance =2 cm and 2.5cm ,on grid ,state/.style ={ circle ,top color =white , draw , text=blue,font=\bfseries\Large , minimum width =1.2cm}]

\node[state] (Xt-1) at (0,0) {$x_{{t-1}}$};
\node[state] (Xt) [right =of Xt-1] {$x_{t}$};
\node[state] (Xt+1) [right =of Xt  ] {$x_{{t+1}}$};

\node[draw=none] (none)[above =of Xt-1] {};
\node[state] (Yt-1)[below =of Xt-1] {$y_{t-1}$};
\node[state] (Yt)[below =of Xt] {$y_{t}$};
\node[state] (Yt+1)[below =of Xt+1] {$y_{t+1}$};

\node[draw=none] (fYt-1)[below =of Yt-1] {};
\node[draw=none] (fYt)  [below =of Yt  ] {};
\node[draw=none] (fYt+1)[below =of Yt+1] {};

\node[draw=none] (ffYt-1)[below =of fYt-1] {};
\node[draw=none] (ffYt)  [below =of fYt  ] {};
\node[draw=none] (ffYt+1)[below =of fYt+1] {};

{};

\path (-2,0) edge [bend left = 0] node {} (Xt-1);
\path (Xt-1) edge [bend left = 0] node {} (Xt);
\path (Xt) edge [bend left = 0] node {} (Xt+1);
\path (Xt+1) edge [bend left = 0] node {} (7,0);

\path (Xt-1) edge [bend left = 40, color = red] node {} (Yt-1);
\path (Xt) edge [bend left = 40, color = red] node {} (Yt);
\path (Xt+1) edge [bend left = 40, color = red] node {} (Yt+1);

\end{tikzpicture}
}
\end{subfigure}
\begin{subfigure}{0.625\textwidth}
\hfill
\resizebox{0.9\textwidth}{!}{
\begin{tikzpicture}[-latex, auto, node distance =1.3 cm and 4cm ,on grid ,state/.style ={ circle ,top color =white , draw , text=blue,font=\bfseries\large, minimum width =1.2 cm}]

\node[state] (Xt-11) at (0,0) {$x_{{t-1}}^1$};
\node[state] (Xt1) [right =of Xt-11]   {$x_{t}^1$};
\node[state] (Xt+11) [right =of Xt1  ] {$x_{{t+1}}^1$};

\node[state] (Yt-11)[below =of Xt-11]   {$y^1_{t-1}$};
\node[state] (Yt1)[below =of Xt1]       {$y^1_{t}$};
\node[state] (Yt+11)[below =of Xt+11]   {$y^1_{t+1}$};

\node[state] (Yt-12)[above =of Xt-11]   {$y^2_{t-1}$};
\node[state] (Yt2)[above =of Xt1]       {$y^2_{t}$};
\node[state] (Yt+12)[above =of Xt+11]   {$y^2_{t+1}$};

\node[state] (Xt-12) [above =of Yt-12] {$x_{t-1}^2$};
\node[state] (Xt2) [above =of Yt2]     {$x_{t}^2$};
\node[state] (Xt+12) [above =of Yt+12] {$x_{t+1}^2$};

\node[state] (Yt-13)[above =of Xt-12]   {$y^3_{t-1}$};
\node[state] (Yt3)[above =of Xt2]       {$y^3_{t}$};
\node[state] (Yt+13)[above =of Xt+12]   {$y^3_{t+1}$};

\node[state] (Xt-13) [above =of Yt-13] {$x_{t-1}^3$};
\node[state] (Xt3) [above =of Yt3]     {$x_{t}^3$};
\node[state] (Xt+13) [above =of Yt+13] {$x_{t+1}^3$};

\node[state] (Yt-14)[above =of Xt-13]   {$y^4_{t-1}$};
\node[state] (Yt4)[above =of Xt3]       {$y^4_{t}$};
\node[state] (Yt+14)[above =of Xt+13]   {$y^4_{t+1}$};

\node[state] (Xt-14) [above =of Yt-14] {$x_{t-1}^4$};
\node[state] (Xt4) [above =of Yt4]     {$x_{t}^4$};
\node[state] (Xt+14) [above =of Yt+14] {$x_{t+1}^4$};

\path (-4, 0  ) edge node {} (Xt-11);
\path (-4, 2.6) edge [bend right = 0] node {} (Xt-11);
\path (-4, 5.2) edge [bend right = 0] node {} (Xt-11);
\path (-4, 7.8) edge [bend right = 0] node {} (Xt-11);

\path (-4, 0  ) edge [bend left = 0] node {} (Xt-12);
\path (-4, 2.6) edge node {} (Xt-12);
\path (-4, 5.2) edge [bend right = 0] node {} (Xt-12);
\path (-4, 7.8) edge [bend right = 0] node {} (Xt-12);

\path (-4, 0  ) edge [bend left = 0] node {} (Xt-13);
\path (-4, 2.6) edge [bend left = 0] node {} (Xt-13);
\path (-4, 5.2) edge node {} (Xt-13);
\path (-4, 7.8) edge [bend right = 0] node {} (Xt-13);

\path (-4, 0  ) edge [bend left = 0] node {} (Xt-14);
\path (-4, 2.6) edge [bend left = 0] node {} (Xt-14);
\path (-4, 5.2) edge [bend left = 0] node {} (Xt-14);
\path (-4, 7.8) edge node {} (Xt-14);

\path (Xt-11) edge node {} (Xt1);
\path (Xt-12) edge [bend right = 0] node {} (Xt1);
\path (Xt-13) edge [bend right = 0] node {} (Xt1);
\path (Xt-14) edge [bend right = 0] node {} (Xt1);

\path (Xt-11) edge [bend left = 0] node {} (Xt2);
\path (Xt-12) edge node {} (Xt2);
\path (Xt-13) edge [bend right = 0] node {} (Xt2);
\path (Xt-14) edge [bend right = 0] node {} (Xt2);

\path (Xt-11) edge [bend left = 0] node {} (Xt3);
\path (Xt-12) edge [bend left = 0] node {} (Xt3);
\path (Xt-13) edge node {} (Xt3);
\path (Xt-14) edge [bend right = 0] node {} (Xt3);

\path (Xt-11) edge [bend left = 0] node {} (Xt4);
\path (Xt-12) edge [bend left = 0] node {} (Xt4);
\path (Xt-13) edge [bend left = 0] node {} (Xt4);
\path (Xt-14) edge node {} (Xt4);

\path (Xt1) edge node {} (Xt+11);
\path (Xt2) edge [bend right = 0] node {} (Xt+11);
\path (Xt3) edge [bend right = 0] node {} (Xt+11);
\path (Xt4) edge [bend right = 0] node {} (Xt+11);

\path (Xt1) edge [bend left = 0] node {} (Xt+12);
\path (Xt2) edge node {} (Xt+12);
\path (Xt3) edge [bend right = 0] node {} (Xt+12);
\path (Xt4) edge [bend right = 0] node {} (Xt+12);

\path (Xt1) edge [bend left = 0] node {} (Xt+13);
\path (Xt2) edge [bend left = 0] node {} (Xt+13);
\path (Xt3) edge node {} (Xt+13);
\path (Xt4) edge [bend right = 0] node {} (Xt+13);

\path (Xt1) edge [bend left = 0] node {} (Xt+14);
\path (Xt2) edge [bend left = 0] node {} (Xt+14);
\path (Xt3) edge [bend left = 0] node {} (Xt+14);
\path (Xt4) edge node {} (Xt+14);

\path (Xt+11) edge node {} (12, 0  );
\path (Xt+11) edge node {} (12, 2.6);
\path (Xt+11) edge node {} (12, 5.2);
\path (Xt+11) edge node {} (12, 7.8);

\path (Xt+12) edge [bend right = 0] node {} (12, 0  );
\path (Xt+12) edge [bend right = 0] node {} (12, 2.6);
\path (Xt+12) edge [bend right = 0] node {} (12, 5.2);
\path (Xt+12) edge [bend right = 0] node {} (12, 7.8);

\path (Xt+13) edge [bend right = 0] node {} (12, 0  );
\path (Xt+13) edge [bend right = 0] node {} (12, 2.6);
\path (Xt+13) edge [bend right = 0] node {} (12, 5.2);
\path (Xt+13) edge [bend right = 0] node {} (12, 7.8);

\path (Xt+14) edge [bend right = 0] node {} (12, 0  );
\path (Xt+14) edge [bend right = 0] node {} (12, 2.6);
\path (Xt+14) edge [bend right = 0] node {} (12, 5.2);
\path (Xt+14) edge [bend right = 0] node {} (12, 7.8);

\path (Xt-11) edge [bend left = 40, color = red] node {} (Yt-11);
\path (Xt-12) edge [bend left = 40, color = red] node {} (Yt-12);
\path (Xt-13) edge [bend left = 40, color = red] node {} (Yt-13);
\path (Xt-14) edge [bend left = 40, color = red] node {} (Yt-14);

\path (Xt1) edge [bend left = 40, color = red] node {} (Yt1);
\path (Xt2) edge [bend left = 40, color = red] node {} (Yt2);
\path (Xt3) edge [bend left = 40, color = red] node {} (Yt3);
\path (Xt4) edge [bend left = 40, color = red] node {} (Yt4);

\path (Xt+11) edge [bend left = 40, color = red] node {} (Yt+11);
\path (Xt+12) edge [bend left = 40, color = red] node {} (Yt+12);
\path (Xt+13) edge [bend left = 40, color = red] node {} (Yt+13);
\path (Xt+14) edge [bend left = 40, color = red] node {} (Yt+14);

\end{tikzpicture}
}
\end{subfigure}
\caption{Left: conditional independence structure of a standard HMM. Right: conditional independence structure of an HMM satisfying \eqref{eq:model_factorization}, with $N=4$.}\label{fig:HMM}
\end{figure*}

\subsection{Hidden Markov Models and Likelihood Computation}

A hidden Markov model (HMM) $\left (\mathbf{x}_0, \left (\mathbf{x}_t, \mathbf{y}_t\right )_{t \geq 1}\right )$ is a stochastic process where the unobserved process $\left (\mathbf{x}_t\right )_{t \geq 0}$ is a Markov chain, and the observed process $\left (\mathbf{y}_t\right )_{t \geq 1}$ is such that, for any $t\geq 1$, $\mathbf{y}_t$ is conditionally independent of all the other variables given $\mathbf{x}_t$. See \cite{Chopin2020} for a comprehensive review of HMMs. 

Within this paper, we focus on HMMs with finite-dimensional state-spaces. Precisely, we consider $\left (\mathbf{x}_t\right )_{t \geq 0}$ to take values on the state-space $\mathcal{X}^N$, which satisfies a product form, $\mathcal{X}^N = \bigtimes_{n \in [N]} \mathcal{X}$, where $\mathcal{X}$ is finite and discrete. We also consider $\left (\mathbf{y}_t\right )_{t \geq 1}$ to take values in the space $\mathcal{Y}^N$, which also satisfies a product form, $\mathcal{Y}^N = \bigtimes_{n \in [N]} \mathcal{Y}$, but here $\mathcal{Y}$ can be of any form.

Given a collection of parameters $\theta$, an HMM is fully defined through its components: the initial distribution $p\left (x_0|\theta\right )$, which is the distribution of $\mathbf{x}_0$; the transition kernel $p\left (x_t|x_{t-1},\theta\right )$, which is the distribution of $\mathbf{x}_t$ given $\mathbf{x}_{t-1}$; and the emission distribution $p\left (y_t|x_t,\theta\right )$, which is the distribution of $\mathbf{y}_t$ given $\mathbf{x}_t$. For a given time horizon $T \in \mathbb{N}$, we may assume that the data sequence $y_1,\dots,y_T$ is generated from the aforementioned hidden Markov model with parameters $\theta^\star$. Our primary interest is in inferring the parameter $\theta^\star$, responsible for generating the data or, in cases where the model is not fully identifiable, a set of parameters that are equally likely. Given the assumption that $\mathcal{X}$ is finite and discrete, the probability distribution $p\left (x_0|\theta\right )$ takes the form of a probability vector with $\textbf{card}\left (\mathcal{X}\right )^N$ elements, while $p\left (x_t|x_{t-1}, \theta\right )$ corresponds to a $\textbf{card}\left (\mathcal{X}\right )^N \times \textbf{card}\left (\mathcal{X}\right )^N$ stochastic matrix.
The computation of the likelihood for HMMs with discrete state-space is relatively straightforward and involves marginalization over the entire state-space:
\begin{equation} \label{eq:HMM_likelihood}
\begin{split}
    &p\left (y_{[T]}|\theta\right )= \sum_{x_{[0:T]}} p\left (x_0|\theta\right ) \prod_{t \in [T]} p\left (x_t|x_{t-1},\theta\right ) p\left (y_t|x_t, \theta\right ).
\end{split}
\end{equation}
In practice, to avoid marginalizing on an exponential-in-time state-space, the likelihood is recursively computed using the forward algorithm, which recursively computes the filtering distribution $p(x_t|y_{t},\theta)$ and the likelihood increments $p(y_t|y_{[t-1]},\theta)$. The $t+1$ step of the forward algorithm comprises two operations, namely, prediction:
\begin{equation}
    \begin{split}
    &\begin{Bmatrix}
        p\left (x_{t+1}|x_{t},\theta\right )\\
        p\left (x_{t}|y_{[t]},\theta\right )
    \end{Bmatrix}
    \overset{\text{prediction}}{\longrightarrow}
    \begin{Bmatrix}
        p\left (x_{t+1}|y_{[t]},\theta\right ) = \sum\limits_{x_{t}} p\left (x_{t+1}|x_{t},\theta\right )p\left (x_{t}|y_{[t]},\theta\right )
    \end{Bmatrix},
    \end{split}
\end{equation}
where the transition kernel is applied to the previous filtering distribution, and correction:
\begin{equation}
    \begin{split}
    &\begin{Bmatrix}
        p\left (y_{t+1}|x_{t+1},\theta\right )\\
        p\left (x_{t+1}|y_{[t]},\theta\right )
    \end{Bmatrix}
    \overset{\text{correction}}{\longrightarrow}\begin{Bmatrix}
        p\left (x_{t+1}|y_{[t+1]},\theta\right ) = \frac{p(y_{t+1}|x_{t+1},\theta) p\left (x_{t+1}|y_{[t]},\theta\right )}{p\left (y_{t+1}|y_{[t]},\theta\right )}\\
        p\left (y_{t+1}|y_{[t]},\theta\right ) \hspace{-0.125cm}= \hspace{-0.25cm}\sum\limits_{x_{t+1}} p(y_{t+1}|x_{t+1},\theta) p\left (x_{t+1}|y_{[t]},\theta\right )
    \end{Bmatrix}
    \end{split}
\end{equation}
from which the likelihood increments $p\left (y_{t}|y_{[t-1]},\theta\right )$, with $p\left (y_{1}|y_{[0]},\theta\right ) \coloneqq p\left (y_{1}|\theta\right )$, are then combined to compute the likelihood:
\begin{equation}
    p\left (y_{[T]}|\theta\right ) =\prod_{t \in [T]} p\left (y_{t}|y_{[t-1]},\theta\right ).
\end{equation}

Despite its simplicity, the forward algorithm necessitates a marginalization on the full state-space, incurring a computational cost that is, at worst, quadratic in the cardinality of the state-space. This translates to a  complexity of $\mathcal{O}\left (\mathbf{card}\left (\mathcal{X}\right )^{2N}\right )$, making the forward algorithm unfeasible for large values of $N$. 


\subsection{Factorial Structure} \label{sec:factorial_structure}


We consider HMMs with initial distribution, transition kernel, and emission distribution that satisfy the following factorizations:
\begin{equation} \label{eq:model_factorization}
    \begin{split}
        &p\left (x_0|\theta\right ) = \prod_{n \in [N]} p\left (x_0^n|\theta\right ),\\
        &p\left (x_t|x_{t-1}, \theta\right ) = \prod_{n \in [N]} p\left (x^n_t|x_{t-1},\theta\right ),\\
        &p\left (y_t|x_{t}, \theta\right ) = \prod_{n \in [N]} p\left (y^n_t|x_{t}^n,\theta\right ),
    \end{split}
\end{equation}
which essentially says that we can decompose the initial distribution in $N$ probability vectors of size $\mathbf{card}\left(\mathcal{X}\right)$, the transition kernel in $N$ stochastic matrices that are $\mathbf{card}\left(\mathcal{X}\right) \times \mathbf{card}\left(\mathcal{X}\right)$, whose elements also depend on $x_{t-1}$, and the observation $n$ is conditionally independent of all the other variables given $\mathbf{x}_t^n$, see Figure \ref{fig:HMM}. Note that the introduced factorization does not resolve our problems as the $p\left (x^n_t|x_{t-1},\theta\right )$ still depends on the whole space at time $t-1$, not allowing a full decoupling across the dimensions. Rather, it serves as an essential foundation upon which we construct our approximation. Furthermore, the factorization given by \eqref{eq:model_factorization} is natural in several real-world applications, including epidemics \citep{rimella2022approximating, rimella2022inference}, traffic modelling \citep{silva2015predicting} and finance \citep{samanidou2007agent}.

It is important to note that \eqref{eq:model_factorization}, apart from holding for many real-world discrete-time models, does not require evenly spaced observation. Indeed, we could equivalently work on a more general version of \eqref{eq:model_factorization}, where data are observed at $0=\tau_0<\tau_1<\dots<\tau_t$, and both the transition kernel and the emission distribution are time inhomogeneous \citep{rimella2022inference,duffield2023state}:
\begin{equation} \label{eq:general_model_factorization}
    \begin{split}
        &p_0\left (x_0|\theta\right ) = \prod_{n \in [N]} p_0\left (x_0^n|\theta\right ),\\
        &p_{\tau_t}\left (x_{\tau_t}|x_{\tau_{t-1}}, \theta\right ) = \prod_{n \in [N]} p_{\tau_t}\left (x_{\tau_t}^n|x_{\tau_{t-1}},\theta\right ),\\
        &p_{\tau_t}\left (y_{\tau_t}|x_{\tau_t}, \theta\right ) = \prod_{n \in [N]} p_{\tau_t}\left (y^n_{\tau_t}|x_{\tau_t}^n,\theta\right ).
    \end{split}
\end{equation}
To avoid cumbersome notation we use the formulation from \eqref{eq:model_factorization} throughout the paper.


\section{Simulation Based Composite Likelihood: SimBa-CL} \label{sec:simba_cl}

From the model structure shown in Section \ref{sec:factorial_structure}, if we fix all but one component of the latent process, $x^{\backslash n}_{[T]}$ say, we can leverage the factorization and calculate probabilities related to the time-trajectory of the remaining state, $x^n_{[T]}$, with a computational cost that is $\mathcal{O}\left (\mathbf{card}\left (\mathcal{X}\right )\right )$. This idea has been used within Gibbs-style MCMC updates for epidemics see \cite{fintzi2017efficient,touloupou2020scalable}. We show how to use this idea, together with using Monte Carlo to average over $x^{\backslash n}_{[T]}$, to calculate the marginal likelihoods $p (y^n_{[T]}|\theta )$. We can then use the product of these marginal likelihoods, $\prod_{n \in [N]} p (y^n_{[T]}|\theta )$, as a composite likelihood \cite[]{varin2008composite} that can be maximized to estimate $\theta$. 

Using $p\left (y_{[T]}|\theta\right ) \approx \prod_{n \in [N]} p\left (y^n_{[T]}|\theta\right )$ still falls short, as the computation of $p\left (y_{[T]}^n|\theta\right )$ continues to require a recursive marginalization on $\mathcal{X}^N$. Yet, we can express the marginal likelihood $p\left (y_{[T]}^n|\theta\right )$ as:
\begin{equation} \label{eq:reduced_cost_marginals}
    \begin{split}
    &p\left (y_{[T]}^n|\theta\right )= \sum_{x_{[0:T-1]}^{\setminus n}} p\left (x_{[0:T-1]}^{\setminus n}|\theta\right ) p\left (y_{[T]}^n|x_{[0:T-1]}^{\setminus n}, \theta\right ),
    \end{split}
\end{equation} 
where:
\begin{equation} \label{eq:conditional_reduced_cost_marginals}
\begin{split}
    &p\left (y_{[T]}^n|x_{[0:T-1]}^{\setminus n}, \theta\right )= \sum_{x_{[0:T]}^n} p\left (x_T^n|x_{T-1}, \theta\right ) p\left (x_{[0:T-1]}^{n}|x_{[0:T-1]}^{\setminus n}, \theta\right )\prod_{t \in [T]} p\left (y_t^n|x_t^n, \theta\right ).
\end{split}
\end{equation}
We have two necessary ingredients for calculating $p\left (y_{[T]}^n|\theta\right )$: firstly, $p\left (y_{[T]}^n|x_{[0:T-1]}^{\setminus n}, \theta\right )$, which demands $T$ recursive marginalizations on $\mathcal{X}$ given $x_{[0:T-1]}^{\setminus n}$; secondly, a marginalization on $\mathcal{X}^{N-1}$ through $p\left (x_{[0:T-1]}^{\setminus n}|\theta\right )$, see \eqref{eq:reduced_cost_marginals}. \eqref{eq:reduced_cost_marginals} using Monte Carlo sampling.

We refer to this procedure as ``Simulation Based Composite Likelihood'', or ``SimBa-CL'' in short. In the following sections, we give an in-depth discussion on SimBa-CL and show how we can target the true marginals of the likelihood and build a likelihood approximation in $\mathcal{O}\left (N^2\right )$, see Section \ref{sec:simba_feed}, how to approximate the marginals of the likelihood and build an approximation of the likelihood in $\mathcal{O}\left (N\right )$, see Section \ref{sec:simba_withoutfeed}, and how to generalize SimBa-CL, see Section \ref{sec:simba_general}. For the sake of presentation, we remove the dependence on the parameter $\theta$ and focus on the filtering aspects of the algorithms for a fixed $\theta$.

\subsection{SimBa-CL with Feedback} \label{sec:simba_feed}

Given efficient sampling from $p \left (x_{[0:T-1]}^{\setminus n} \right )$ and low-cost evaluation of $p \left ( y_{[T]}^n|x_{[0:T-1]}^{\setminus n} \right )$ for a given $x_{[0:T-1]}^{\setminus n}$, we can obtain a Monte Carlo estimate of the marginal likelihood from \eqref{eq:reduced_cost_marginals}:
\begin{equation}\label{eq:sampling_marginals}
    \begin{split}
        p\left (y^n_{[T]}\right ) 
        & \approx
        \frac{1}{P} \sum_{i \in [P]} p\left (y_{[T]}^n|x_{[0:T-1]}^{(i), \setminus n} \right ),
    \end{split}
\end{equation} 
where $P \in \mathbb{N}$ is the number of Monte Carlo samples and $x_{[0:T-1]}^{(i),\setminus n} \sim p\left (x_{[0:T-1]}^{\setminus n}\right )$. Repeating  \eqref{eq:sampling_marginals} for all $n \in [N]$ and computing the product across $n$ of these Monte Carlo estimates represents a reasonable strategy for approximating the likelihood of the model.

Two ingredients are pivotal in the computation of  \eqref{eq:sampling_marginals}: (i) sampling from the model and (ii) calculating $p \left (y_{[T]}^n|x_{[0:T-1]}^{\setminus n} \right )$. Sampling from $p \left (x_{[0:T-1]}^{\setminus n} \right )$ can be achieved by sampling $x_{[0:T-1]}$ from $p \left (x_{[0:T-1]} \right )$ and then selecting the subset $x_{[0:T-1]}^{\setminus n}$. Sampling from the entire process is generally straightforward.

For the computation of $p \left (y_{[T]}^n|x_{[0:T-1]}^{\setminus n} \right )$, it is important to recognize that $p \left (x_{[0:T-1]}^{n}|x_{[0:T-1]}^{\setminus n} \right )$ can be reformulated as a product between the transition dynamics and the probability of observing a certain simulation outcome:
\begin{equation}\label{eq:recursive_feedback_kernel}
\begin{split}
    &p\left (x_{[0:T-1]}^n|x_{[0:T-1]}^{\setminus n}\right ) =
    p\left (x_{0}^n\right ) \prod_{t \in [T-1]} p\left (x_{t}^n|x_{t-1}\right ) f \left (x_{t-1}^n, x_{[0:t]}^{\setminus n} \right ),
\end{split}
\end{equation}
where we refer to $f(x_{t-1}^n, x_{[0:t]}^{\setminus n} )\coloneqq p ( x_{t}^{\setminus n}|x_{t-1}^n, x_{[0:t-1]}^{\setminus n} )$ as the simulation feedback, and so:
\begin{equation} \label{eq:simulation_feed_definition}
\begin{split}
    &f \left (x_{t-1}^n, x_{[0:t]}^{\setminus n} \right )=
    \frac{ \prod\limits_{\bar{n} \in [N] \setminus n} p\left (x_{t}^{\bar{n}}|x_{t-1}\right )}{\sum\limits_{\bar{x}_{t-1}^n} \prod\limits_{\bar{n} \in [N] \setminus n} p\left (x_{t}^{\bar{n}}|\bar{x}_{t-1}^n, x_{[t-1]}^{\setminus n}\right ) p\left (\bar{x}_{t-1}^n| x_{[0:t-1]}^{\setminus n}\right )},
\end{split}
\end{equation}
where $p\left ({x}_{0}^n| x_{[0:0]}^{\setminus n}\right ) = p\left(x_0^n\right)$ for the factorization of the initial distribution. The intuition is that this term accounts for how changing $x^n_{t-1}$ affects the probability of $x^{\setminus n}_t$. More details on the factorization \eqref{eq:recursive_feedback_kernel} and the derivation of the simulation feedback \eqref{eq:simulation_feed_definition} are available in Section A.1 of the supplementary material.

By reformulating $p\left (x_{[0:T-1]}^n|x_{[0:T-1]}^{\setminus n}\right )$ as depicted in \eqref{eq:recursive_feedback_kernel}, we arrive at the following expression:
\begin{equation} \label{eq:simulation_likelihood_with_feedback_subroutine}
\begin{split}
    &p\left (y_{[T]}^n|x_{[0:T-1]}^{\setminus n}\right )= \sum_{x_{[0:T]}^{n}} p\left (x_{0}^n\right ) \prod_{t \in [T-1]} f \left ( x_{t-1}^n, x_{[0:t]}^{\setminus n} \right )\prod_{t \in [T]} p\left (x_{t}^n|x_{t-1}\right )  p\left (y_t^n|x_t^n\right ),
\end{split}
\end{equation}
which resembles the likelihood of an HMM, see \eqref{eq:HMM_likelihood}. Specifically, it comprises the usual transition dynamic term $p\left (x_t^n|x_{t-1}\right )$ accompanied by two likelihood terms: one originating from the simulation outcome $f \left ( x_{t-1}^n, x_{[0:t]}^{\setminus n} \right )$, and another concerning the observation $p\left (y_t^n|x_t^n\right )$. We can then establish a forward algorithm involving two corrections, one that is correcting according to the emission distribution:
\begin{equation}\label{eq:wfeed_corr_obs}
    \begin{split}
    &\begin{Bmatrix}
        p\left (y_{t}^n | x_t^n \right )\\
        p\left (x_{t}^n|y_{[t-1]}^n,x_{[0:t]}^{\setminus n}\right )
    \end{Bmatrix}
    \overset{\substack{\text{observation}\\\text{correction}}}{\longrightarrow}
    \begin{Bmatrix}
        p\left (x_{t}^n|y_{[t]}^n,x_{[0:t]}^{\setminus n}\right ) = \frac{p\left (y_t^n | x_t^n \right ) p\left (x_{t}^n|y_{[t-1]}^n,x_{[0:t]}^{\setminus n}\right )}{p\left (y_{t}^n|y_{[t-1]}^n,x_{[0:t]}^{\setminus n}\right )}\\
        p\left (y_{t}^n|y_{[t-1]}^n,x_{[0:t]}^{\setminus n}\right ) \qquad\qquad\qquad\qquad\qquad\\= \sum\limits_{x_t^n} p\left (y_t^n | x_t^n \right ) p\left (x_{t}^n|y_{[t-1]}^n,x_{[0:t]}^{\setminus n}\right )
    \end{Bmatrix},
    \end{split}
\end{equation}
and the other that is correcting according to the simulation feedback:
\begin{equation*}\label{eq:wfeed_corr_feed}
    \begin{split}
    &\begin{Bmatrix}
        f\left ({x}_{t}^n, x_{[0:t+1]}^{\setminus n}\right )\\
        p\left (x_{t}^n|y_{[t]}^n,x_{[0:t]}^{\setminus n}\right )
    \end{Bmatrix}
    \overset{\substack{\text{feedback}\\\text{correction}}}{\longrightarrow}
    \begin{Bmatrix}
        p\left (x_{t}^n|y_{[t]}^n,x_{[0:t+1]}^{\setminus n}\right )\qquad\qquad\qquad\qquad\qquad\quad\\ = \frac{f\left ({x}_{t}^n, x_{[0:t+1]}^{\setminus n}\right ) p\left (x_{t}^n|y_{[t]}^n,x_{[0:t]}^{\setminus n}\right )}{p\left (x_{t+1}^{\setminus n}|y_{[t]}^n,x_{[0:t]}^{\setminus n}\right )}\\
        p\left (x_{t+1}^{\setminus n}|y_{[t]}^n,x_{[0:t]}^{\setminus n}\right )\qquad\qquad\qquad\qquad\qquad\quad\\
        = \sum\limits_{x_{t}^n} f\left ({x}_{t}^n, x_{[0:t+1]}^{\setminus n}\right ) p\left (x_{t}^n|y_{[t]}^n,x_{[0:t]}^{\setminus n}\right )
    \end{Bmatrix}.
    \end{split}
\end{equation*}
The prediction follows as is in the basic HMM scenario with $p(x_t^n|x_{t-1})$ as transition kernel and $p\left (x_{t-1}^n|y_{[t-1]}^n,x_{[0:t]}^{\setminus n}\right )$ for the distribution to update:
\begin{equation}\label{eq:wfeed_pred_feed}
    \begin{split}
    &\begin{Bmatrix}
        p(x_t^n|x_{t-1})\\
        p\left (x_{t-1}^n|y_{[t-1]}^n,x_{[0:t]}^{\setminus n}\right )
    \end{Bmatrix}
    \overset{\substack{\text{feedback}\\\text{prediction}}}{\longrightarrow}
    \begin{Bmatrix}
        p\left (x_{t}^n|y_{[t-1]}^n,x_{[0:t]}^{\setminus n}\right ) \qquad\qquad\qquad\qquad\\= \sum\limits_{x_{t-1}^n} p(x_t^n|x_{t-1}) p\left (x_{t-1}^n|y_{[t-1]}^n,x_{[0:t]}^{\setminus n}\right )
    \end{Bmatrix}.
    \end{split}
\end{equation}

The computation of the simulation feedback $f\left (x_{t-1}^n, x_{[0:t-1]}^{\setminus n}\right )$ relies on $p\left (x_{t-1}^n|x_{[0:t-1]}^{\setminus n}\right )$, the posterior distribution of ${x}_{t-1}^n$ given the simulation output as observations. Consequently, this interpretation enables the employment of another forward algorithm to compute recursively these intermediate quantities, where the correction step is given by: 
\begin{equation}\label{eq:ffeed_corr}
    \begin{split}
    &\begin{Bmatrix}
        \prod\limits_{\bar{n} \in [N] \setminus n} p\left (x_t^{\bar{n}}| x_{t-1} \right )\\
        p\left (x_{t-1}^n|x_{[0:t-1]}^{\setminus n}\right )
    \end{Bmatrix}
    \overset{\text{correction}}{\longrightarrow} 
    \begin{Bmatrix}
        p\left (x_{t-1}^n|x_{[0:t]}^{\setminus n}\right ) \qquad\qquad\qquad\qquad\qquad\qquad\\= \frac{\prod\limits_{\bar{n} \in [N] \setminus n} p\left (x_t^{\bar{n}}| x_{t-1} \right ) p\left (x_{t-1}^n|x_{[0:t-1]}^{\setminus n}\right )}{p(x_t^{\setminus n}|x_{[0:t-1]}^{\setminus n})}\\
        p(x_t^{\setminus n}|x_{[0:t-1]}^{\setminus n})\qquad\qquad\qquad\qquad\qquad\qquad\\ = \sum\limits_{x_{t-1}^n} \prod\limits_{\bar{n} \in [N] \setminus n} p\left (x_t^{\bar{n}}| x_{t-1} \right ) p\left (x_{t-1}^n|x_{[0:t-1]}^{\setminus n}\right )
    \end{Bmatrix},
    \end{split}
\end{equation}
and the prediction follows:
\begin{equation}\label{eq:ffeed_pred}
    \begin{split}
    &\begin{Bmatrix}
        p\left(x_t^n|x_{t-1}\right)\\
        p\left (x_{t-1}^n|x_{[0:t]}^{\setminus n}\right )
    \end{Bmatrix}
    \overset{\text{prediction}}{\longrightarrow}
    \begin{Bmatrix}
        p\left (x_{t}^n|x_{[0:t]}^{\setminus n}\right ) = \sum\limits_{x_{t-1}^n} p\left(x_t^n|x_{t-1}\right) p\left (x_{t-1}^n|x_{[0:t]}^{\setminus n}\right )
    \end{Bmatrix}.
    \end{split}
\end{equation}

An iterative application of the aforementioned steps provides a collection of likelihood increments on both the simulation output and the observations, enabling the computation of $p\left (y_{[T]}^n|x_{[0:T-1]}^{\setminus n}\right )$ as follows:
\begin{equation} \label{eq:simulation_likelihood}
\begin{split}
    &p\left (y_{[T]}^n|x_{[0:T-1]}^{\setminus n}\right )= p\left (y_T^n|y^n_{[T-1]},x_{[0:T-1]}^{\setminus n}\right ) \prod_{t \in [T-1]} p\left (y_t^n|y^n_{[t-1]},x_{[0:t]}^{\setminus n}\right ) p\left (x_t^{\setminus n}|y^n_{[t-1]},x_{[0:t-1]}^{\setminus n}\right ),
\end{split}
\end{equation}
where $p\left (y_1^n|y^n_{[0]},x_{[0:0]}^{\setminus n}\right ) \coloneqq p\left (y_1^n|x_{0}^{\setminus n}\right )$ and $p\left (x_1^{\setminus n}|y^n_{[0]},x_{[0]}^{\setminus n}\right )\coloneqq p\left (x_1^{\setminus n}|x_{[0]}^{\setminus n}, \theta\right )$. 

\begin{algorithm*}[httb!]
\caption{SimBa-CL with feedback}\label{alg:simulation_likelihood_with_feedback}
\begin{algorithmic}
\Require $p\left (x_0\right )$, $p\left (x_t|x_{t-1}\right ), p\left (y_t|x_t\right )$ and their factorizations
    \For{\textbf{each }$i \in [P]$}
        \State $x_{0}^{(i),} \sim p\left (x_{0}\right )$
        \For{$t \in [T]$}
            \State  $x_{t}^{(i),} \sim p\left (x_{t}|x_{t-1}^{(i),}\right )$ and compute $f\left (x_{t-1}^n,x_{[0:t]}^{(i),\setminus n} \right)$
            \For{\textbf{each }$n \in [N]$}
             \State if $t\neq T$, run the feedback correction \eqref{eq:wfeed_corr_feed} and get: 
             $
             \begin{Bmatrix}
                p\left (x_{t-1}^n|y_{[t-1]}^n,x_{[0:t]}^{(i),\setminus n}\right )\\
                p\left (x_{t}^{(i),\setminus n}|y_{[t-1]}^n,x_{[0:t-1]}^{(i),\setminus n}\right )
             \end{Bmatrix}
             $
             \State Run the prediction \eqref{eq:wfeed_pred_feed} and get: 
             $
             \begin{Bmatrix}
                 p\left (x_{t}^n|y_{[t-1]}^n,x_{[0:t]}^{(i),\setminus n}\right )
             \end{Bmatrix}
             $
             \State Run the observation correction \eqref{eq:wfeed_corr_obs} and get: 
             $
             \begin{Bmatrix}
                 p\left (x_{t}^n|y_{[t]}^n,x_{[0:t]}^{(i),\setminus n}\right )\\
                 p\left (y_{t}^n|y_{[t-1]}^n,x_{[0:t]}^{(i),\setminus n}\right )
             \end{Bmatrix}
             $
             \State Run the correction \eqref{eq:ffeed_corr} and get:
             $
             \begin{Bmatrix}
                 p\left (x_{t-1}^n|x_{[0:t]}^{(i),\setminus n}\right )\\
                 p(x_t^{(i),\setminus n}|x_{[0:t-1]}^{(i),\setminus n})
             \end{Bmatrix}
             $
             \State Run the prediction \eqref{eq:ffeed_pred} and get: 
             $
             \begin{Bmatrix}
                 p\left (x_{t}^n|x_{[0:t]}^{(i),\setminus n}\right )
             \end{Bmatrix}
             $
        \EndFor
        \EndFor
        \State Compute $p\left (y_{[T]}^n|x_{[0:T-1]}^{(i),\setminus n}\right )$ as in \eqref{eq:simulation_likelihood}
    \EndFor
\State Return $\frac{1}{P} \sum_{i \in [P]} p\left (y_{[T]}^n|x_{[0:T-1]}^{(i),\setminus n}\right )$
\end{algorithmic}
\end{algorithm*}

The final algorithm, named ``SimBa-CL with feedback'', is presented in Algorithm \ref{alg:simulation_likelihood_with_feedback} and it shows a computational complexity of $\mathcal{O}\left ( T N^2 P \mathbf{card}\left (\mathcal{X} \right )^2\right )$, wherein $P$, $T$, $N$ come from looping over number of simulations, time steps and dimensions, $\mathbf{card}\left (\mathcal{X}\right )^2$ comes from marginalizing over the state-space $\mathcal{X}$, and the extra $N$ terms refers to the simulation feedback computation.\\ This cost can potentially be reduced to $\mathcal{O}\left ( T N P \max_n \{ \mathbf{card}\left (\mathbf{Neig}\left (n\right )\right ) \} \mathbf{card}\left (\mathcal{X}\right )^2\right )$ if the transition kernel $p\left (x_t^n|x_{t-1},\theta\right )$ presents some local structure. Precisely, if $p\left (x_t^n|x_{t-1},\theta\right )=p\left (x_t^n|\bar{x}_{t-1},\theta\right )$ for any $x^{\mathbf{Neig}\left (n\right )}_{t-1} = \bar{x}^{\mathbf{Neig}\left (n\right )}_{t-1}$, where $\mathbf{Neig}$ represents a function mapping any $n \in [N]$ onto a set in the power set of $[N]$. This indicates that computing the simulation feedback can be computationally cheaper if the inter-dimension interactions are sparse. Also, the algorithm can be readily parallelized across both the dimensions and the simulations, rendering the dependence concerning $(P, N)$ less heavy in the presence of suitable hardware. 

\subsection{SimBa-CL without Feedback} \label{sec:simba_withoutfeed}

An important aspect of SimBa-CL with feedback is that it targets the true marginals of the likelihood. 
However, one could contemplate a strategy involving the removal of the simulation feedback and design a SimBa-CL that is more computationally efficient, while being ``close'' to SimBa-CL with feedback. 

Looking back at \eqref{eq:simulation_likelihood_with_feedback_subroutine}, omitting the simulation feedback yields to:
\begin{equation} \label{eq:simulation_likelihood_no_feedback_subroutine}
\begin{split}
    &p\left (y_{[T]}^n|x_{[0:T-1]}^{\setminus n}\right )
    \approx \sum_{x_{[0:T]}^{n}} p\left (x_{0}^n\right ) \prod_{t \in [T]} p\left (x_{t}^n|x_{t-1}\right )  p\left (y_t^n|x_t^n\right )
    \eqqcolon \tilde{p}\left (y_{[T]}^n|x_{[0:T-1]}^{\setminus n}\right ),
\end{split}
\end{equation}
from which we have the following marginal likelihood approximation:
\begin{equation} \label{eq:reduced_cost_marginals_no_fedback}
	\begin{split}
	p\left (y_{[T]}^n\right ) 
    \approx \sum_{x_{[0:T-1]}^{\setminus n}} p\left (x_{[0:T-1]}^{\setminus n}\right ) \tilde{p}\left (y_{[T]}^n|x_{[0:T-1]}^{\setminus n}\right )
    \eqqcolon \tilde{p}\left (y_{[T]}^n\right ),
	\end{split}
\end{equation} 
where we emphasized that we are relying on an approximation by using the notation $\tilde{p}$. It can be seen that $\tilde{p}\left (y_{[T]}^n\right )$ is still a proper marginal likelihood, as it sums to one when marginalizing on $\mathcal{Y}$, but it is not the true marginal likelihood. 

Upon scrutinizing \eqref{eq:simulation_likelihood_no_feedback_subroutine}, we can recognize the same likelihood structure as  \eqref{eq:HMM_likelihood}. This time, we are isolating our calculation to a single component $n$ and fixing the others through simulation from the model. This suggests that a simple forward algorithm can be run in isolation on each dimension. 
Concretely, the corresponding forward algorithm will require a prediction step:
\begin{equation}\label{eq:pred_nofeed}
    \begin{split}
        &\begin{Bmatrix}
            p\left(x_{t+1}^n|x_{t}\right)\\
            \tilde{p}\left (x^n_{t}|y^n_{[t]}, x_{[0:t-1]}^{\setminus n}\right )
        \end{Bmatrix} 
        \overset{\text{prediction}}{\longrightarrow}
        \begin{Bmatrix}
            \tilde{p}\left (x^n_{t+1}|y^n_{[t]}, x_{[0:t]}^{\setminus n}\right )\qquad\qquad\qquad\qquad\\ = \sum\limits_{x_{t}^n} p\left(x_{t+1}^n|x_{t}\right)\tilde{p}\left (x^n_{t}|y^n_{[t]}, x_{[0:t-1]}^{\setminus n}\right )
        \end{Bmatrix},
    \end{split}
\end{equation}
and a correction step:
\begin{equation}\label{eq:corr_nofeed}
    \begin{split}
        &\begin{Bmatrix}
            p\left(y_{t+1}^n|x_{t+1}^n\right)\\
            \tilde{p}\left (x^n_{t+1}|y^n_{[t]}, x_{[0:t]}^{\setminus n}\right )
        \end{Bmatrix} 
        \overset{\text{correction}}{\longrightarrow} 
        \begin{Bmatrix}
            \tilde{p}\left (x^n_{t+1}|y^n_{[t+1]}, x_{[0:t]}^{\setminus n}\right )\qquad\qquad\qquad\qquad\\ = \frac{p\left(y_{t+1}^n|x_{t+1}^n\right) \tilde{p}\left (x^n_{t+1}|y^n_{[t]}, x_{[0:t]}^{\setminus n}\right )}{\tilde{p}\left (y^n_{t+1}|y^n_{[t]}, x_{[0:t]}^{\setminus n}\right )}\\
            \tilde{p}\left (y^n_{t+1}|y^n_{[t]}, x_{[0:t]}^{\setminus n}\right )\qquad\qquad\qquad\qquad\\ = \sum\limits_{x_{t+1}^n} p\left(y_{t+1}^n|x_{t+1}^n\right) \tilde{p}\left (x^n_{t+1}|y^n_{[t]}, x_{[0:t]}^{\setminus n}\right )
        \end{Bmatrix}.
    \end{split}
\end{equation}

\begin{algorithm*}[h]
\caption{Simulation likelihood without feedback}\label{alg:simulation_likelihood_no_feedback}
\begin{algorithmic}
\Require $p\left (x_0\right )$, $p\left (x_t|x_{t-1}\right )$, $p\left (y_t|x_t\right )$ and their factorization
 \For{\textbf{each }$i \in [P]$}
    \State $x_0^{(i),} \sim p(x_0)$
    \For{$t \in [T]$}
        \State $x_t^{(i),} \sim p\left (x_t| x_{t-1}^{(i),}\right )$
        \For{\textbf{each }$n \in [N]$}
        \State Run the prediction \eqref{eq:pred_nofeed} and get: 
        $
        \begin{Bmatrix}
            \tilde{p}\left (x^n_{t}|y^n_{[t-1]}, x_{[0:t-1]}^{(i),\setminus n}\right )
        \end{Bmatrix}
        $
        \State Run the correction \eqref{eq:corr_nofeed} and get: 
        $
        \begin{Bmatrix}
            \tilde{p}\left (x^n_{t}|y^n_{[t]}, x_{[0:t-1]}^{(i),\setminus n}\right )\\
            \tilde{p}\left (y^n_{t}|y^n_{[t-1]}, x_{[0:t-1]}^{(i),\setminus n}\right )
        \end{Bmatrix}
        $
        \EndFor
        \EndFor
\EndFor
\State Return $\frac{1}{P} \sum_{i \in [P]} \tilde{p}\left (y_{[T]}^n|x_{[0:T-1]}^{i,\setminus n}, \theta\right )$
\end{algorithmic}
\end{algorithm*}

Recursively applying \eqref{eq:pred_nofeed} and \eqref{eq:corr_nofeed} provides a sequence of approximate marginal likelihood increments which can be then used to approximate the marginal likelihood for a fixed simulation in the usual way:
\begin{equation} \label{eq:simulation_likelihood_no_feed}
    \tilde{p}\left (y_{[T]}^n|x_{[0:T-1]}^{\setminus n}\right ) = \prod_{t \in [T]} \tilde{p}\left (y_t^n|y^n_{[t-1]},x_{[0:t-1]}^{\setminus n}\right ),
\end{equation}
where $\tilde{p}\left (y_1^n|y^n_{[0]},x_{[0:0]}^{\setminus n}\right ) \coloneqq \tilde{p}\left (y_1^n|x_{[0]}^{\setminus n}\right )$. The marginal likelihood approximation is then obtained as the mean of the Monte Carlo approximations, and we named this final algorithm SimBa-CL without feedback, see Algorithm \ref{alg:simulation_likelihood_no_feedback}. 

When comparing Algorithm \ref{alg:simulation_likelihood_no_feedback} with Algorithm \ref{alg:simulation_likelihood_with_feedback}, the simplicity of the latter becomes evident. The computational cost is reduced from $\mathcal{O}\left (T N^2 P \mathbf{card}\left (\mathcal{X}\right )^2 \right )$ to 
$\mathcal{O}\left (T N P \mathbf{card}\left (\mathcal{X}\right )^2\right )$. As with Algorithm \ref{alg:simulation_likelihood_with_feedback}, our new SimBa-CL procedure is parallelizable on both $N$ and $P$. 

\subsection{KL-bound for SimBa-CL with and without Feedback}

To evaluate the impact of excluding the simulation feedback from SimBa-CL, and so understand when is worth including the feedback, a natural approach is to compare the two estimates of the marginal likelihood:
\begin{align} 
    &p\left (y_{[T]}^n\right )=  \sum_{x_{[0:T]}^{\setminus n}} p\left (x_{[0:T]}^{\setminus n}\right )
 \sum_{x_{[0:T]}^n} p\left (x_{[0:T]}^{n}|x_{[0:T]}^{\setminus n}\right ) \prod_{t \in [T]} p\left (y_t^n|x_t^n\right ),\\
    &\tilde{p}\left (y_{[T]}^n\right )= \sum_{x_{[0:T]}^{\setminus n}} p\left (x_{[0:T]}^{\setminus n}\right ))
 \sum_{x_{[0:T]}^n} p\left (x_0^{n}\right ) \prod_{t \in [T]} p\left (x_t^n|x_{t-1},\right ) p\left (y_t^n|x_t^n\right ).
\end{align}
As both $p\left (y_{[T]}^n\right )$ and $\tilde{p}\left (y_{[T]}^n\right )$ are probability distributions on $\mathcal{Y}^{T}$, a simple way of comparing them is to measure the Kullback-Leibler divergence, which we denote with $\mathbf{KL}\left [ p\left (\mathbf{y}_{[T]}^n\right ) || \tilde{p}\left (\mathbf{y}_{[T]}^n\right ) \right ]$. 
The objective of this section is then to establish an upper bound for the $\mathbf{KL}$, demonstrating that under suitable assumptions, in the large-$N$ limit, the ``without feedback'' approximate marginal likelihood is a good approximation of the true marginal likelihood.

Naturally, we must rely on technical assumptions for theoretical results, which we will strive to explain from an intuitive perspective as much as possible. Our result relies on some assumptions, which we now explain.
\begin{assumption}\label{ass:kernel}
For any $n, \bar{n} \in [N]$ and for any $x_t^{\bar{n}} \in \mathcal{X}$, if $x_{t-1}, \bar{x}_{t-1} \in \mathcal{X}^N$ are such that $x_{t-1}^{\setminus {n}}= \bar{x}_{t-1}^{\setminus {n}}$ then:
\begin{equation}
\begin{split}
    &\left | p\left (x_t^{\bar{n}}|x_{t-1}\right ) - p\left (x_t^{\bar{n}}|\bar{x}_{t-1}\right ) \right |\leq \frac{1}{N} \left |d_{n,\bar{n}}\left (x_{t-1}^{{n}}\right ) - d_{n,\bar{n}}\left (\bar{x}_{t-1}^{{n}}\right ) \right |,
\end{split}
\end{equation}
where $d_{n,\bar{n}}:\mathcal{X} \to \mathbb{R}_+$. 
\end{assumption}

Assumption \ref{ass:kernel} ensures the boundness of the transition dynamic when altering the states of only ${n} \in [N]$ at time $t-1$. This essentially asserts that changing the state of a single component at time $t-1$ minimally impacts the dynamics of the other components. In essence, the impact is measured in terms of the function $d_{n,\bar{n}}$, which shows how changes in $n$ affect any other dimension $\bar{n}$. This concept is similar in flavour to the decay of correlation property explained in \cite{rebeschini2015can} and \cite{rimella2022exploiting}, which ensures a weak sensitivity of the conditional distributions on any $n$ given a perturbation on any other dimension. Conceptually, the main emphasis is that the dynamics are only impacted by single dimension-level perturbations at the scale $N^{-1}$. In systems of increasing size with interconnected dimensions, we expect that single dimension-level changes have diminishing effects on the dimension as a whole, which is also intuitively plausible in numerous real-world applications, like individual-based models for epidemiology \citep{rimella2022inference,cocker2023investigating}. Observe that we could redefine $d_{n,\bar{n}}$ as $d_{n,\bar{n}} \slash N$ and essentially mask the $N^{-1}$ scaling. However, we made the $N^{-1}$ dependence explicit. Our theoretical results will depend on the variability of $d_{n,\bar{n}}$, and will give $O(N^{-1})$ error bounds for models where the $d_{n,\bar{n}}$ functions have variance that is bounded as $N$ increases. 
 
\begin{assumption}\label{ass:boundeness}
    For any $n, \bar{n} \in [N]$, if $x_{t-1}, \bar{x}_{t-1} \in \mathcal{X}^N$ are such that $x_{t-1}^{\setminus \bar{n}}= \bar{x}_{t-1}^{\setminus \bar{n}}$ then there exists $ 0 < \epsilon < 1$ such that:
    \begin{equation}
    \begin{split}
        &\sum_{x^{{n}}_t} p\left (x_t^{{n}}|x_{t-1}\right ) \frac{1}{p\left (x_t^{n}|\bar{x}_{t-1}\right )^2} \leq \frac{1}{\epsilon^2}, \quad \text{and} \quad \sum_{x^{{n}}_t} p\left (x_t^{{n}}|x_{t-1}\right ) \frac{1}{p\left (x_t^{n}|\bar{x}_{t-1}\right )^3} \leq \frac{1}{\epsilon^3}.
    \end{split}
    \end{equation}
\end{assumption}

Assumption \ref{ass:boundeness} states that given two initial states $x_{t-1}, \bar{x}_{t-1}$ which differ only in their $n$th element, any transition
which is possible under $x_{t-1}$ is not too unlikely under $\bar{x}_{t-1}$. Note that this assumption is not ruling out absorbing states, but rather ensures that the non-zero elements of $p\left (x_t^{{n}}|x_{t-1}\right )$ will stay not-zero in $p\left (x_t^{n}|\bar{x}_{t-1}\right )$. 



Under assumptions \ref{ass:kernel}-\ref{ass:boundeness} and the additional assumption that the effect of an interaction on a single dimension does not exceed $N$, we can state the following theorem.

\begin{theorem}\label{thm:KL_bound}
If $\left | d_{n,\bar{n}} \left ( x^{n} \right )  - d_{n,\bar{n}} \left ( \bar{x}^{n} \right ) \right | < N$ for any $x^n,\bar{x}^n \in \mathcal{X}$ and assumptions \ref{ass:kernel}-\ref{ass:boundeness} hold, then for any $n \in [N]$:
    \begin{equation}
    \begin{split}
        &\mathbf{KL} \left [ p\left (\mathbf{y}_{[T]}^n\right )||\tilde{p}\left (\mathbf{y}_{[T]}^n\right ) \right ]\leq 
        \frac{a(\epsilon)}{N} \sum_{t \in [T]} 
        \mathbb{E} \left \{ \frac{1}{N}\sum_{\bar{n} \in [N], \bar{n} \neq n} \mathbb{V}ar \left [ d_{n,\bar{n}} \left ( \mathbf{x}_{t-1}^{n} \right ) \Big{|} \mathbf{x}_{[0:t-1]}^{\setminus n} \right ] \right \},
    \end{split}
    \end{equation}
 where $a(\epsilon) \coloneqq 2\left [ \frac{1}{2 \epsilon^2} + \frac{1}{3 \epsilon^3} \right ]$.
\end{theorem}

\begin{proof}
    The proof requires the Data Processing inequality, a Taylor expansion of the function $f\left (z\right )=\log\left (1+z\right )$, and Jensen's inequality. The full proof is available in Section A.2 of the supplementary material.
\end{proof}


From Theorem \ref{thm:KL_bound} we can observe that the approximation improves when: (i) $N$ increases; and (ii) the expected variance of the interaction term across dimensions decreases. Hence, our SimBa-CL without feedback will be more or less the same as the SimBa-CL with feedback whenever we are considering a sufficiently large $N$ and when changing the state on one dimension does not affect the other dimensions too much.

\subsection{SimBa-CL on General Partitions} \label{sec:simba_general}

Up until now, we have implicitly assumed that approximating $p\left (y_{[T]}\right )$ involves a product of marginals across all dimensions. However, it is worth considering that a complete factorization across $[N]$ might not be the optimal choice. 
For example, when considering an epidemic model with households, it may be better to factorize over households rather than individuals, due to the strong dependencies within each household.

Consider a partition $\mathcal{K}$ on $[N]$ and reformulate our likelihood approximation as follows:
\begin{equation} \label{eq:simba_general_part}
\begin{split}
    &p\left (y_{[T]}\right ) \approx \prod_{ K \in \mathcal{K}} p\left (y_{[T]}^K\right ), \quad \text{or} \quad p\left (y_{[T]}\right ) \approx \prod_{ K \in \mathcal{K}} \tilde{p}_{\mathcal{K}}\left (y_{[T]}^K\right ),
\end{split}
\end{equation}
where on the top we have the actual product of the true marginals and on the bottom $\tilde{p}_{\mathcal{K}}$ denotes the generalization of $\tilde{p}$. As for SimBa-CL with and without feedback we can reformulate our marginals and approximate marginals as a simulation from the model followed by an HMM likelihood:
\begin{equation}\label{eq:general_part_with_feed_1}
    \begin{split}
    & p\left (y_{[T]}^K\right ) \coloneqq  \sum_{x_{[0:T]}^{\setminus K}} p\left (x_{[0:T]}^{\setminus K}\right ) \sum_{x_{[0:T]}^K} p\left (x_{[0:T]}^{K}|x_{[0:T]}^{\setminus K}\right ) \prod_{t \in [T]} \prod_{n \in K} p\left (y_t^n|x_t^n\right ),  
    \end{split}
\end{equation}
which corrects according to the interactions inside $K$, and
\begin{equation}\label{eq:general_part_without_feed}
    \begin{split}
    & \tilde{p}_\mathcal{K}\left (y_{[T]}^K\right )\coloneqq \sum_{x_{[0:T]}^{\setminus K}} p\left (x_{[0:T]}^{\setminus K}\right ) \sum_{x_{[0:T]}^K} \prod_{n \in K} p\left (x_0^{n}\right ) \prod_{t \in [T]} p\left (x_t^n|x_{t-1}\right ) p\left (y_t^n|x_t^n\right ), 
    \end{split}
\end{equation}
which does not use any feedback.

As seen in Section \ref{sec:simba_feed} and in Section \ref{sec:simba_withoutfeed},  \eqref{eq:general_part_with_feed_1} aims to approximate the true marginals over $K \in \mathcal{K}$, while \eqref{eq:general_part_without_feed} only offers approximations. Once again, akin to SimBa-CL with feedback, if we want to approximate the true marginals of the likelihood we need some form of simulation feedback. This time, the simulation feedback will be from $x_{[0:t]}^{\setminus K}$ onto $x_{t-1}^K$, as we are considering probability distributions on $K$. 



\begin{table*}[httb!]
    \centering
    \begin{tabular}{llll}
    \hline
     Method & Computational cost\\
    \hline
    Simulation & $\mathcal{O}(N \mathcal{C}_{\mathcal{M}}(1))$\\
    \hline
    SMC & $\mathcal{O}(N \mathcal{C}_{\mathcal{M}}(P) +  N \mathcal{C}_q(P))$\\
    \hline
    SimBa-CL WF & $\mathcal{O}(N \mathcal{C}_{\mathcal{M}}(P) + N^2 P \mathbf{card}(\mathcal{X})^2)$\\
    SimBa-CL WOF & $\mathcal{O}(N \mathcal{C}_{\mathcal{M}}(P) + N P \mathbf{card}(\mathcal{X})^2)$\\
    SimBa-CL WF on $\mathcal{K}$ & $\mathcal{O}(N \mathcal{C}_{\mathcal{M}}(P) + \mathbf{card}(\mathcal{K}) (N-K_{max}) P \mathbf{card}(\mathcal{X})^{2 K_{max}})$\\
    SimBa-CL WOF on $\mathcal{K}$ & $\mathcal{O}(N \mathcal{C}_{\mathcal{M}}(P) + \mathbf{card}(\mathcal{K}) P \mathbf{card}(\mathcal{X})^{2 K_{max}})$\\
    \hline
    \end{tabular}
    \caption{Computational cost for a single time step.  We denote with $\mathcal{C}_{\mathcal{M}}(1)$ the cost of simulating from a single dimension of the model. For the SMC, $\mathcal{C}_q(1)$ is the cost of computing the proposal distribution $q$ (see \cite{rimella2022approximating} for an example). For SimBa-CL, we use ``WF'' to indicate ``with feedback'', ``WOF'' to indicate ``without feedback'', and $K_{max} \coloneqq \max_{K \in \mathcal{K}} \mathbf{card}(K)$.}
    \label{table:comp_cost}
\end{table*}

We can then easily adapt Algorithm \ref{alg:simulation_likelihood_with_feedback} and Algorithm \ref{alg:simulation_likelihood_no_feedback}, transitioning from a full factorization to a factorization on the partition $\mathcal{K}$. The algorithm requires operations on the space $\mathcal{X}^K$ and so a computational cost that is exponential in the maximum number of components contained in $K$. 
More details and theoretical results are available in Section A.3 of the supplementary material.


We refer to Table \ref{table:comp_cost} for a comparison of computational costs for a single time step. The final computational cost of each algorithm is then obtained by multiplying the values in Table \ref{table:comp_cost} by $T$. SMC is used as the computational baseline. We observe that SMC is more expensive to run compared to SimBa-CL whenever the time required to compute the proposal distribution exceeds the time required to run the recursion (with or without feedback) in SimBa-CL. When considering SimBa-CL on partitions, the computational cost becomes linear in the number of partition elements, but exponential in the size of the largest element of the partition (in the worst-case scenario). This is because the dimensions within each element of the partition are now considered jointly. When including the feedback in SimBa-CL on partitions, we must iterate over all the elements external to the partition, i.e. $[N]\setminus K$,  requiring $N-\max_{K \in \mathcal{K}} \mathbf{card}(K)$ (in the worst-case scenario).


\section{Inference and Confidence Sets} \label{sec:var_sens_matrix}


Composite likelihoods are generally obtained as the product of likelihood components, whose structure is dependent on the considered model, see \cite{varin2011overview} for a review of composite likelihood methods. 
Both SimBa-CL with and without feedback calculate a composite likelihood as the product of the marginals of the likelihood. We can use asymptotic theory for composite likelihood to get approximate confidence regions around the maximum composite likelihood estimator. 

The first step is to find the maximum composite likelihood estimator $\hat{\theta}_{CL}$ by maximizing the composite likelihood $\mathcal{L}_{CL}\left (\theta;y_{[T]}\right ) = \prod_{K \in \mathcal{K}} \mathcal{L}^K_{CL}\left (\theta;y^K_{[T]}\right )$ or, equivalently, the composite log-likelihood $\ell_{CL}\left (\theta;y_{[T]}\right ) = \sum_{K \in \mathcal{K}}\ell^K_{CL}\left (\theta;y^K_{[T]}\right )$, where $\mathcal{L}^K_{CL}\left (\theta;y^K_{[T]}\right )$ and $\ell^K_{CL}\left (\theta;y^K_{[T]}\right )$ depends on the considered SimBa-CL.

The second step is to notice that in the composite likelihood literature, we have some form of asymptotic normality for $\hat{\theta}_{CL}$ \citep{varin2008composite}:
\begin{equation}\label{eq:approx_gaussian}
    \hat{\theta}_{CL} \overset{d}{\approx}  \mathcal{N} \left ( \theta, G\left (\theta\right ) ^{-1} \right ),
\end{equation}
where $G\left (\theta\right )$ is the Godambe information matrix \citep{godambe1960optimum}. It then follows that, upon estimating the Godambe information matrix, we can build confidence sets for $\hat{\theta}_{CL}$ as multidimensional ellipsoids.

The final step is to estimate the Godambe information matrix, which is given by $G\left (\theta\right ) = S\left (\theta\right )  V\left (\theta\right )^{-1} S\left (\theta\right )$, decomposed in terms of the sensitivity matrix and the variability matrix:
\begin{equation}
\begin{split}
    &S\left (\theta\right ) = \mathbb{E}_{\theta} \left \{ -\Hessian_{\theta} \left [ \ell_{CL}\left (\theta;\mathbf{y}_{[T]}\right ) \right ] \right \} \quad \text{and} \quad V\left (\theta\right ) = \mathbb{V}\text{ar}_{\theta} \left \{ \nabla_{\theta} \left [ \ell_{CL}\left (\theta;\mathbf{y}_{[T]}\right ) \right ] \right \},
\end{split}
\end{equation}
where $\Hessian_{\theta}$ and $\nabla_{\theta}$ are the Hessian and the gradient with respect to $\theta$. Given that we can compute the Hessian and the gradient given $\theta,{y}_{[T]}$, we can also estimate expectation and the variance via simulations from the model (expected information). More details on this aspect of SimBa-CL along with approximating $S\left (\theta\right )$ and $V\left (\theta\right )$ with the actual observations (observed information) are discussed in Section B of the supplementary material, with also more experiments available in Section C. 


\subsection{Bartlett Identities} \label{sec:Bartlett_identities}

The computation of the Hessian can be resource-intensive and variance estimation might be noisy. We can then simplify the form of the sensitivity and variability matrix by invoking the first and second Bartlett identities. When considering SimBa-CL with feedback the identities hold asymptotically in the number of Monte Carlo samples $P$, while for SimBa-CL without feedback they hold only approximately. The sensitivity matrix and variability matrix can be then reformulated as:
\begin{equation}
\begin{split}
    &S\left (\theta\right )= 
    \sum_{K \in \mathcal{K}} \mathbb{E}_{\theta} \left [ \nabla_{\theta} \ell^K_{CL}\left (\theta;y^K_{[T]}\right ) \nabla_{\theta} \ell^K_{CL}\left (\theta;y^K_{[T]}\right )^\top \right ],\\
    &V\left (\theta\right )= \sum_{K, \tilde{K} \in \mathcal{K}} \mathbb{E}_{\theta} \left [ \nabla_{\theta} 
     \ell^K_{CL}\left (\theta;y^K_{[T]}\right ) \nabla_{\theta}
     \ell^{\tilde{K}}_{CL}\left (\theta;y^{\tilde{K}}_{[T]}\right )^\top \right ],
\end{split}
\end{equation}
where $=$ becomes $\approx$ if we do not include the feedback, and where both matrices can be once again estimated via simulation. More details are available in Section B.3 of the supplementary material.


\section{Limitations and extensions} \label{sec:limit_ext}

It is evident from Table \ref{table:comp_cost} that considering SimBa-CL on a partition $\mathcal{K}$ results in a computational cost that is linear in the number of elements in the partition but exponential in the size of the elements within the partition, which quickly becomes infeasible. As shown in Section \ref{sec:KL_div_exp} and Section \ref{sec:like_surfaces_exp}, choosing a coarser partition does not affect performance when considering a model with dimensions that exhibit low correlation. This is the case, for instance, when using an individual-based model in epidemiology, where individuals have similar behaviours, e.g. they are homogeneous mixing \citep{ju2021sequential,rimella2022approximating}. However, it is common for dimensions to cluster together in terms of correlation. For example, again in an individual-based model for epidemiology, individuals are often assigned to households and have a high likelihood of being infected by others within the same household \citep{rimella2022inference}. In such scenarios, the SimBa-CL approach described in Section \ref{sec:simba_general} with $\mathcal{K}$ representing the partition induced by the households is preferable, but comes with an increased computational cost. As a solution, one could further partition households into smaller subsets and implement localized feedback to adjust accordingly. This hybrid SimBa-CL would eliminate feedback across elements of the partition while retaining feedback within the elements of the partition, which could potentially correct the approximation.

The factorial structure shown in \eqref{eq:model_factorization} restricts SimBa-CL to discrete-time models with possibly unevenly timed observations. If we are interested in continuous-time models and the instantaneous dynamics are independent across dimensions given the current state, an Euler approximation of the continuous-time dynamics can be used to construct the discrete-time model. However, if the time between observations is large, such an Euler approximation may not be accurate. In such cases, we can use a sufficiently small discrete-time step to ensure the accuracy of the Euler approximation, with the additional time steps corresponding to points where no observations are available. This approach results in several prediction steps without correction, followed by a correction step when data are observed.

Unconditional simulation from the model has the benefit of making SimBa-CL almost as computationally efficient as simulating directly from the model. Indeed, after the simulation phase, only simple marginalizations on $\mathcal{X}$ are required, which can be even parallelized across the $N$ dimensions. However, the simulation procedure does not incorporate information from the data, which makes it possible for SimBa-CL to generate realizations of $\mathbf{x}_{[0:t]}$ that are inconsistent with the data and might affect its performance via the transition kernel. Fortunately, Assumption \ref{ass:boundeness} guarantees that all the non-zero elements of the transition kernel remain non-zero, ruling out the possibility of SimBa-CL collapsing to zero likelihood estimates. That said, mismatches can still occur if there is significant uncertainty about the trajectory of the latent process. However, increasing the number of simulations $P$ reduces the probability of such mismatches. Note that we do not imply that there are no advantages to incorporating data information during the simulation process.  On the contrary, the closer the simulations are to the actual hidden trajectories, the better the performance \citep{whiteley2014,rimella2022approximating}. \\
One option, for example, could be to simulate from ${p}\left (x^n_{t}|y^n_{[t]}, x_{[0:t]}^{\setminus n}, \theta \right )$ for SimBa-CL with feedback, or from $\tilde{p}\left (x^n_{t}|y^n_{[t]}, x_{[0:t-1]}^{\setminus n}, \theta \right )$ for SimBa-CL without feedback. Including the observations should make the simulations more consistent with the data, possibly reducing both the bias and the variance of SimBa-CL. From a theoretical perspective, this could even result in a better bound compared to the one from Theorem \ref{thm:KL_bound}, which currently increases linearly with time. We provide a small experiment on this ``conditional'' SimBa-CL in Section C.6 of the supplementary material. \\
Alternatively, one could even attempt to avoid simulations from the model altogether by performing the prediction step with the modified transition kernel $\prod_{n \in [N]} p\left (x^n_t|x_{t-1}^n, \mathbb{E}_{\rho^{\setminus n}}\left [\mathbf{x}_{t-1}^{\setminus n} \right ], \theta \right )$, i.e. substituting the contribution from ${x}_{t-1}^{\setminus n}$ with its expected contribution under a distribution $\rho^{\setminus n}$, which could be either unconditional or conditional on the data. This approach would still exploit the factorial structure \eqref{eq:model_factorization}, and the computational cost could be further reduced whenever closed-form solutions are available \citep{rimella2025CAL}.

It can be observed from Section \ref{sec:simba_feed} and Section \ref{sec:simba_withoutfeed} that both SimBa-CL with and without feedback require simple operations that combine the initial distribution, transition kernel, and emission distribution. If these three quantities are continuous with respect to the parameters $\theta$, then, given the Monte Carlo simulations, it is possible to compute gradients via automatic differentiation. This process is readily implemented in multiple Machine Learning libraries and can be combined with any gradient descent technique to optimize the parameters \citep{zeiler2012adadelta,kingma2014adam}. However, the Monte Carlo simulation itself depends on the parameters $\theta$, which complicates the differentiation process. To be more precise, we consider quantities of the form $\sum_{n \in [N]} \log \left ( p(y_{[T]}^n|\theta)\right )$, with $\tilde{p}$ for SimBa-CL without feedback, meaning that we need to compute:
\begin{equation}\label{eq:diff_exp}
    \begin{split}
        &\frac{\partial p(y_{[T]}^n|\theta)}{\partial \theta}
        = \frac{\partial}{\partial \theta} \mathbb{E}_{\mathcal{M}} \left [ p\left (y_{[T]}^n|\mathbf{x}_{[0:T-1]}^{\setminus n},\theta\right ) \right ],
    \end{split}
\end{equation}
where $\mathcal{M}$ stands for simulating from the model. In our experiments we consider the approximation:
\begin{equation}
    \begin{split}
        \frac{\partial p(y_{[T]}^n|\theta)}{\partial \theta}
        &\approx \mathbb{E}_{\mathcal{M}} \left [ \frac{\partial}{\partial \theta} p\left (y_{[T]}^n|\mathbf{x}_{[0:T-1]}^{\setminus n},\theta\right ) \right ].
    \end{split}
\end{equation}
which computes derivatives by conditioning on the results of the Monte Carlo simulation, and essentially approximates the derivative of an expectation with the expectation of a derivative. On the one hand, this approximation reduces computational complexity by ignoring part of the derivative. On the other hand, it produces biased estimates of the gradient.\\
As we require simulation from a finite state-space and therefore simulation from categorical random variables, we could alternatively use a continuous relaxation of categorical random variables known as the Gumbel-Softmax trick \citep{maddison2016concrete,jang2016categorical}. This technique substitutes the $\arg\max$ of i.i.d. Gumbel random variables \citep{gumbel1954statistical} with a Softmax function. In this way, sampling from a categorical distribution is expressed as a continuous function of i.i.d. random variables and suitable to the reparameterization trick \citep{blundell15}, which allows to exchange expectations with derivatives without any approximation. In our case we get:
\begin{equation}
    \begin{split}
        \frac{\partial p(y_{[T]}^n|\theta)}{\partial \theta}
        &\approx
        \mathbb{E}_{\mathcal{G}} \left [ \frac{\partial }{\partial \theta}p\left (y_{[T]}^n|s(\mathbf{G},\theta,\tau),\theta\right ) \right ],
    \end{split}
\end{equation}
where $\mathbf{G}$ is a collection of i.i.d. Gumbel random variables with distribution $\mathcal{G}$, and $s$ is a composition of continuous functions involving the Gumbel random variables, the parameter $\theta$, and the temperature parameter $\tau$.\\
Even though this approximation produces unbiased gradient estimates as $\tau \to 0$, it is computationally heavier and challenging to use when considering large $P$ and when computing confidence regions. We refer to Section C.6 of the supplementary material for a detailed discussion and additional experiments.


\section{Experiments}

The experiments centre on the field of epidemiological modelling, and specifically they focus on individual-based models (IBMs). Individual-based models arise when we want to describe an epidemic from an individual perspective. The complexity of IBMs lies in their high-dimensional state-space, making closed-form likelihood computationally unfeasible. However, these models satisfy the factorization outlined in \eqref{eq:model_factorization}, making them the perfect candidates for our SimBa-CL.

SimBa-CL is implemented in Python using the library TensorFlow and is available at the GitHub repository:\\ https://github.com/LorenzoRimella/SimBa-CL. All the experiments were run on a 32gb Tesla V100 GPU available on ``The High End Computing'' (HEC) facility at Lancaster University. For gradient computation we use TensorFlow's default automatic differentiation, see Section \ref{sec:limit_ext} for a deeper discussion.

In Section C of the supplementary material, interested readers can find more details about the experiments reported in the paper, as well as additional experiments on the following topics: a spatial SIS individual-based model; a comparison between SimBa-CL and its extensions; the role of $P$ in terms of the quality of the approximation.

\begin{table*}
    \centering
    \begin{tabular}{ccccccc}
    \hline
     $N$ & $1000$ &  $100$  & $100$  & $100$   & $100$  & $10$ \\ 
     $\beta_0$ & base & high & low  &  base & base & base \\             
     $\iota$ & base & base &  base &  base & low  & base \\        
     $\mathbf{KL}$ on feedback  & 0.3 (0.2) & 7.5 (3.2) & 0.8 (0.4) & 5.9 (2.9) & 1.4 (1) & 49.6 (17.2) \\ 
     $\mathbf{KL}$ on partition  & 1.4 (0.3) & 36.1 (6.6) & 3 (0.9) & 35.1 (7) & 5.6 (2) & 177.3 (36.5) \\
    \hline
    \end{tabular}
    \caption{Comparing empirical KL between SimBa-CL and under different scenarios. All the numerical values have been multiplied by $10^9$ to improve visualization.}
    \label{tab:KL_comparison}
\end{table*}

\subsection{The Effect of Feedback and Partition's Choice on SimBa-CL} \label{sec:comparison_simba}

In this section we perform a cross-comparison on a susceptible-infected-susceptible (SIS) IBM on SimBa-CL with feedback on $\mathcal{K}=\{\{1\},\dots,\{N\}\}$ (``fully factorized SimBa-CL with feedback''), SimBa-CL without feedback on $\mathcal{K}=\{\{1\},\dots,\{N\}\}$ (``fully factorized SimBa-CL without feedback''), and SimBa-CL without feedback on $\mathcal{K}=\{\{1,2\},\dots,\{N-1,N\}\}$ (``coupled SimBa-CL without feedback''), where $N$ is assumed to be even. 

\subsubsection{Model} \label{sec:model_SIS}

Building upon the framework introduced by \cite{ju2021sequential} and \cite{rimella2022approximating}, where $n$ represents an individual and $w_n$ is an observed vector of covariates. We consider a $p\left (x_0^n|\theta\right )$ with an initial probability of infection of ${1}\slash {1+\exp{\left (-\beta_0^\top w_n\right )}}$ and a transition kernel $p\left (x_t^n|x_{t-1}, \theta\right )$ with a probability of transitioning from S to I of \\ $1- \exp \left [ { - \lambda_n \left ( {\sum_{n \in [N]} \mathbb{I}\left (x_{t-1}^n=2\right ) }\slash {N} + \iota \right ) } \right ]$ and a probability of transitioning from I to S of $1-\exp \left (- \gamma_n \right)$, where $\lambda_n = {1}\slash{(1+\exp{\left (-\beta_{\lambda}^\top w_n\right )})}$ and $\gamma_n = {1} \slash {(1+\exp{\left (-\beta_{\gamma }^\top w_n\right )})}$ and with $w_n, \beta_0,\beta_\lambda,\beta_\gamma \in \mathbb{R}^2$. Moreover we consider $p\left (y_t^n|x_{t}^n, \theta\right ) = q^{x_t^n} \mathbb{I}\left (y_t^n \neq 0\right ) +\left (1- q^{x_t^n}\right ) \mathbb{I}\left (y_t^n = 0\right )$, with $q \in [0,1]^2$.
Unless specify otherwise, our baseline model employs $N=1000$, $T= 100$, $w_n$ to be such that $w_n^1 = 1$ and $w_n^2 \sim \mathbf{Normal}\left (0,1\right )$, and the data-generating parameters $\beta_0 = [-\log\left (\left (1 \slash 0.01\right )-1\right ), 0]^\top$, $\beta_{\lambda} = [-1, 2]^\top$, $\beta_{\gamma} = [-1, -1]^\top$, $q = [0.6, 0.4]^\top$ and $\iota=0.001$. It is also important to mention that the considered SIS satisfies Assumption \ref{ass:kernel} and Assumption \ref{ass:boundeness}, see Section A.1 of the supplementary material.

\begin{figure*}
    \centering
    \includegraphics[width=\textwidth]{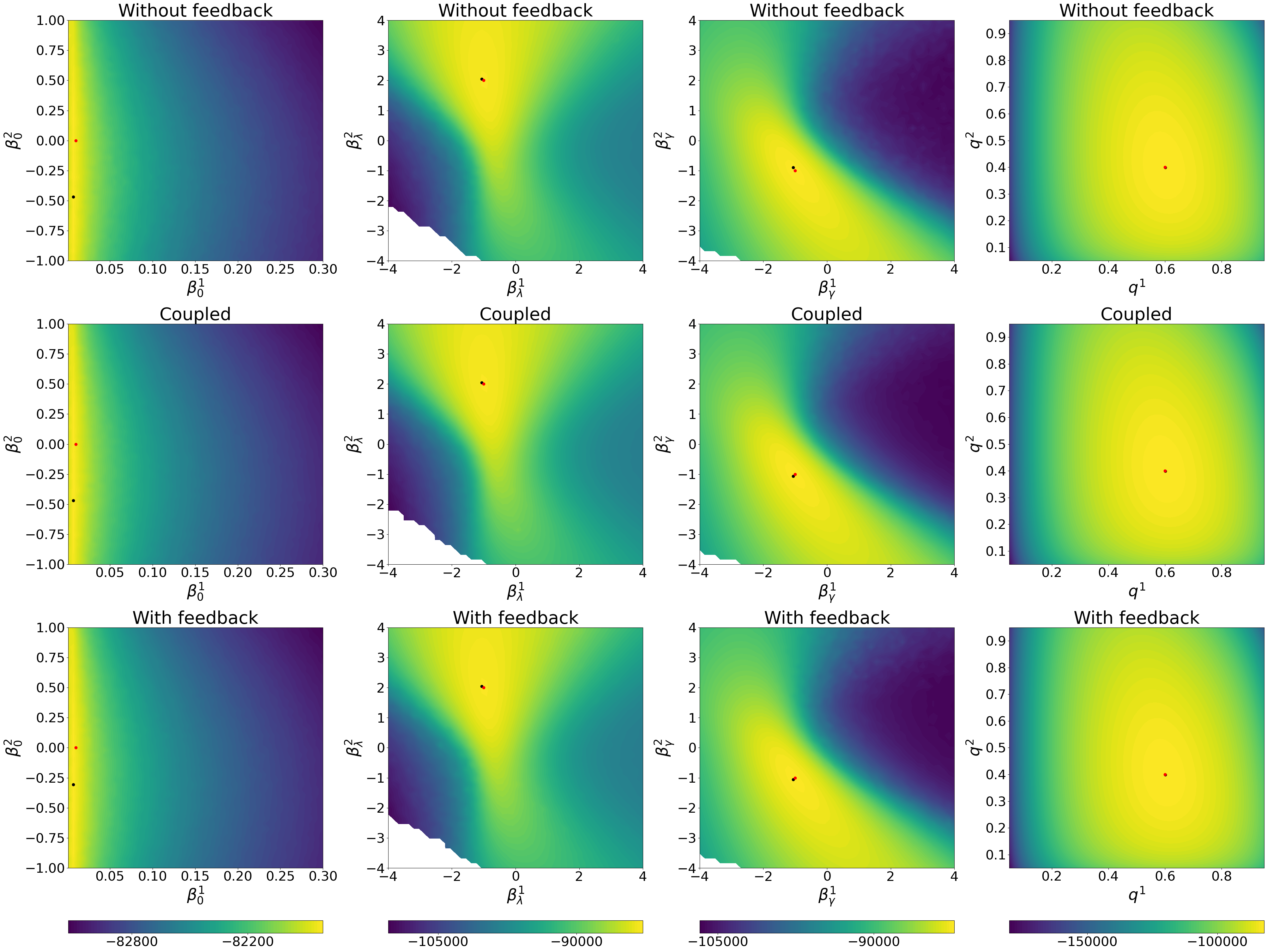}
    \caption{Profile log-likelihood surfaces for $\beta_0,\beta_\lambda,\beta_\gamma,q$. From top to bottom, fully factorized SimBa-CL without feedback, coupled SimBa-CL without feedback, and fully factorized SimBa-CL with feedback. Dots locate the data-generating parameter and the maximum on the grid. }
    \label{fig:log_like_surfaces}
\end{figure*}

\subsubsection{Empirical evaluation of the Kullback-Leibler divergence} \label{sec:KL_div_exp}

We start by comparing our SimBa-CL methods in terms of empirical Kullback-Leibler divergence (KL) on a set of simulated data. Different settings are considered: an increasing population size $N=[10,100,1000]$; either ``high'' $\beta_0 = [-6.9, 0]^\top$ or ``base'' $\beta_0 = [-4.60, 0]^\top$ or ``low'' $\beta_0 = [-2.20, 0]^\top$, i.e. around either $0.1\%$ or $1\%$ or $10\%$ of initial infected; and ``base'' and ``low'' $\iota=[0.001,0.01]$. Note that different $\beta_0$ and $\iota$ control the variance of the process, as having more infected at the beginning of the epidemic or including more environmental effects results in an epidemic that is closer to the equilibrium. We set $P=1024$ for all scenarios and for both SimBa-CL with feedback, SimBa-CL without feedback and coupled SimBa-CL without feedback. All the versions of SimBa-CL were run under the data-generating parameter. 


Table \ref{tab:KL_comparison} reports the mean and standard deviations of the empirical KL per each scenario. Focusing on the fourth row, which contrasts the fully factorized SimBa-CL with and without feedback, we can notice that: increasing $N$ decreases the KL, and decreasing the variance decreases the KL. Similar conclusions can be drawn for the fifth row, which compares the fully factorized SimBa-CL without feedback with the coupled SimBa-CL without feedback. These comments suggest less and less differences across the methods when increasing $N$ and decreasing the variance, which is in line with our theoretical results. 

\subsubsection{Comparing likelihood surfaces} \label{sec:like_surfaces_exp}

We proceed to undertake a comparison of profile likelihood surfaces for the baseline SIS IBM using the following protocol: (i) choose one among $\beta_0$, $\beta_{\lambda}$, $\beta_{\gamma}$ and $q$; (ii) simulate using the baseline model, and ensure at least 10 infected in the epidemic realization; (iii) create a bi-dimensional grid on the chosen parameter; (iv) per each element of the grid compute our SimBa-CL methods by fixing the other parameters to their real values. As for the previous section, we set $P=1024$ when computing all surfaces and for both SimBa-CL with feedback, SimBa-CL without feedback and coupled SimBa-CL without feedback. 

The outcomes of this experiment are illustrated in Figure \ref{fig:log_like_surfaces}. Interestingly, all the considered SimBa-CL exhibit a consistent shape, meaning that, in the SIS scenario, including the feedback or choosing a coarser partition has a limited impact on the overall likelihood. Furthermore, it becomes evident that all these methods effectively recover the data-generating parameter, except for $\beta_0$. Here there is an obvious identifiability issue with $\beta_0$ as $\beta_0^2=0$ implies covariate $w_n^2$ to not be used. However, given that $w_n^2$ is random, we could have more initial infection associated with $w_n^2<0$, which gives a higher likelihood to models where $\beta_0^2<0$. Similar reasoning can be replicated for $w_n^2>0$, which explains the symmetry of the likelihood surface of $\beta_0$ on the vertical axis.

\begin{figure*}
    \centering
    \includegraphics[width=0.8\textwidth]{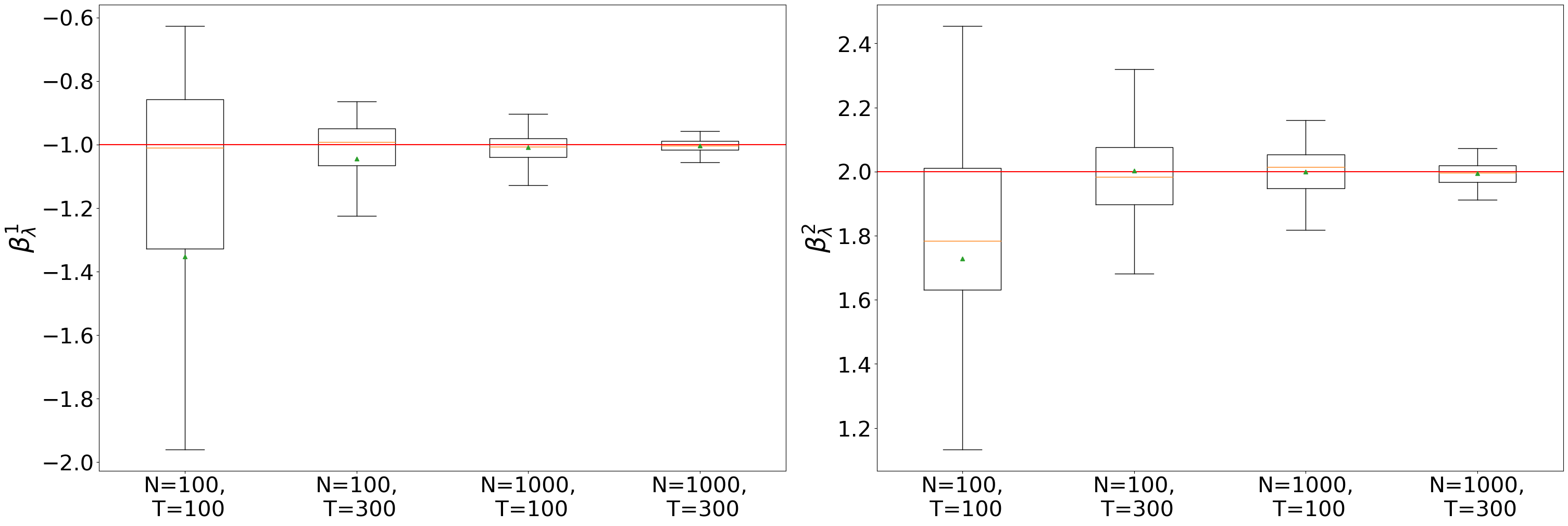}
    \caption{Box plots on the optimized $\beta_\lambda$. On the left, $\beta_\lambda^1$, on the right, $\beta_\lambda^2$. Horizontal solid lines within the boxes show the medians and triangles are used for the means. Horizontal solid lines show the true parameters.}
    \label{fig:bi_dim_beta_lambda_boxplot}
\end{figure*}

\subsection{Asymptotic Properties of SimBa-CL} \label{sec:coverage}

We can now turn to the problem of computing the maximum composite likelihood estimator and the corresponding confidence sets. Remark that for a sufficiently ``regular'' model and a sufficiently large $N$, including the simulation feedback and using a coarser partition bears marginal significance for SimBa-CL. We hence narrow our studies to the asymptotic properties of the fully factorized SimBa-CL without feedback. As a toy model, we consider again the SIS IBM described in Section \ref{sec:model_SIS}.

Crucially, it is worth emphasizing once more that gradients and Hessians can be computed through automatic differentiation as we have implemented SimBa-CL using TensorFlow, see Section \ref{sec:limit_ext}. This computational capability grants us access to the estimates of the variability and sensitivity matrix, as from Section \ref{sec:var_sens_matrix}, thereby avoiding any manual derivation of derivatives.

Remark that in both Section \ref{subsec:2dim_optim} and Section \ref{subsec:9dim_optim}, we set $P=500$ for parameter optimization and $P=100$ for computing confidence regions. These values of $P$ were selected to avoid out-of-memory errors.

\subsubsection{Maximum Simba-CL convergence and coverage in two dimensions} \label{subsec:2dim_optim}

We start our exploration by looking at the bi-dimensional parameter $\beta_\lambda$, given all the other parameters fixed to their baseline values. The hope is that $\beta_\lambda$ is easily identifiable as it highly influences the evolution of the epidemics, and so we can test the asymptotic properties on a well-behaved parameter. 

To explore the asymptotics of SimBa-CL we investigate four scenarios with an increasing amount of data: (i) $N=100, T=100$; (ii) $N=100, T=300$; (iii) $N=1000, T=100$; (iv) $N=1000, T=300$. Per each scenario, we simulate $100$ epidemics and per each dataset we optimize $\beta_\lambda$ through Adam optimization \citep{kingma2014adam}, aiming to minimize the negative log-likelihood. After optimization, we have a sample of $100$ bi-dimensional parameters per each scenario, which can be turned into box plots as shown in Figure \ref{fig:bi_dim_beta_lambda_boxplot}. As both $T$ and $N$ increase, we can observe an evident shrinkage towards the true parameter, suggesting consistency of the maximum SimBa-CL estimator when $N$ and $T$ increase. 

\begin{figure*}
    \centering
    \includegraphics[width= \textwidth]{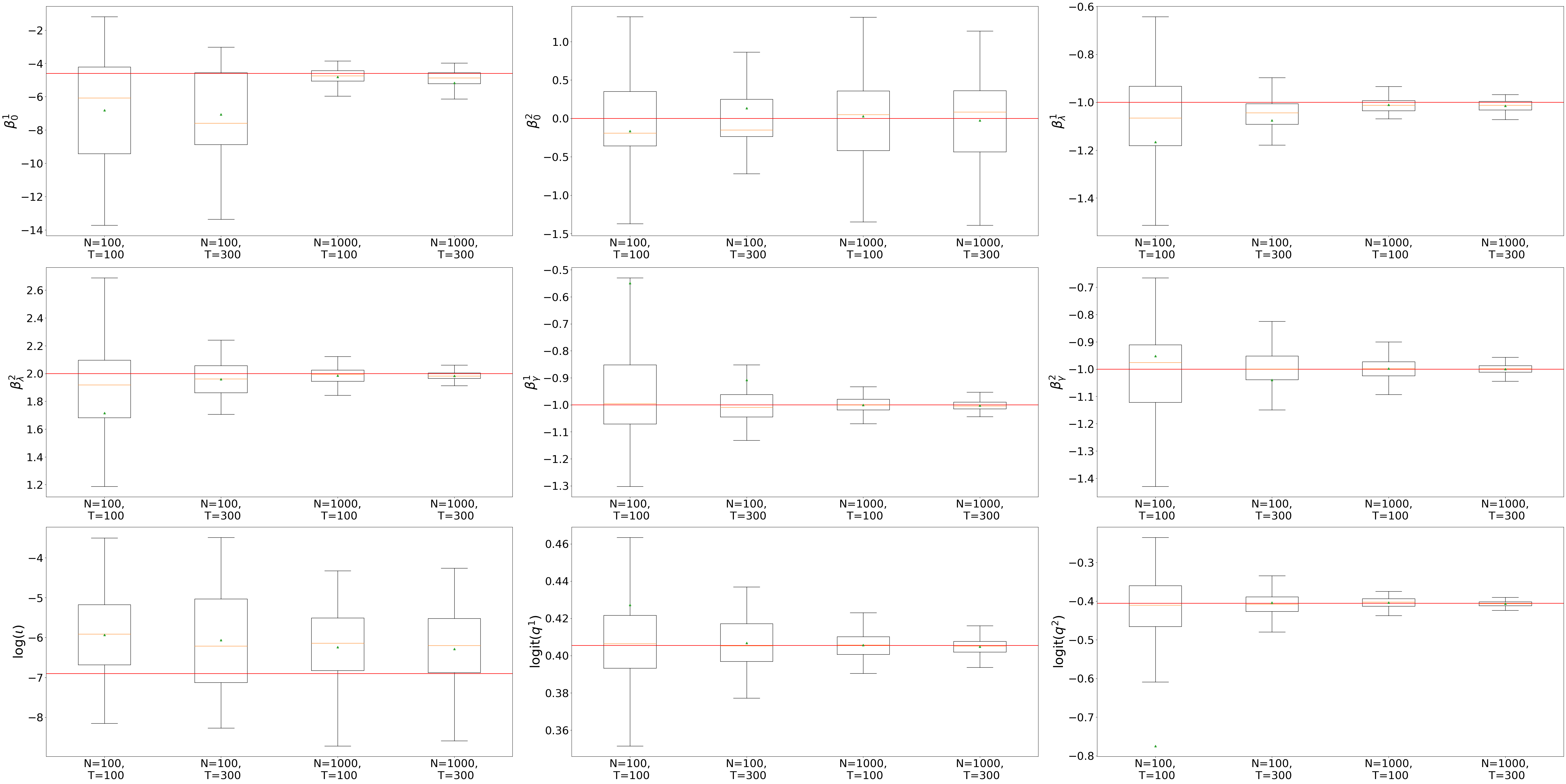}
    \caption{Box plots on the optimized $\theta = \left (\beta_0, \beta_\lambda, \beta_\gamma, q, \iota\right )$. Parameter labels are reported on the y-axis. Horizontal solid lines within the boxes show the medians and triangles are used for the means. Horizontal solid lines show the true parameters.}
    \label{fig:9dim_asympt}
\end{figure*}

\begin{table*}
    \centering
    \begin{tabular}{lccccc}
     Parameter & $\beta_0$ & $\beta_\lambda$ & $\beta_\gamma$ & $q$ & $\iota$  \\
    \hline
    Without Bartlett & $0.17 \text{ and }0.05$ & $0.61 \text{ and }0.87$ & $0.8 \text{ and }1.$ & $0.87 \text{ and } 0.5$ & $0.02$\\ 
    With Bartlett & $0.98 \text{ and } 0.89$ & $0.99 \text{ and } 0.75$ & $0.97 \text{ and } 0.97$ & $1. \text{ and } 0.98$ & $0.92$\\
    \hline
    \end{tabular}
    \caption{Empirical coverage per each parameter when computing the Godambe information matrix with and without the approximate Bartlett identities. Whenever the parameter is bi-dimensional the coverage per each component is reported in the same cell separated by ``$\text{and}$''.}
    \label{tab:9dim_coverage}
\end{table*}

Taking the investigation a step further, we analyse the empirical coverage of confidence sets built as explained in Section \ref{sec:var_sens_matrix}. We consider $N=1000, T=300$, and the optimized parameters from the previous experiment. We start by calculating the Godambe information matrix without using the Bartlett identities, and we build $95\%$ 2-dimensional confidence sets for our parameter $\beta_\lambda$. The procedure results in a coverage of $1$ and so an overestimation of uncertainty. However, when repeating the same procedure using the Bartlett identities, the coverage now aligns with the theoretical coverage of $0.95$. This favourable outcome can be attributed to less noisy estimates, as we are exploiting the factorization in the model and so computing expectations on lower dimensional spaces, see Section C.3 of the supplementary materials.

\subsubsection{Maximum Simba-CL convergence and coverage in nine dimensions}\label{subsec:9dim_optim}

Transitioning to a substantially more intricate scenario, we now estimate all the parameters of the model $\theta = \left (\beta_0, \beta_\lambda, \beta_\gamma, q, \iota\right )$. Analogous to the $2$-dimensional case, we simulate from the model $100$ times, and per each simulation, we optimize the parameters using Adam optimizer. The outcomes are reported in Figure \ref{fig:9dim_asympt}.

Foremost, it becomes apparent that increasing the value of $T$ does not influence $\beta_0$. This arises because observations in the later periods carry scarce information on the initial condition. Moreover, recall that $w_n^2 \sim \mathbf{Normal}\left (0, 1\right )$ and $\beta_0^2=0$. This makes the parameter not identifiable, as commented in Section \ref{sec:comparison_simba}. This ill-posed model definition implies that increasing $N$ will not improve the uncertainty around our estimate as $\beta_0^2<0$ and $\beta_0^2>0$ can be equally likely, while the unbiasedness is preserved due to the symmetry of the set of equally likely parameters. At the same time, the parameters $\iota$ and $\beta_\lambda$ are also hard to identify as smaller (or bigger) estimates of $\beta_\lambda$ will lead to bigger (or smaller) estimates of $\iota$. Indeed, $\beta_\lambda$ governs the infection rates from the community, while $\iota$ represents the environmental effect. It is then clear that generating an epidemic from $\beta_\lambda,\iota$ is equivalent to generating one by decreasing $\beta_\lambda$ and increasing accordingly $\iota$. This correlation is especially vivid in Figure \ref{fig:9dim_asympt}, as an overestimation of $\iota$ (log-scale) leads to an underestimation of $\beta_\lambda$.

Clearly, in a $9$-dimensional scenario, the process of recovering empirically the theoretical coverage is substantially more complicated. Jointly, we find $0$ coverage of the $9$-dimensional $95\%$ confidence sets, irrespective of whether the Bartlett identities are employed or not. We then compute the confidence intervals on single parameters by marginalizing the $9$-dimensional Gaussian distribution. Marginally, we find that using the approximate Bartlett identities improves the coverages, see Table \ref{tab:9dim_coverage} for numerical values, with comments similar to the previous section for the quality of the estimates. 

\subsection{Comparing SimBa-CL with Sequential Monte Carlo} \label{sec:comparison_SMC}

As SimBa-CL methods provide biased estimates of the likelihood, the objective of this section is to compare SimBa-CL methods with both sequential Monte Carlo (SMC) and block sequential Monte Carlo (BSMC) algorithms. While SMC provides an unbiased particle estimate of the likelihood \citep{Chopin2020}, this quantity can suffer high variance in high-dimensional scenarios. On the other hand, BSMC deals with the curse of dimensionality by providing a factorized, albeit biased, particle estimate of the likelihood \citep{rebeschini2015can}. 

Building upon the previous sections, we compare SMC and BSMC with fully factorized SimBa-CL without feedback. Regarding the SMC comparison, we consider two approaches: the auxiliary particle filter (APF) and the SMC with proposal distribution given by \cite{rimella2022approximating}. Due to the curse of dimensionality, we expect poor performances of the APF, hence we include the Block APF in our analysis. The block APF works as the Block particle filter \citep{rebeschini2015can}, a BSMC algorithm, but it proposes particles according to the transition kernel informed by the current observation. 

\subsubsection{Model}

In the subsequent sections, we work again on the SIS IBM from Section \ref{sec:model_SIS}. However, we also analyse a susceptible-exposed-infected-removed (SEIR) IBM. Specifically, we still have some bi-dimensional covariates $w_n$, while $p(x_0^n|\theta)$ is now $4$-dimensional with the second and the fourth components being zero and the first and the third components being as in SIS IBM. Similarly, the transition kernel $p(x_t^n|x_{t-1}, \theta)$ is a $4$ by $4$ matrix with the same dynamics of SIS IBM when considering transitions from S to E and from I to R, with the addition of a transition E to I with probability $1-\exp(-\rho)$. Unless specified otherwise, we consider as the baseline model the one with $N=1000$, $T= 100$ and the data-generating parameters set to $\beta_0 = [-\log\left (\left (1 \slash 0.01\right )-1\right ), 0]^\top$, $\beta_{\lambda} = [-1, 2]^\top$, $\rho=0.2$, $\beta_{\gamma} = [-1, -1]^\top$, $q = [0,0,0.6, 0.4]^\top$. Observe, that the first two elements of $q$ are both zero, meaning that we do not observe any susceptible and exposed.

\subsubsection{SimBa-CL and SMC for an individual-based SIS model}

We consider the baseline SIS model with $N=1000$, and precisely: we generate the data; we run SimBa-CL and the baseline algorithms $100$ times on the given data using the data-generating parameters; and we estimate the mean and standard deviation of the log-likelihood. The results are reported in Table \ref{tab:N1000_SIS}. 

\begin{table*}[httb!]
    \centering
    \begin{tabular}{lcccc}
     P                       & 512               & 1024              & 2048              & Time (sec)\\
    \hline
    APF              & -81103.37 (46.04) & -81046.57 (49.65) & -80976.34 (36.08) & 1.05s \\
    h=5              & -79551.92 (1.79)  & -79552.24 (1.6)   & -79552.81 (1.57)  & 3.78s \\
    h=10             & -79551.9 (1.81)   & -79552.22 (1.47)  & -79553.01 (1.56)  & 5.61s \\
    Block APF        & -79565.69 (5.84)  & -79558.44 (3.95)  & Out of memory     & 2.97s \\
    SimBa-CL           & -79612.74 (3.4)   & -79612.31 (2.37)  & -79612.34 (1.55)  & 1.03s \\
    \hline
    \end{tabular}
    \caption{Log-likelihood means and log-likelihood standard deviations for the baseline SIS model with $N=1000$. $h$ is the number of future observations included in \cite{rimella2022approximating} ($h=0$ correspond to APF).}
    \label{tab:N1000_SIS}
\end{table*}

\begin{figure*}[httb!]
    \centering
    \includegraphics[width=\textwidth]{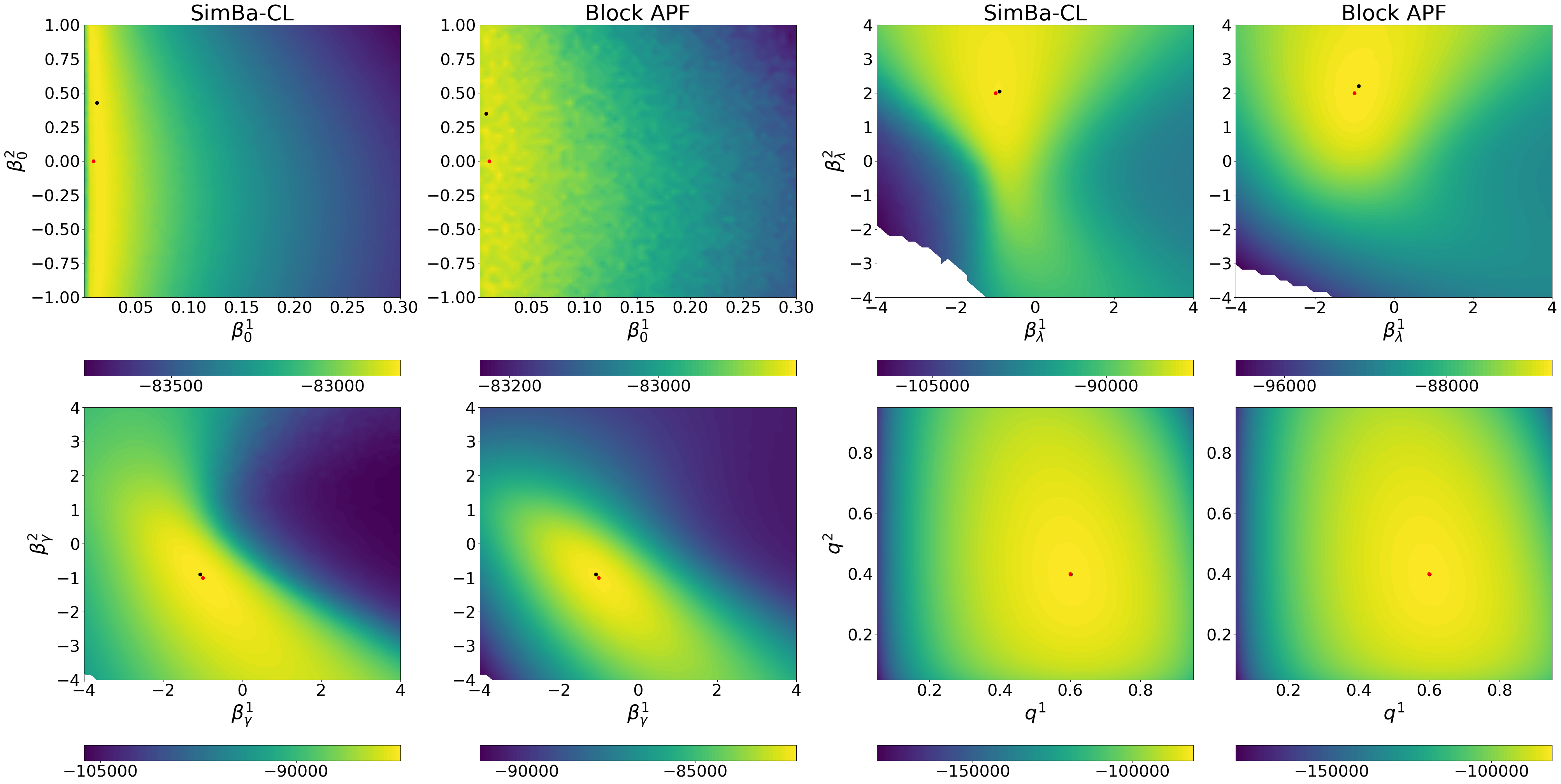}
    \caption{Profile log-likelihood surfaces for $\beta_0,\beta_\lambda,\beta_\gamma,q$ from fully factorized SimBa-CL without feedback (first and third column) and Block APF (second and the fourth columns). Dots are used for the data-generating parameter and the maximum on the grid. }
    \label{fig:simba_vs_block}
\end{figure*}

Notably, the method proposed by \cite{rimella2022approximating} emerges as the best method in terms of log-likelihood mean and variance, as it yields unbiased estimates of the likelihood and reduces the variance. Our SimBa-CL exhibits superior computational efficiency, with a running time that is also almost three times faster than vanilla Block APF. It is worth noting that even though the bias of SimBa-CL is significant the log-likelihood variance is considerably lower than vanilla APF and comparable with the Block APF. 
It can be observed that the Block APF run out of memory when increasing the number of particles. This is attributed to the necessity of maintaining $P$ particles for each individual and, at the same time, performing a multinomial resampling for each of these particles. 

We also run a paired comparison on the profile likelihood surfaces as in Section \ref{sec:comparison_simba}, and report them in Figure \ref{fig:simba_vs_block}. It can be noticed that both  SimBa-CL and Block APF generate similar likelihood surfaces whose maxima are close to the data-generating parameter. Both $P$ and the number of particles were set to $1024$. Unfortunately, we could not include in our studies the SMC from \cite{rimella2022approximating} for computational reasons.

\begin{table*}[httb!]
    \centering
    \begin{tabular}{lcccc}
     P                       & 512              & 1024             & 2048             & Time (sec)\\
    \hline
     APF               & Failed            & Failed            & Failed            & 1.2s \\
     h=5               & -43447.56 (52.04) & -43419.52 (51.08) & -43391.0 (52.41)  & 4.44s \\
     h=20              & -43004.55 (5.38)  & -43001.9 (4.65)   & -42999.76 (3.7)   & 11.08s \\
     h=50              & -42999.93 (3.44)  & -42998.13 (2.72)  & -42996.74 (2.39)  & 20.88s \\
     Block APF         & Failed            & Failed            & Failed            & 2.09s \\
     SimBa             & -43683.85 (9.54)  & -43683.67 (7.35)  & -43683.76 (5.16)  & 1.25s \\
    \hline
    \end{tabular}
    \caption{Log-likelihood means and log-likelihood standard deviations for the baseline SEIR model with $N=1000$. $h$ is the number of future observations included in \cite{rimella2022approximating} ($h=0$ correspond to APF).}
    \label{tab:N1000_SEIR}
\end{table*}

\subsubsection{SimBa-CL and SMC for an individual-based SEIR model}

The section concludes with a comparison of the baseline SEIR IBM. As for the SIS model, we simulate synthetic data, we run our SimBa-CL along with the baseline algorithms using the data-generating parameter, and we estimate the mean and variance of the resulting log-likelihood computations. The outcomes are reported in Table \ref{tab:N1000_SEIR}. 

Note that the SEIR scenario is considerably more complex than the SIS scenario. In the SEIR case, observing only infected and removed makes it difficult for the SMC algorithms to prevent particle failure, without the use of very informative proposal distributions. The underlying intuition is that if we propose a susceptible individual at time $t$ and then we observe that individual to be infected, or removed, at time $t+1$, it becomes impossible to recover from our incorrect proposal at time $t$. 

Table \ref{tab:N1000_SEIR} clearly shows that, to avoid failure of the SMC, we need a smart proposal distribution as the one proposed by \cite{rimella2022approximating}, and also a large $h$ to reach a reasonable log-likelihood variance. On the other hand, our SimBa-CL is able to reach comparable log-likelihood variance almost ten times faster than the SMC. 

\subsection{2001 UK Foot and Mouth Disease Outbreak} \label{sec:FMD}

In the year 2001, the United Kingdom experienced an outbreak of foot and mouth disease, a highly contagious virus affecting cloven-hoofed animals. Over an 8-month period, 2026 farms out of 188361 in the UK were infected, concentrated in the North and South West of England, and costing an estimated £8 billion to public and private sectors \citep{FMDreport2002}.  

Data on the outbreak are covered by copyright, see Defra (Department for environment, food and rural affairs) website (\href{http://www.defra.gov.uk}{http://www.defra.gov.uk}). From the full dataset we extracted: farm coordinates, farm notification status with date, number of cattle and number of sheep in the farm. We also localised our studies in the surroundings of Cumbria and selected a total of $8791$ farms, which integrates the studies by \cite{jewell2009bayesian}, see Figure \ref{fig:FMD_cities} for a graphical representation. 

\begin{figure}
    \centering
    \includegraphics[width=0.8\textwidth]{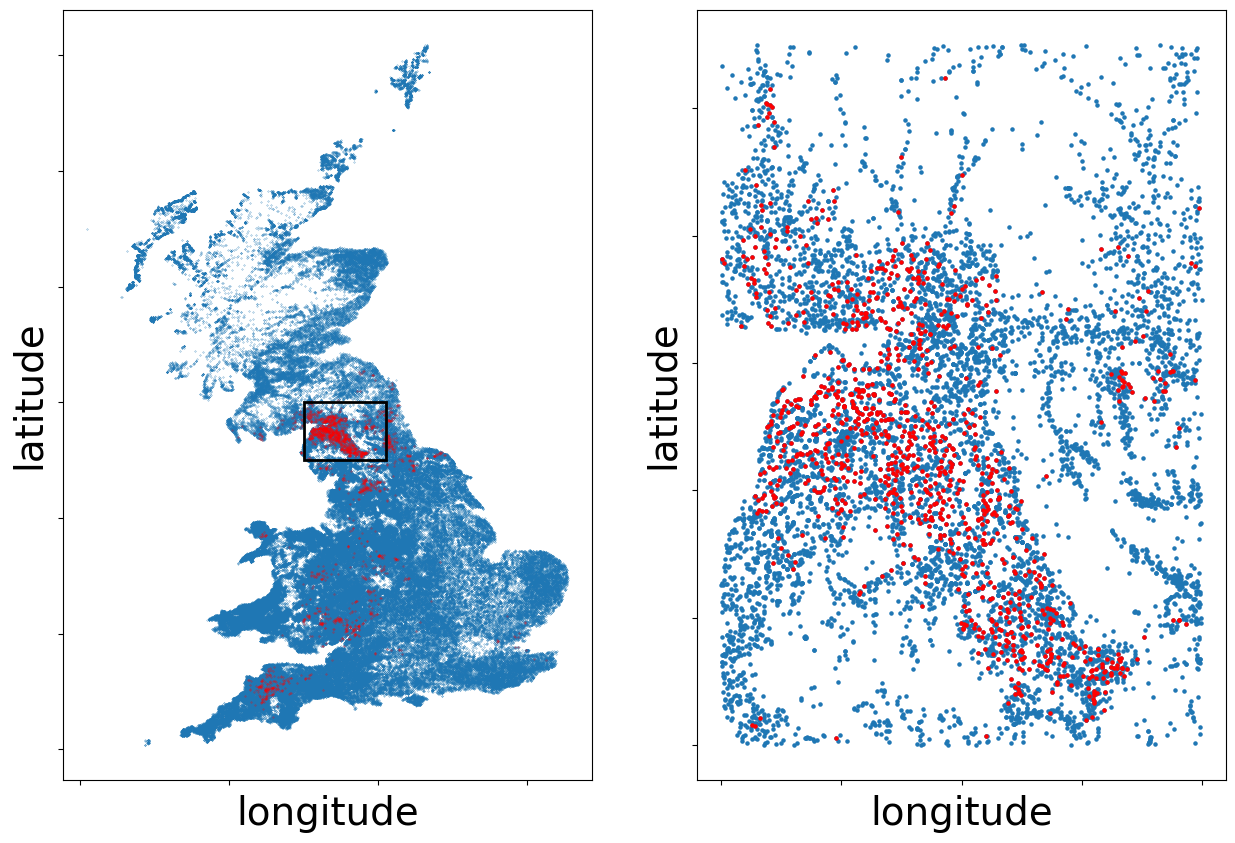}
    \caption{Graphical representation of the FMD outbreak. Top: the full outbreak in Great Britain. Bottom: the analysed farms. Red dots are used to indicate farms that got infected during the outbreak. The black rectangle encloses the analysed farms in the full plot.}
    \label{fig:FMD_cities}
\end{figure}

\begin{figure*}
    \centering
    \includegraphics[width=0.8\textwidth]{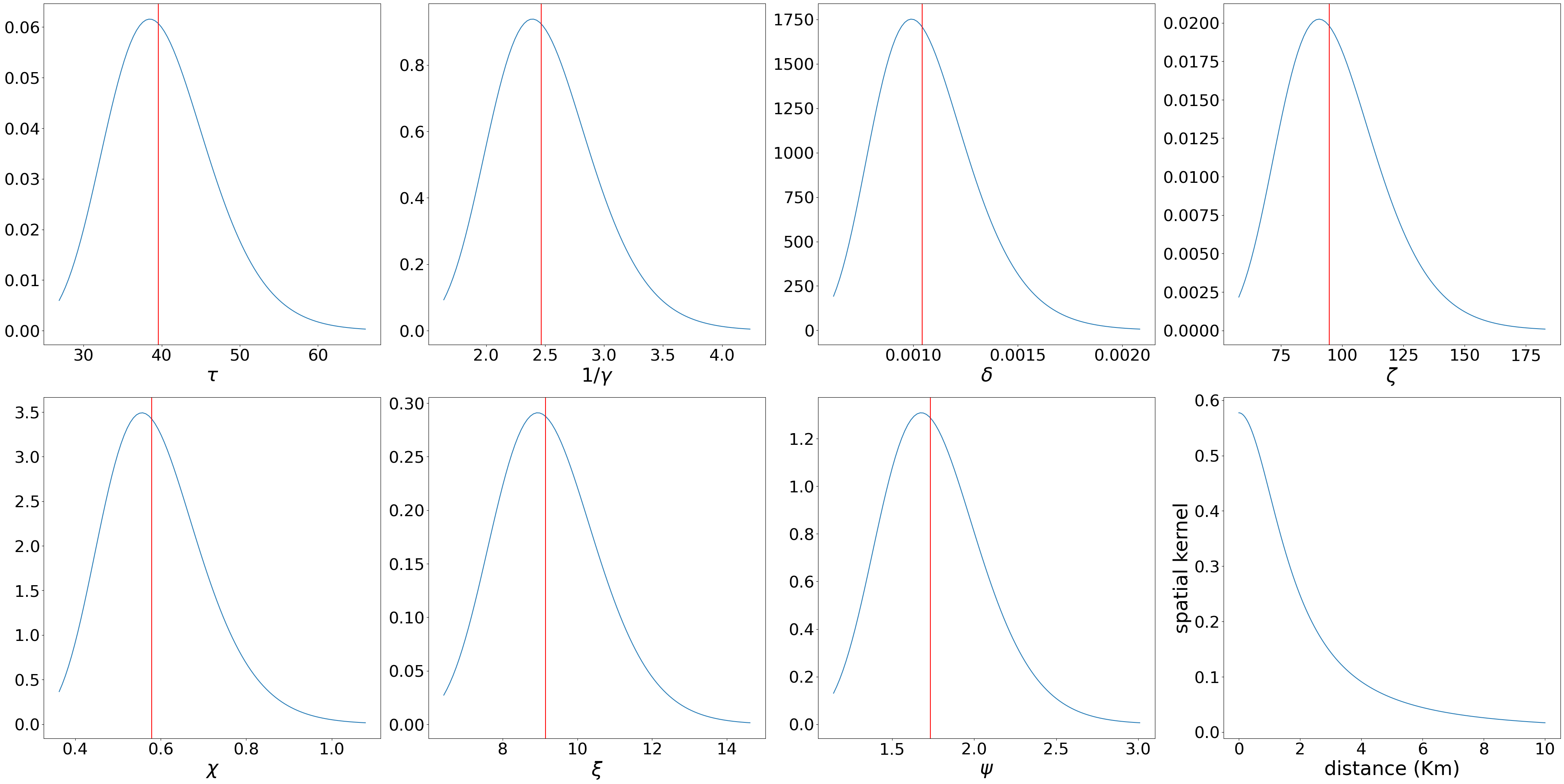}
    \caption{FMD parameters' distributions and spatial kernel decay. Parameters' labels are reported on the x-axis. Solid lines represent the maximum SimBa-CL estimator. The last plot of the second row shows the spatial decay of infectivity in Km. }
    \label{fig:FMD_parameters}
\end{figure*}

\subsubsection{Model}

Similar to previous models, we consider an IBM with farms as the individual. We assume farms exist in Susceptible, Infected, Notified (i.e. quarantined on detection), and Removed states (the SINR model).  Transitions from S to I and I to N follow a discrete-time stochastic process, with infected farms immediately quarantined, and N to R (farm culling) occurs deterministically after 1 day.

We consider an initial probability of infection for farm $n$ of $1-\exp\left \{-\tau {\sum_{\tilde{n} \in [N]} \lambda_{\tilde{n}, n}} \slash {N }\right \}$, with parameter $\tau > 0$.  We assume transition probabilities from I to N of $1-\exp\{-\gamma\}, \; \gamma > 0$, and from N to R of $1$. We assume individual infection probabilities, i.e. transitions from S to I, of $1-\exp\left \{- {\sum_{\tilde{n} \in [N]} \lambda_{\tilde{n}, n} \mathbb{I}\left( x_{t-1}^{\tilde{n}}=2\right)}\right \}$, where $\lambda_{\tilde{n}, n}$ is the infection pressure exerted by an infected farm $\tilde{n}$ to a susceptible farm $n$ and formulated as:
\begin{equation}\label{eq:FMD_is}
\begin{split}
    \lambda_{\tilde{n},n} = \frac{\delta}{N} &\left [ \zeta \left (w_{\tilde{n}}^c\right )^{\chi} + \left (w_{\tilde{n}}^s\right )^{\chi} \right ] \left [ \xi \left (w_n^c\right )^{\chi} + \left (w_n^s\right )^{\chi} \right ]  \frac{\psi}{E_{\tilde{n},n}^2 + \psi^2},
\end{split}
\end{equation}
with $\delta, \xi, \zeta, \chi, \psi$ positive parameters, $w_n^c$ number of cattle in the $n$-th farm, $w_n^s$ number of sheep in the $n$-th farm and $E_{\tilde{n},n}$ the Euclidean distance in kilometres between farm $\tilde{n}$ and farm $n$. This SINR model is an example of heterogeneous mixing IBM as the infectious contacts are not homogeneous in space. The emission distribution follows the usual formulation $p\left (y_t^n|x_{t}^n, \theta\right ) = q^{x_t^n} \mathbb{I}\left (y_t^n \neq 0\right ) +\left (1- q^{x_t^n}\right ) \mathbb{I}\left (y_t^n = 0\right )$, with $q=[0,0,1,0]^\top$ as we observe perfectly the notified.

We set $P=500$ when performing the optimization and $P=200$ when computing confidence regions.

\subsubsection{Inference}

We run $100$ optimization using Adam on our fully factorized SimBa-CL without feedback and select the ``best'' simulation according to its final SimBa-CL score. We then estimate the Godambe information matrix using the approximate Bartlett identities. As we learned the parameters in a log-scale, we need to use log-Normal when plotting the parameters distributions, see Figure \ref{fig:FMD_parameters}. 

The parameter $\tau$ can be intuitively understood as the time interval between the first infection and the first notified infections, marking the onset of the notification process. In Figure \ref{fig:FMD_parameters} we can recognize an optimal $\tau$ of about $40$, suggesting a relatively slow start in notifying farms. Furthermore, we can observe a mean time before notification of about $2.5$ days, leading to an estimate of $3.5$ days for the mean infection period, encompassing the period from farm infection to culling. This implies a relatively fast intervention once the notification process is implemented. Also, from the last row of Figure \ref{fig:FMD_parameters}, we can notice a decrease of over $60\%$ of the infectivity after just 2 km, which could be used to define containment zones around infected farms.

\begin{table}[httb!]
    \centering
    \begin{tabular}{rrrr}
       Nr. cattle &   Nr. sheep &   Mean susc. &   Mean infec. \\
       \hline
       100 &        0 &                131.83 &            1364.99 \\
          0 &     1000 &                 54.69 &              54.69 \\
         50 &      500 &                124.84 &             950.17 \\
          2 &        6 &                 16.49 &             144.37 \\
       \hline
       \end{tabular}
    \caption{Mean susceptibility and mean infectivity for four farms conformations.}
    \label{tab:FMD_SI_farms}
\end{table}

Parameters $\zeta, \chi, \xi$ are more difficult to interpret as they regulate the susceptibility and infectivity of farms according to the number of animals. To help visualize their effect we produce Table \ref{tab:FMD_SI_farms}, which shows the average susceptibility and infectivity of a medium-sized farm with only cattle, a medium-sized farm with only sheep, a large-sized farm, and a small-sized farm. From Table \ref{tab:FMD_SI_farms} we can deduce that the effect of owing cattle is significantly higher than the one of owing sheep for both infectivity and susceptibility, and that even small farms can affect the epidemic spread. This agrees with the study from \cite{jewell2013bayesian} and also with the \cite{EUFMDreport2003}.


\section{Discussion}

In this work, we propose SimBa-CL, a novel composite likelihood approach for inference in high-dimensional HMM, which relies on simulations from the model and basic forward recursions to compute Monte Carlo approximations of the marginals of the likelihood. The computational cost of the algorithm is quadratic in the dimension $N$ when considering simulations with feedback, targeting the marginals of the likelihood exactly, and can be reduced to linear when ignoring them, albeit by introducing an approximation. For well-mixing models where the effect of changing the state of one individual has a decreasing impact on the dynamics of the other individuals, we provide Kullback-Leibler divergence bounds showing that the effect of the feedback becomes negligible as $N$ increases. We demonstrate that SimBa-CL is competitive with state-of-the-art SMC algorithms in terms of likelihood variance while being significantly faster and suited to automatic differentiation, allowing straightforward optimization of the parameters via Machine Learning libraries.

Although our experiments with SimBa-CL are limited to individual-based models for epidemiology, we have presented the methodology in the context of general HMMs. Indeed, SimBa-CL could be applied to standard compartmental models for epidemics \citep{chowell2004basic,lekone2006statistical}, where individual counts are modelled instead of individual states. Here, conjugacy properties of the model could be exploited to simplify the SimBa-CL filter and feedback \citep{whitehouse2022consistent}. Furthermore, state-space models and HMMs are used to model unknown skill levels in competitive sports based on game outcomes \citep{Herbrich2006,Strumbelj2012,duffield2023state}. For instance, SimBa-CL on a partition of football teams could be applied to estimate the skilled players across teams over time. Also, the factorization in \eqref{eq:model_factorization} can be recognized in traffic congestion applications \citep{silva2015predicting,rimella2022exploiting}, where business states are estimated using data from preallocated sensors. In this context, SimBa-CL appears to be a promising approach, in particular, SimBa-CL with feedback would allow spatial correlations to be properly accounted for, ensuring that no information is lost.

Our empirical results show good performance for parameter estimation, particularly for large $N$. We believe this stems from the fact that the approximation of ignoring feedback decreases as $N\rightarrow\infty$ and composite likelihood methods have good asymptotic properties, albeit these are shown for simpler data models \citep{varin2008composite,varin2011overview}. In essence, we are conjecturing that the likelihood surface induced by SimBa-CL might have some asymptotic shape whose maximum is coalescing around the data-generating parameter or a set of equally likely parameters. However, a rigorous proof of consistency is beyond the scope of this article and we leave it to future works.

\section{Data availability}

Code is open-source and available at the Github repository:\\ \href{https://github.com/LorenzoRimella/SimBa-CL}{https://github.com/LorenzoRimella/SimBa-CL}. 

Proofs and additional details are available in the supplementary material.

\section{Acknowledgements}

The authors acknowledge the contributions of the editor and two anonymous reviewers of the journal ``Statistics and Computing'', whose insightful feedback significantly improved the paper. The authors also thank the High-End Computing (HEC) facility at Lancaster University for providing the computational resources necessary to run the experiments. This work was supported by the EPSRC grants EP/R018561/1 (Bayes4Health) and EP/R034710/1 (CoSInES). PF acknowledges funding from the EPSRC grant EP/Y028783/1 (Prob\_AI).

\bibliographystyle{chicago}
\bibliography{references.bib}

\begin{thebibliography}{}

\bibitem[\protect\citeauthoryear{Blundell, Cornebise, Kavukcuoglu, and
  Wierstra}{Blundell et~al.}{2015}]{blundell15}
Blundell, C., J.~Cornebise, K.~Kavukcuoglu, and D.~Wierstra (2015, 07--09 Jul).
\newblock Weight uncertainty in neural network.
\newblock In F.~Bach and D.~Blei (Eds.), {\em Proceedings of the 32nd
  International Conference on Machine Learning}, Volume~37 of {\em Proceedings
  of Machine Learning Research}, Lille, France, pp.\  1613--1622. PMLR.

\bibitem[\protect\citeauthoryear{Chopin, Papaspiliopoulos, et~al.}{Chopin
  et~al.}{2020}]{Chopin2020}
Chopin, N., O.~Papaspiliopoulos, et~al. (2020).
\newblock {\em An introduction to sequential Monte Carlo}, Volume~4.
\newblock Berlin: Springer.

\bibitem[\protect\citeauthoryear{Chowell, Hengartner, Castillo-Chavez,
  Fenimore, and Hyman}{Chowell et~al.}{2004}]{chowell2004basic}
Chowell, G., N.~W. Hengartner, C.~Castillo-Chavez, P.~W. Fenimore, and J.~M.
  Hyman (2004).
\newblock The basic reproductive number of {Ebola} and the effects of public
  health measures: the cases of {Congo} and {Uganda}.
\newblock {\em Journal of Theoretical Biology\/}~{\em 229\/}(1), 119--126.

\bibitem[\protect\citeauthoryear{Cocker, Chidziwisano, Mphasa, Mwapasa, Lewis,
  Rowlingson, Sammarro, Bakali, Salifu, Zuza, et~al.}{Cocker
  et~al.}{2023}]{cocker2023investigating}
Cocker, D., K.~Chidziwisano, M.~Mphasa, T.~Mwapasa, J.~M. Lewis, B.~Rowlingson,
  M.~Sammarro, W.~Bakali, C.~Salifu, A.~Zuza, et~al. (2023).
\newblock Investigating one health risks for human colonisation with extended
  spectrum beta-lactamase producing e. coli and k. pneumoniae in malawian
  households: a longitudinal cohort study.
\newblock {\em The Lancet Microbe\/}.

\bibitem[\protect\citeauthoryear{Cocker, Sammarro, Chidziwisano, Elviss, Jacob,
  Kajumbula, Mugisha, Musoke, Musicha, Roberts, Rowlingson, Singer, Byrne,
  Edwards, Lester, Wilson, Hollihead, Thomson, Jewell, Morse, and
  Feasey}{Cocker et~al.}{2023}]{cocker2023drum}
Cocker, D., M.~Sammarro, K.~Chidziwisano, N.~Elviss, S.~Jacob, H.~Kajumbula,
  L.~Mugisha, D.~Musoke, P.~Musicha, A.~Roberts, B.~Rowlingson, A.~Singer,
  R.~Byrne, T.~Edwards, R.~Lester, C.~Wilson, B.~Hollihead, N.~Thomson,
  C.~Jewell, T.~Morse, and N.~Feasey (2023).
\newblock {Drivers of Resistance in Uganda and Malawi (DRUM): a protocol for
  the evaluation of One-Health drivers of Extended Spectrum Beta Lactamase
  (ESBL) resistance in Low-Middle Income Countries (LMICs) [version 2; peer
  review: 1 approved, 1 approved with reservations]}.
\newblock {\em Wellcome Open Res\/}~{\em 7}, 55.

\bibitem[\protect\citeauthoryear{{Directive of the Council of the European
  Union}}{{Directive of the Council of the European
  Union}}{2003}]{EUFMDreport2003}
{Directive of the Council of the European Union} (2003).
\newblock {COUNCIL DIRECTIVE 2003\/85\/EC of 29 September 2003 on Community
  measures for the control of foot-and-mouth disease repealing Directive
  85\/511\/EEC and Decisions 89\/531\/EEC and 91\/665\/EEC and amending
  Directive 92\/46\/EEC, Official Journal of the European Union.}

\bibitem[\protect\citeauthoryear{Duffield, Power, and Rimella}{Duffield
  et~al.}{2023}]{duffield2023state}
Duffield, S., S.~Power, and L.~Rimella (2023).
\newblock A state-space perspective on modelling and inference for online skill
  rating.
\newblock {\em arXiv preprint arXiv:2308.02414\/}.

\bibitem[\protect\citeauthoryear{Fearnhead and Prangle}{Fearnhead and
  Prangle}{2012}]{fearnhead2012constructing}
Fearnhead, P. and D.~Prangle (2012).
\newblock Constructing summary statistics for approximate {Bayesian
  computation: semi-automatic approximate Bayesian computation}.
\newblock {\em Journal of the Royal Statistical Society Series B: Statistical
  Methodology\/}~{\em 74\/}(3), 419--474.

\bibitem[\protect\citeauthoryear{Fintzi, Cui, Wakefield, and Minin}{Fintzi
  et~al.}{2017}]{fintzi2017efficient}
Fintzi, J., X.~Cui, J.~Wakefield, and V.~N. Minin (2017).
\newblock Efficient data augmentation for fitting stochastic epidemic models to
  prevalence data.
\newblock {\em Journal of Computational and Graphical Statistics\/}~{\em
  26\/}(4), 918--929.

\bibitem[\protect\citeauthoryear{Glennie, Adam, Leos-Barajas, Michelot,
  Photopoulou, and McClintock}{Glennie et~al.}{2023}]{glennie2023hidden}
Glennie, R., T.~Adam, V.~Leos-Barajas, T.~Michelot, T.~Photopoulou, and B.~T.
  McClintock (2023).
\newblock Hidden {M}arkov models: Pitfalls and opportunities in ecology.
\newblock {\em Methods in Ecology and Evolution\/}~{\em 14\/}(1), 43--56.

\bibitem[\protect\citeauthoryear{Godambe}{Godambe}{1960}]{godambe1960optimum}
Godambe, V.~P. (1960).
\newblock An optimum property of regular maximum likelihood estimation.
\newblock {\em The Annals of Mathematical Statistics\/}~{\em 31\/}(4),
  1208--1211.

\bibitem[\protect\citeauthoryear{Gumbel}{Gumbel}{1954}]{gumbel1954statistical}
Gumbel, E.~J. (1954).
\newblock {\em Statistical theory of extreme values and some practical
  applications: a series of lectures}, Volume~33.
\newblock New York: US Government Printing Office.

\bibitem[\protect\citeauthoryear{Herbrich, Minka, and Graepel}{Herbrich
  et~al.}{2006}]{Herbrich2006}
Herbrich, R., T.~Minka, and T.~Graepel (2006).
\newblock True{S}kill: A {B}ayesian {S}kill {R}ating {S}ystem.
\newblock {\em Advances in neural information processing systems\/}~{\em 19}.

\bibitem[\protect\citeauthoryear{Jang, Gu, and Poole}{Jang
  et~al.}{2016}]{jang2016categorical}
Jang, E., S.~Gu, and B.~Poole (2016).
\newblock Categorical reparameterization with gumbel-softmax.
\newblock {\em arXiv preprint arXiv:1611.01144\/}.

\bibitem[\protect\citeauthoryear{Jewell, Brown, Keeling, Green, Roberts,
  et~al.}{Jewell et~al.}{2013}]{jewell2013bayesian}
Jewell, C., J.~Brown, M.~Keeling, L.~Green, G.~Roberts, et~al. (2013).
\newblock Bayesian epidemic risk prediction-knowledge transfer and usability at
  all levels.
\newblock In {\em Society for Veterinary Epidemiology and Preventive Medicine.
  Proceedings of a meeting held in Madrid, Spain, 20-22 March 2013}, pp.\
  127--142. Society for Veterinary Epidemiology and Preventive Medicine.

\bibitem[\protect\citeauthoryear{Jewell, Kypraios, Neal, and Roberts}{Jewell
  et~al.}{2009}]{jewell2009bayesian}
Jewell, C.~P., T.~Kypraios, P.~Neal, and G.~O. Roberts (2009).
\newblock Bayesian analysis for emerging infectious diseases.
\newblock {\em Bayesian Analysis\/}~{\em 4}, 465--496.

\bibitem[\protect\citeauthoryear{Ju, Heng, and Jacob}{Ju
  et~al.}{2021}]{ju2021sequential}
Ju, N., J.~Heng, and P.~E. Jacob (2021).
\newblock Sequential monte carlo algorithms for agent-based models of disease
  transmission.
\newblock {\em arXiv preprint arXiv:2101.12156\/}.

\bibitem[\protect\citeauthoryear{Keeling and Rohani}{Keeling and
  Rohani}{2011}]{keeling2011modeling}
Keeling, M.~J. and P.~Rohani (2011).
\newblock {\em Modeling infectious diseases in humans and animals}.
\newblock Princeton: Princeton university press.

\bibitem[\protect\citeauthoryear{Kingma and Ba}{Kingma and
  Ba}{2014}]{kingma2014adam}
Kingma, D.~P. and J.~Ba (2014).
\newblock Adam: A method for stochastic optimization.
\newblock {\em arXiv preprint arXiv:1412.6980\/}.

\bibitem[\protect\citeauthoryear{Lekone and Finkenst{\"a}dt}{Lekone and
  Finkenst{\"a}dt}{2006}]{lekone2006statistical}
Lekone, P.~E. and B.~F. Finkenst{\"a}dt (2006).
\newblock Statistical inference in a stochastic epidemic {SEIR} model with
  control intervention: {Ebola} as a case study.
\newblock {\em Biometrics\/}~{\em 62\/}(4), 1170--1177.

\bibitem[\protect\citeauthoryear{Maddison, Mnih, and Teh}{Maddison
  et~al.}{2016}]{maddison2016concrete}
Maddison, C.~J., A.~Mnih, and Y.~W. Teh (2016).
\newblock The concrete distribution: A continuous relaxation of discrete random
  variables.
\newblock {\em arXiv preprint arXiv:1611.00712\/}.

\bibitem[\protect\citeauthoryear{Rebeschini and Van~Handel}{Rebeschini and
  Van~Handel}{2015}]{rebeschini2015can}
Rebeschini, P. and R.~Van~Handel (2015).
\newblock Can local particle filters beat the curse of dimensionality?
\newblock {\em The Annals of Applied Probability\/}~{\em 25}, 2809--2866.

\bibitem[\protect\citeauthoryear{Rimella, Alderton, Sammarro, Rowlingson,
  Cocker, Feasey, Fearnhead, and Jewell}{Rimella
  et~al.}{023b}]{rimella2022inference}
Rimella, L., S.~Alderton, M.~Sammarro, B.~Rowlingson, D.~Cocker, N.~Feasey,
  P.~Fearnhead, and C.~Jewell (2023b, 07).
\newblock {Inference on extended-spectrum beta-lactamase Escherichia coli and
  Klebsiella pneumoniae data through SMC2}.
\newblock {\em Journal of the Royal Statistical Society Series C: Applied
  Statistics\/}, qlad055.

\bibitem[\protect\citeauthoryear{Rimella, Jewell, and Fearnhead}{Rimella
  et~al.}{023a}]{rimella2022approximating}
Rimella, L., C.~Jewell, and P.~Fearnhead (2023a, 05).
\newblock {Approximating Optimal SMC Proposal Distributions in Individual-Based
  Epidemic Models}.
\newblock {\em Statistica Sinica\/}, SS--2022--0198.

\bibitem[\protect\citeauthoryear{Rimella and Whiteley}{Rimella and
  Whiteley}{2022}]{rimella2022exploiting}
Rimella, L. and N.~Whiteley (2022, 08).
\newblock Exploiting locality in high-dimensional factorial hidden markov
  models.
\newblock {\em The Journal of Machine Learning Research\/}~{\em 23\/}(1),
  134--167.

\bibitem[\protect\citeauthoryear{Rimella, Whiteley, Jewell, Fearnhead, and
  Whitehouse}{Rimella et~al.}{2025}]{rimella2025CAL}
Rimella, L., N.~Whiteley, C.~Jewell, P.~Fearnhead, and M.~Whitehouse (2025).
\newblock Scalable calibration for partially observed individual-based epidemic
  models through categorical approximations.
\newblock {\em arXiv preprint arXiv:2501.03950\/}.

\bibitem[\protect\citeauthoryear{Samanidou, Zschischang, Stauffer, and
  Lux}{Samanidou et~al.}{2007}]{samanidou2007agent}
Samanidou, E., E.~Zschischang, D.~Stauffer, and T.~Lux (2007).
\newblock Agent-based models of financial markets.
\newblock {\em Reports on Progress in Physics\/}~{\em 70\/}(3), 409.

\bibitem[\protect\citeauthoryear{Scott}{Scott}{2002}]{scott2002bayesian}
Scott, S.~L. (2002).
\newblock Bayesian methods for hidden {M}arkov models: Recursive computing in
  the 21st century.
\newblock {\em Journal of the American statistical Association\/}~{\em
  97\/}(457), 337--351.

\bibitem[\protect\citeauthoryear{Silva, Kang, and Airoldi}{Silva
  et~al.}{2015}]{silva2015predicting}
Silva, R., S.~M. Kang, and E.~M. Airoldi (2015).
\newblock Predicting traffic volumes and estimating the effects of shocks in
  massive transportation systems.
\newblock {\em Proceedings of the National Academy of Sciences\/}~{\em
  112\/}(18), 5643--5648.

\bibitem[\protect\citeauthoryear{Touloupou, Finkenst{\"a}dt, and
  Spencer}{Touloupou et~al.}{2020}]{touloupou2020scalable}
Touloupou, P., B.~Finkenst{\"a}dt, and S.~E. Spencer (2020).
\newblock {Scalable Bayesian inference for coupled hidden Markov and
  semi-Markov models}.
\newblock {\em Journal of Computational and Graphical Statistics\/}~{\em
  29\/}(2), 238--249.

\bibitem[\protect\citeauthoryear{{UK National Audit Office}}{{UK National Audit
  Office}}{2022}]{FMDreport2002}
{UK National Audit Office} (2022).
\newblock {Report by the Comptroller and auditor general: "The 2001 outbreak of
  foot and mouth disease."}.

\bibitem[\protect\citeauthoryear{Varin}{Varin}{2008}]{varin2008composite}
Varin, C. (2008).
\newblock On composite marginal likelihoods.
\newblock {\em Asta advances in statistical analysis\/}~{\em 92\/}(1), 1--28.

\bibitem[\protect\citeauthoryear{Varin, Reid, and Firth}{Varin
  et~al.}{2011}]{varin2011overview}
Varin, C., N.~Reid, and D.~Firth (2011).
\newblock An overview of composite likelihood methods.
\newblock {\em Statistica Sinica\/}, 5--42.

\bibitem[\protect\citeauthoryear{Whitehouse, Whiteley, and Rimella}{Whitehouse
  et~al.}{2023}]{whitehouse2022consistent}
Whitehouse, M., N.~Whiteley, and L.~Rimella (2023, 07).
\newblock {Consistent and fast inference in compartmental models of epidemics
  using Poisson Approximate Likelihoods}.
\newblock {\em Journal of the Royal Statistical Society Series B: Statistical
  Methodology\/}, qkad065.

\bibitem[\protect\citeauthoryear{Whiteley and Lee}{Whiteley and
  Lee}{2014}]{whiteley2014}
Whiteley, N. and A.~Lee (2014).
\newblock {Twisted particle filters}.
\newblock {\em The Annals of Statistics\/}~{\em 42\/}(1), 115 -- 141.

\bibitem[\protect\citeauthoryear{Wilkinson}{Wilkinson}{2018}]{wilkinson2018stochastic}
Wilkinson, D.~J. (2018).
\newblock {\em Stochastic modelling for systems biology}.
\newblock New York: CRC press.

\bibitem[\protect\citeauthoryear{Zeiler}{Zeiler}{2012}]{zeiler2012adadelta}
Zeiler, M.~D. (2012).
\newblock Adadelta: an adaptive learning rate method.
\newblock {\em arXiv preprint arXiv:1212.5701\/}.

\bibitem[\protect\citeauthoryear{Štrumbelj and Vračar}{Štrumbelj and
  Vračar}{2012}]{Strumbelj2012}
Štrumbelj, E. and P.~Vračar (2012).
\newblock Simulating a basketball match with a homogeneous {M}arkov model and
  forecasting the outcome.
\newblock {\em International Journal of Forecasting\/}~{\em 28\/}(2), 532--542.

\end{thebibliography}

\appendix

\section{Simulation based composite likelihood: assumptions and proofs}

In this section we provide the details about the theoretical results from the main paper. Apart from giving full proofs of the results, we will also discuss the validity of the assumptions. 

\subsection{Discussion on the assumptions and the simulation feedback}

Here we show that our assumptions from the main paper are valid for a specific class of models. Specifically, we consider an SIS individual-based model as from the experimental section, and so:
\begin{align}
    &p\left (x_0^n\right ) = 
    \begin{bmatrix}
        1-\frac{1}{1+\exp{\left (-\beta_0^\top w_n\right )}} \\ 
            \frac{1}{1+\exp{\left (-\beta_0^\top w_n\right )}}	
    \end{bmatrix};\\
    & p\left (x_t^n|x_{t-1}\right ) = 
    \begin{bmatrix}
       e^{ - \lambda_n \left ( \frac{\sum_{\bar{n} \in [N]} \mathbb{I}\left (x_{t-1}^{\bar{n}}=2\right ) }{N} + \iota \right ) }&   
            1- e^{ - \lambda_n \left ( \frac{\sum_{\bar{n} \in [N]} \mathbb{I}\left (x_{t-1}^{\bar{n}}=2\right ) }{N} + \iota \right ) }\\
        1-e^{- \gamma_n } & 
            e^{- \gamma_n }
    \end{bmatrix};\\
    &p\left (y_t^n|x_{t}^n\right ) = q^{x_t^n} \mathbb{I}\left (y_t^n \neq 0\right ) +\left (1- q^{x_t^n}\right ) \mathbb{I}\left (y_t^n = 0\right ),
\end{align}
where $\lambda_n, \gamma_n$ are positive and $q \in [0,1]^2$.

The first assumption we consider is reported below. 

\begin{assumption}\label{ass:kernel_app}
For any $n, \bar{n} \in [N]$ and for any $x_t^{\bar{n}} \in \mathcal{X}$, if $x_{t-1}, \bar{x}_{t-1} \in \mathcal{X}^N$ are such that $x_{t-1}^{\setminus {n}}= \bar{x}_{t-1}^{\setminus {n}}$ then:
\begin{equation}
    \left | p\left (x_t^{\bar{n}}|x_{t-1}\right ) - p\left (x_t^{\bar{n}}|\bar{x}_{t-1}\right ) \right | \leq \frac{1}{N} \left |d_{n,\bar{n}}\left (x_{t-1}^{{n}}\right ) - d_{n,\bar{n}}\left (\bar{x}_{t-1}^{{n}}\right ) \right |,
\end{equation}
where $d_{n,\bar{n}}:\mathcal{X} \to \mathbb{R}_+$. 
\end{assumption}

This is one of the key assumptions to prove our first result on SimBa-CL and precisely provide bounds on the Kullback-Leibler divergence between the fully factorised SimBa-CL with and without feedback.

To prove the validity of the bound we simply consider the difference between $p\left (x_t^{\bar{n}}|x_{t-1}\right )$ and $p\left (x_t^{\bar{n}}|\bar{x}_{t-1}\right )$ when $x_{t-1}, \bar{x}_{t-1} \in \mathcal{X}^N$ are such that $x_{t-1}^{\setminus {n}}= \bar{x}_{t-1}^{\setminus {n}}$.
\begin{equation}
    \begin{split}
        &p\left (x_t^{\bar{n}}|x_{t-1}\right ) - p\left (x_t^{\bar{n}}|\bar{x}_{t-1}\right ) 
        \\&=
            \begin{cases}
            e^{ - \lambda_{\bar{n}} \left ( \frac{\sum_{{m} \in [N]} \mathbb{I}\left (x_{t-1}^m=2\right ) }{N} + \iota \right ) } - e^{ - \lambda_n \left ( \frac{\sum_{m \in [N]} \mathbb{I}\left (\bar{x}_{t-1}^m=2\right ) }{N} + \iota \right ) } & \text{if } x_{t-1}^{\bar{n}}=1, x_t^{\bar{n}}=1\\  
            -e^{ - \lambda_{\bar{n}} \left ( \frac{\sum_{{m} \in [N]} \mathbb{I}\left (x_{t-1}^m=2\right ) }{N} + \iota \right ) } + e^{ - \lambda_n \left ( \frac{\sum_{m \in [N]} \mathbb{I}\left (\bar{x}_{t-1}^m=2\right ) }{N} + \iota \right ) } & \text{if } x_{t-1}^{\bar{n}}=1, x_t^{\bar{n}}=2\\
            0 & \text{if } x_{t-1}^{\bar{n}}=2, x_t^{\bar{n}}=1 \\
            0 & \text{if } x_{t-1}^{\bar{n}}=2, x_t^{\bar{n}}=2
    \end{cases}
    \end{split}
\end{equation}
hence:
\begin{equation}
    \left | p\left (x_t^{\bar{n}}|x_{t-1}\right ) - p\left (x_t^{\bar{n}}|\bar{x}_{t-1}\right ) \right | \leq \left |e^{ - \lambda_{\bar{n}} \left ( \frac{\sum_{{m} \in [N]} \mathbb{I}\left (x_{t-1}^m=2\right ) }{N} + \iota \right ) } - e^{ - \lambda_n \left ( \frac{\sum_{m \in [N]} \mathbb{I}\left (\bar{x}_{t-1}^m=2\right ) }{N} + \iota \right ) }\right |.
\end{equation}
Consider the function $\exp(-x)$, from the mean value theorem we have that:
\begin{equation}
    e^{-a}-e^{-b} = (-e^{-c})(a-b),
\end{equation}
with $c \in [a,b]$ provide $a<b$. Then we can apply the mean value theorem on the domain $[0,+\infty]$ and get:
\begin{equation}
\begin{split}
   &\left | p\left (x_t^{\bar{n}}|x_{t-1}\right ) - p\left (x_t^{\bar{n}}|\bar{x}_{t-1}\right ) \right | \leq \left |e^{ - \lambda_{\bar{n}} \left ( \frac{\sum_{{m} \in [N]} \mathbb{I}\left (x_{t-1}^m=2\right ) }{N} + \iota \right ) } - e^{ - \lambda_n \left ( \frac{\sum_{m \in [N]} \mathbb{I}\left (\bar{x}_{t-1}^m=2\right ) }{N} + \iota \right ) }\right |\\
   &= \left |e^{ - c} \right | \left | \lambda_{\bar{n}} \left ( \frac{\sum_{{m} \in [N]} \mathbb{I}\left (x_{t-1}^m=2\right ) }{N} + \iota \right )  -  \lambda_n \left ( \frac{\sum_{m \in [N]} \mathbb{I}\left (\bar{x}_{t-1}^m=2\right ) }{N} + \iota \right ) \right |\\
   &\leq   \lambda_{\bar{n}} \left | \left ( \frac{\sum_{{m} \in [N]} \mathbb{I}\left (x_{t-1}^m=2\right ) }{N} + \iota \right )  -  \left ( \frac{\sum_{m \in [N]} \mathbb{I}\left (\bar{x}_{t-1}^m=2\right ) }{N} + \iota \right ) \right |,
\end{split}
\end{equation}
where the last step follows from $c$ and $\lambda_{\bar{n}}$ being positive. Now remark that $x_{t-1}^{\setminus {n}}= \bar{x}_{t-1}^{\setminus {n}}$, hence:
\begin{equation}
\begin{split}
   &\left | p\left (x_t^{\bar{n}}|x_{t-1}\right ) - p\left (x_t^{\bar{n}}|\bar{x}_{t-1}\right ) \right | 
   \leq   \lambda_{\bar{n}} \left | \left ( \frac{\sum_{{m} \in [N]} \mathbb{I}\left (x_{t-1}^m=2\right ) }{N} + \iota \right )  -  \left ( \frac{\sum_{m \in [N]} \mathbb{I}\left (\bar{x}_{t-1}^m=2\right ) }{N} + \iota \right ) \right |\\
   &\leq   \lambda_{\bar{n}} \left | \frac{\sum_{{m} \in [N]} \mathbb{I}\left (x_{t-1}^m=2\right ) }{N} - \frac{\sum_{m \in [N]} \mathbb{I}\left (\bar{x}_{t-1}^m=2\right ) }{N}  \right |\\
   &\leq   \frac{\lambda_{\bar{n}}}{N} \left | \mathbb{I}\left (x_{t-1}^n=2\right ) - \mathbb{I}\left (\bar{x}_{t-1}^n=2\right )  \right |,
\end{split}
\end{equation}
and so our assumption is satisfied for $d_{n,\bar{n}}(x) = \lambda_{\bar{n}} \mathbb{I}\left (x=2\right )$. Note that the proof is straightforward to generalise for transition kernels that include infectivity and spatial kernels. Indeed, instead of $\lambda_{\bar{n}} \sum_{{m} \in [N]} \mathbb{I}\left (x_{t-1}^m=2\right )$ we could have $\lambda_{\bar{n}} \sum_{{m} \in [N]} \xi_{m} \psi_{\bar{n},m} \mathbb{I}\left (x_{t-1}^m=2\right )$ where $\xi_{m}$ is the infectivity and $\psi_{\bar{n},m}$ is the spatial kernel telling how individuals $\bar{n}$ and $m$ are connected. This alternative transition kernel will similarly satisfy our assumption with for $d_{n,\bar{n}}(x) = \lambda_{\bar{n}} \xi_n \psi_{n, \bar{n}} \mathbb{I}\left (x=2\right )$.

We now conclude by proving the validity of the second assumption, reported for completeness below.

\begin{assumption}\label{ass:boundeness_app}
    For any $n, \bar{n} \in [N]$, if $x_{t-1}, \bar{x}_{t-1} \in \mathcal{X}^N$ are such that $x_{t-1}^{\setminus \bar{n}}= \bar{x}_{t-1}^{\setminus \bar{n}}$ then there exists $ 0 < \epsilon < 1$ such that:
    \begin{equation}
        \sum_{x^{{n}}_t} p\left (x_t^{{n}}|x_{t-1}\right ) \frac{1}{p\left (x_t^{n}|\bar{x}_{t-1}\right )^2} \leq \frac{1}{\epsilon^2}, \quad \text{and} \quad
        \sum_{x^{{n}}_t} p\left (x_t^{{n}}|x_{t-1}\right ) \frac{1}{p\left (x_t^{n}|\bar{x}_{t-1}\right )^3} \leq \frac{1}{\epsilon^3}.
    \end{equation}
\end{assumption}

Again if we consider our SIS individuals-based model we can notice that \ref{ass:boundeness_app} is simply requiring that the matrix obtained by doing the ratio element by element of the transition kernel computed in ${x}_{t-1}$ with the square\slash cube of itself computed in $\bar{x}_{t-1}$ has the sum of the rows bounded. This is straightforward to prove by noting that:
\begin{equation}
\begin{split}
     p\left (x_t^n|x_{t-1}\right ) &= 
    \begin{bmatrix}
    e^{ - \lambda_n \left ( \frac{\sum_{\bar{n} \in [N]} \mathbb{I}\left (x_{t-1}^{\bar{n}}=2\right ) }{N} + \iota \right ) }&   
        1- e^{ - \lambda_n \left ( \frac{\sum_{\bar{n} \in [N]} \mathbb{I}\left (x_{t-1}^{\bar{n}}=2\right ) }{N} + \iota \right ) }\\
    1-e^{- \gamma_n } & 
        e^{- \gamma_n }
    \end{bmatrix}   \\
    &\geq
    \begin{bmatrix}
    e^{ - \lambda_n \left ( 1 + \iota \right )  }&   
        1- e^{ - \lambda_n \iota }\\
    1-e^{- \gamma_n } & 
        e^{- \gamma_n }
    \end{bmatrix},
\end{split}
\end{equation}
hence if we choose $\epsilon = \max_{n \in [N]} \max \{ \exp( - \lambda_n \left ( 1 + \iota \right )  ), 1- \exp( - \lambda_n \iota ), 1-\exp(- \gamma_n ), \exp(- \gamma_n )\}$, the validity of the assumption follows trivially.

There is still an additional assumption that used for the proof of theorem, and precisely: $\left | d_{n,\bar{n}} \left ( x^{n} \right )  - d_{n,\bar{n}} \left ( \bar{x}^{n} \right ) \right | < N$. Again for the specific case of SIS we have $d_{n,\bar{n}}(x) = \lambda_{\bar{n}}\mathbb{I}\left (x=2\right )$ and so $\lambda_{\bar{n}} \left |\mathbb{I}\left (x=2\right )  - \mathbb{I}\left (\bar{x}=2\right ) \right |$ which is less or equal to 1 when $\lambda_{\bar{n}}<1$ and so smaller than $N$. Remark that in practice $\lambda_{\bar{n}}$ is a logistic regression of some covariates of individual $\bar{n}$, guaranteeing $\lambda_{\bar{n}}<1$.

Our assumptions are not limited to assumptions \ref{ass:kernel_app}-\ref{ass:boundeness_app}, and we need some adjusted versions of them when working with SimBa-CL for general partitions, which are again reported below for completeness.

\begin{assumption}\label{ass:general_kernel_ass}
For any $K \in \mathcal{K}$ and for any $x_t^{\bar{n}} \in \mathcal{X}$ with ${\bar{n}} \notin K$, if $x_{t-1}, \bar{x}_{t-1} \in \mathcal{X}^N$ are such that $x_{t-1}^{\setminus {K}}= \bar{x}_{t-1}^{\setminus {K}}$ then:
    \begin{equation}
        \left | p\left (x_t^{\bar{n}}|x_{t-1}\right ) - p\left (x_t^{\bar{n}}|\bar{x}_{t-1}\right ) \right | \leq \frac{1}{N} \left |d_{K,\bar{n}}\left (x_{t-1}^{n}\right ) - d_{K,\bar{n}}\left (\bar{x}_{t-1}^{n}\right ) \right |.
    \end{equation}
where $d_{K,\bar{n}}:\mathcal{X}^K \to \mathbb{R}_+$.
\end{assumption}

\begin{assumption}\label{ass:general_boundeness_ass}
    For any $K, \bar{K} \in \mathcal{K}$, if $x_{t-1}, \bar{x}_{t-1} \in \mathcal{X}^N$ are such that $x_{t-1}^{\setminus \bar{K}}= \bar{x}_{t-1}^{\setminus \bar{K}}$ then there exists $ 0 < \epsilon < 1$ such that:
    \begin{equation}
        \sum_{x^{{n}}_t} p\left (x_t^{{n}}|x_{t-1}\right ) \frac{1}{p\left (x_t^{n}|\bar{x}_{t-1}\right )^2} \leq \frac{1}{\epsilon^2}, \quad \text{and} \quad
        \sum_{x^{{n}}_t} p\left (x_t^{{n}}|x_{t-1}\right ) \frac{1}{p\left (x_t^{n}|\bar{x}_{t-1}\right )^3} \leq \frac{1}{\epsilon^3}.
    \end{equation}
\end{assumption}

Again both assumptions are valid for the SIS case. Indeed, Assumption \ref{ass:general_kernel_ass} can be proven to be valid for SIS by following the same steps and by observing that:
\begin{equation}
\begin{split}
   &\left | p\left (x_t^{\bar{n}}|x_{t-1}\right ) - p\left (x_t^{\bar{n}}|\bar{x}_{t-1}\right ) \right | 
   \leq   \frac{\lambda_{\bar{n}}}{N} \left | \sum_{n \in K} \mathbb{I}\left (x_{t-1}^n=2\right ) - \mathbb{I}\left (\bar{x}_{t-1}^n=2\right )  \right |,
\end{split}
\end{equation}
and so our $d_{K,\bar{n}}(x^K) = \lambda_{\bar{n}}\sum_{n \in K} \mathbb{I}\left (x^n=2\right )$, from which we can also notice $|d_{K,\bar{n}}(x^K) - d_{K,\bar{n}}(\bar{x}^K)|<N$. At the same time, the proof of Assumption \ref{ass:general_boundeness_ass} does not change.

We now discuss the simulation feedback. Note that we have a recursive formula for $p(x_{[0:T]}^{n}|x_{[0:T]}^{\setminus n})$:
\begin{equation}
    \begin{split}
        p(x_{[0:T]}^n|x_{[0:T]}^{\setminus n}) 
        &=
        p(x_{T}|x_{[0:T-1]}) \frac{p(x_{[0:T-1]}^n|x_{[0:T-1]}^{\setminus n})}{p(x_{[T]}^{\setminus n}|x_{[0:T-1]}^{\setminus n})} \\
        &=
        p(x_{T}^n|x_{T-1}) \frac{p(x_{T}^{\setminus n}|x_{T-1})}{p(x_{[T]}^{\setminus n}|x_{[0:T-1]}^{\setminus n})} p(x_{[0:T-1]}^n|x_{[0:T-1]}^{\setminus n})\\
        &=
        p(x_{T}^n|x_{T-1})\\
        &\quad \frac{ \prod_{\bar{n} \in [N], \bar{n} \neq n} p(x_{T}^{\bar{n}}|x_{T-1})}{\sum_{\bar{x}_{T-1}^n} \prod_{\bar{n} \in [N], \bar{n} \neq n} p(x_{T}^{\bar{n}}|\bar{x}_{T-1}^n, x_{[T-1]}^{\setminus n}) p(\bar{x}_{t-1}^n|x_{[0:T-1]}^{\setminus n})} \\
        &\quad p(x_{[0:T-1]}^n|x_{[0:T-1]}^{\setminus n}),
    \end{split}
\end{equation}
and so if we follow the same step down to $t=1$ we obtain: 
\begin{equation} \label{eq:recursive_conditional_joint}
    \begin{split}
        p(x_{[0:T]}^n|x_{[0:T]}^{\setminus n}) 
        &=
        \prod_{t \in [T]} p(x_{t}^n|x_{t-1})\\
        &\quad  
        \prod_{t \in [T]} \frac{ \prod_{\bar{n} \in [N], \bar{n} \neq n} p(x_{t}^{\bar{n}}|x_{t-1})}{\sum_{\bar{x}_{t-1}^n} \prod_{\bar{n} \in [N], \bar{n} \neq n} p(x_{t}^{\bar{n}}|\bar{x}_{t-1}^n, x_{[t-1]}^{\setminus n}) p(\bar{x}_{t-1}^n|x_{[0:T-1]}^{\setminus n})} \\
        &\quad 
        p(x_{0}^n|x_{0}^{\setminus n})\\
        &=
        p(x_{0}^n) \prod_{t \in [T]} p(x_{t}^n|x_{t-1})\\
        &\quad  
        \prod_{t \in [T]} \frac{ \prod_{\bar{n} \in [N], \bar{n} \neq n} p(x_{t}^{\bar{n}}|x_{t-1})}{\sum_{\bar{x}_{t-1}^n} \prod_{\bar{n} \in [N], \bar{n} \neq n} p(x_{t}^{\bar{n}}|\bar{x}_{t-1}^n, x_{[t-1]}^{\setminus n}) p(\bar{x}_{t-1}^n|x_{[0:T-1]}^{\setminus n})}.
    \end{split}
\end{equation}
Here we recognise our simulation feedback indeed in the main paper we have:
\begin{equation} 
    f \left (x_{t-1}^n, x_{[0:t]}^{\setminus n} \right ) =
    \frac{ \prod_{\bar{n} \in [N] \setminus n} p\left (x_{t}^{\bar{n}}|x_{t-1},\right )}{\sum_{\bar{x}_{t-1}^n} \prod_{\bar{n} \in [N] \setminus n} p\left (x_{t}^{\bar{n}}|\bar{x}_{t-1}^n, x_{[t-1]}^{\setminus n},\right ) p\left (\bar{x}_{t-1}^n| x_{[0:t-1]}^{\setminus n},\right )}.
\end{equation}

We can then notice that removing the simulation feedback is like assuming some form of independence on the future as it is like we are recursively removing terms of the form ${p(x_{t}^{\setminus n}|x_{t-1})} \slash {p(x_{[t]}^{\setminus n}|x_{[0:t-1]}^{\setminus n})}$.

\subsection{KL bounds for fully factorized SimBa-CL}

To prove Theorem 1 we require a series of intermediate results.

We start by showing the Data processing inequality, which allows us to bound the KL between the marginals with the KL between the joints.

\begin{lemma}{\bf (Data processing inequality)}\label{lemma:data_proc_ineq}
Consider two joint distributions $p(x,y)$ and $q(x,y)$ and let $p_y(y)$ and $q_y(y)$ be their marginals over $y$, then:
\begin{equation}\label{eq:data_proc_ineq}
    \mathbf{KL}(p_y(\mathbf{y})||q_y(\mathbf{y})) \leq \mathbf{KL}(p(\mathbf{x},\mathbf{y})||q(\mathbf{x},\mathbf{y}))
\end{equation}
provided that $q$ is absolutely continuous with respect to $p$.
\end{lemma}

\begin{proof}
It can be easily proved that if $q$ is absolutely continuous with respect to $p$ then $q_y$ is absolutely continuous with respect to $p_y$ and so also the left-hand side of the statement is well-defined.
	
To prove \eqref{eq:data_proc_ineq} we can simply use the chain rule:
\begin{equation}
    \begin{split}
    \mathbf{KL}(p(\mathbf{x},\mathbf{y})||q(\mathbf{x},\mathbf{y})) 
    &= 
    \sum_{x,y} p(x,y) \log \left ( \frac{p(x,y)}{q(x,y)} \right )
    = 
    \sum_{x,y} p_{x|y}(x|y)p_y(y) \log \left ( \frac{p_{x|y}(x|y)p_y(y)}{q_{x|y}(x|y)q_y(y)} \right )\\
    &=
    \sum_{x,y} p_{x|y}(x|y)p_y(y) \left [ \log \left ( \frac{p_{x|y}(x|y)}{q_{x|y}(x|y)} \right ) + \log \left ( \frac{p_y(y)}{q_y(y)} \right ) \right ]\\
    &=
    \sum_{x,y} p_{x|y}(x|y)p_y(y) \log \left ( \frac{p_{x|y}(x|y)}{q_{x|y}(x|y)} \right ) + \sum_{x,y} p_{x|y}(x|y)p_y(y) \log \left ( \frac{p_y(y)}{q_y(y)} \right ) \\
    &=
    \mathbb{E}_{p_y(\mathbf{y})} \left [ \sum_{x} p_{x|y}(x|\mathbf{y}) \log \left ( \frac{p_{x|y}(x|\mathbf{y})}{q_{x|y}(x|\mathbf{y})} \right ) \right ] + \sum_{y} p_y(y) \log \left ( \frac{p_y(y)}{q_y(y)} \right ) \\
    &=
    \mathbb{E}_{p_y(\mathbf{y})} \left [ \mathbf{KL}(p_{x|y}(\mathbf{x}|\mathbf{y})||q_{x|y}(\mathbf{x}|\mathbf{y})) \right ] + \mathbf{KL}(p_y(\mathbf{y})||q_y(\mathbf{y}))\\
        &\geq \mathbf{KL}(p_y(\mathbf{y})||q_y(\mathbf{y})),
    \end{split}
\end{equation}
where the last step follows from the positivity of the KL divergence.
\end{proof}

The next proposition gives us a bound for the log of a ratio, which will be used in the proof of the main theorem to bound the KL.

\begin{proposition}\label{prop:log_bound}
	Consider the function $f(x,y) = \log ( x \slash y)$ with $x,y \in (0,1]$ then:
	\begin{equation}
		f(x,y) \leq 
		\begin{cases}
		-\frac{y-x}{x}+\frac{(y-x)^2}{2 x^2} & y \geq x\\
		-\frac{y-x}{x}+\frac{(y-x)^2}{2 x^2} + \frac{\left({x-y} \right )^3}{3 x^2 y}  & y <x
		\end{cases}.
	\end{equation}
\end{proposition}

\begin{proof}
Note that:
$$
f(x,y) = log \left ( \frac{x}{y} \right ) = -log \left ( \frac{y}{x} \right ) = -log \left ( 1+\frac{y-x}{x} \right )
$$
and given that ${y} \slash {x} > 0$ also $1+\frac{y-x}{x} > 0$, meaning that $\frac{y-x}{x} > -1$. We can then provide bounds for $h(z)=-\log(1+z)$ with $z \in (-1, +\infty)$. To do so we analyse $z \in (-1,0)$ and $z \in [0, +\infty)$. 

Let start with the case $z \in [0, +\infty)$, notice that $\tilde{h}(z)=h(z)+z-\frac{z^2}{2}$ is such that:
\begin{itemize}
    \item $\tilde{h}(0)=h(0)=0$;
    \item $\tilde{h}^\prime(z) = {h}^\prime(z)+1-z= \frac{-1+1-z^2}{1+z}=\frac{-z^2}{1+z}<0$ for all $z \in [0, +\infty)$;
\end{itemize}	
from which we conclude $\tilde{h}(z) \leq 0$ for all $z \in [0, +\infty)$, given that $\tilde{h}(0)=0$ and the function is strictly decreasing on the domain. This implies $h(z) \leq -z+\frac{z^2}{2}$ on $z \in [0, +\infty)$.

Let us now analyse the case $z \in (-1, 0)$, notice that $\tilde{h}(z)$ is also decreasing on $z \in (-1, 0)$ and given that $\tilde{h}(0)=0$, we can only prove $\tilde{h}(0)>0$. To find an upper bound for $h(z)$ we use its Taylor expansion, and the fact that $-z \in (0,1)$:
\begin{equation}
    \begin{split}
        {h}(z) 
        &= \sum_{i=0}^\infty \frac{{h}^{(i)}(0) z^i}{i !} = \sum_{i=1}^\infty \frac{(-1)^i z^i}{i } \leq -z + \frac{z^2}{2} + \frac{1}{3} \sum_{i=3}^\infty (-z)^i \\
        &= -z + \frac{z^2}{2} + \frac{1}{3} \left ( \frac{1}{1+z} -1+z-z^2\right )\\
        &= -z + \frac{z^2}{2} + \frac{1}{3} \left ( \frac{1-1+z^2-z^2-z^3}{1+z}\right )\\
        &= -z + \frac{z^2}{2} - \frac{1}{3} \left ( \frac{z^3}{1+z}\right ),
    \end{split}
\end{equation}
which prove the upper bound for $z \in (-1, 0)$.
	
We can put everything together and complete the proof, consider $z = \frac{y-x}{x}$ we have then:
\begin{equation}
    \begin{split}
        f(x,y) &= -log \left ( 1+\frac{y-x}{x} \right ) \\
        &\leq
        \begin{cases}
        -\frac{y-x}{x}+\frac{\left(\frac{y-x}{x} \right )^2}{2} & \frac{y-x}{x} \in [0, +\infty)\\
        -\frac{y-x}{x}+ \frac{\left(\frac{y-x}{x} \right )^2}{2} - \frac{1}{3} \left ( \frac{\left(\frac{y-x}{x} \right )^3}{1+\left(\frac{y-x}{x} \right )}\right ) & \frac{y-x}{x} \in (-1, 0)
        \end{cases}\\
        &=
        \begin{cases}
        -\frac{y-x}{x}+\frac{\left(\frac{y-x}{x} \right )^2}{2} & \frac{y-x}{x} \geq 0\\
        -\frac{y-x}{x}+ \frac{\left(\frac{y-x}{x} \right )^2}{2} - \frac{1}{3} \left ( \frac{\left({y-x} \right )^3}{x^3\left ( 1+\left(\frac{y-x}{x} \right ) \right )}\right ) & \frac{y-x}{x} <0
        \end{cases}\\
        &=
        \begin{cases}
        -\frac{y-x}{x}+\frac{(y-x)^2}{2 x^2} & y \geq x\\
        -\frac{y-x}{x}+\frac{(y-x)^2}{2 x^2} + \frac{\left({x-y} \right )^3}{3 x^2 y} & y <x
        \end{cases}.
    \end{split}
\end{equation}
	
\end{proof}

The following corollary applies the above proposition to bound the KL between two distributions.

\begin{corollary}\label{corol:general_KL_bound}
Consider two probability distribution $p, \tilde{p}$ on a finite state space, with $p$ absolutely continuous with respect to $\tilde{p}$, then:
\begin{equation}
\begin{split}
\mathbf{KL}(p(\mathbf{x})||\tilde{p}(\mathbf{x})) &\leq \frac{1}{2}\sum_x p(x) \frac{1}{p(x)^2} \left ( \tilde{p}(x) - p(x) \right )^2 \\
& \quad + \frac{1}{3}\sum_x p(x) \frac{1}{\tilde{p}(x)^3} \left | p(x)-\tilde{p}(x) \right |^3\\.
\end{split}
\end{equation}
\end{corollary}

\begin{proof}
The proof follows from Proposition \ref{prop:log_bound} with the convention $0\log(0)=0$ and $0 \log(0 \slash 0)=0$. Notice that whenever $p(x)=0$ we have $0\log(0)$ and whenever $\tilde{p}(x)=0$ we have $0\log(0 \slash 0)$ from absolute continuity, hence following the convention we can safely remove those $x$'s from the sum and apply our proposition (apply Proposition \ref{prop:log_bound} with $x=p(x)$ and $y=\tilde{p}(x) \in (0,1]$). The bound follows from:
\begin{equation}
    \begin{split}
        \mathbf{KL}(p(\mathbf{x})||\tilde{p}(\mathbf{x})) &= \sum_x p(x) \log \left ( \frac{p(x)}{\tilde{p}(x)}\right ) 
        \leq
        \sum_x p(x) \left [ -\frac{\tilde{p}(x) - p(x)}{p(x)} + \frac{(\tilde{p}(x) - p(x))^2}{2p(x)^2} \right ]\\
        & \quad + \sum_x \mathbb{I}(\tilde{p}(x) < p(x)) p(x) \frac{\left({p(x)-\tilde{p}(x)} \right )^3}{3 p(x)^2 \tilde{p}(x)}\\
        & =
        -\sum_x p(x) \left [ \frac{\tilde{p}(x) - p(x)}{p(x)} \right ]+ \sum_x p(x) \left [\frac{(\tilde{p}(x) - p(x))^2}{2p(x)^2} \right ]\\
        & \quad + \sum_x \mathbb{I}(\tilde{p}(x) < p(x)) p(x) \frac{\left({p(x)-\tilde{p}(x)} \right )^3}{3 p(x)^2 \tilde{p}(x)}\\
        & =
        \sum_x p(x) \left [ \frac{(\tilde{p}(x) - p(x))^2}{2p(x)^2} \right ] + \sum_x \mathbb{I}(\tilde{p}(x) < p(x)) p(x) \frac{\left({p(x)-\tilde{p}(x)} \right )^3}{3 p(x)^2 \tilde{p}(x)}\\
        & =
        \frac{1}{2}\sum_x p(x) \frac{1}{p(x)^2} \left ( \tilde{p}(x) - p(x) \right )^2 \\
        & \quad + \frac{1}{3}\sum_x \mathbb{I}(\tilde{p}(x) < p(x)) p(x) \frac{1}{p(x)^2 \tilde{p}(x)} \left ( p(x)-\tilde{p}(x) \right )^3\\
        & \leq
        \frac{1}{2}\sum_x p(x) \frac{1}{p(x)^2} \left ( \tilde{p}(x) - p(x) \right )^2 \\
        & \quad + \frac{1}{3}\sum_x \mathbb{I}(\tilde{p}(x) < p(x)) p(x) \frac{1}{\tilde{p}(x)^3} \left ( p(x)-\tilde{p}(x) \right )^3\\
        & \leq
        \frac{1}{2}\sum_x p(x) \frac{1}{p(x)^2} \left ( \tilde{p}(x) - p(x) \right )^2 \\
        & \quad + \frac{1}{3}\sum_x p(x) \frac{1}{\tilde{p}(x)^3} \left | p(x)-\tilde{p}(x) \right |^3\\
    \end{split}
\end{equation}
\end{proof}

We now provide the full proof of Theorem 1. We start by proving absolute continuity and so the validity of the KL computation. We then move to the Data processing inequality which we use to switch from a KL between $ p (y_{[T]}^n )$ and $\tilde{p} (y_{[T]}^n )$ to a KL between $ p (x_{[0:T]}, y_{[T]}^n )$ and $\tilde{p} (x_{[0:T]}, y_{[T]}^n )$. After some reformulation of the quantities, we can apply Jensen inequality and Corollary \ref{corol:general_KL_bound} which gives us a bound in terms of transition kernel difference and ratios. The final step is then the application of the assumptions along with recognizing the definition of the variance.

\begin{proof}[proof of Theorem 1]
Remark that we want to compute the KL-divergence between:
\begin{itemize}
    \item $p(y_{[T]}^n ) =  \sum_{x_{[0:T]}} p(x_{[0:T]}^{\setminus n} ) p(x_{[0:T]}^{n}|x_{[0:T]}^{\setminus n} ) \prod_{t \in [T]} p(y_t^n|x_t^n )$;
    \item $\tilde{p}(y_{[T]}^n ) = \sum_{x_{[0:T]}} p(x_{[0:T]}^{\setminus n} ) p(x_0^{n}| \theta) \prod_{t \in [T]} p(x_t^n|x_{t-1} ) p(y_t^n|x_t^n )$;
\end{itemize}
hence the first step is to prove that $p(y^n_{[T]} )$ is absolutely continuous with respect to $\tilde{p}(y^n_{[T]} )$, and so for a fixed $y_{[T]}^n$ we have $\tilde{p}(y^n_{[T]} )=0$ implies $p(y^n_{[T]} )=0$, this is necessary to ensure that the KL divergence is well-defined. Note that we have $\tilde{p}(y^n_{[T]} )=0$ if and only if for any $x_{[0:T]}$:
\begin{enumerate}
    \item $p(x_{[0:T]}^{\setminus n} )=0$ or \slash and, 
    \item $p(x_0^{n}| \theta) \prod_{t \in [T]} p(x_t^n|x_{t-1} )=0$ or \slash and, 
    \item $\prod_{t \in [T]} p(y_t^n|x_t^n )=0$.
\end{enumerate}
We can observe that conditions 1. and 3. implies also $p(y^n_{[T]} )=0$, so it is enough to prove that 2. implies $p(y^n_{[T]} )=0$ to ensure absolute continuity. Consider $p(x_{[0:T]}^n|x_{[0:T]}^{\setminus n},\theta)$, given that from \eqref{eq:recursive_conditional_joint} we have:
\begin{equation}
    p(x_{[0:T]}^n|x_{[0:T]}^{\setminus n},\theta) = p(x_0^n ) \prod_{t \in [T]} p(x_t^n|x_{t-1},\theta) f \left (x_{t-1}^n, x_{[0:t]}^{\setminus n} \right ),
\end{equation}
we can conclude that 2. implies $p(x_{[0:T]}^n|x_{[0:T]}^{\setminus n},\theta)=0$. Note that the proof of absolute continuity relies on the joint $p(x_{[0:T]}, y^n_{[T]} )$ being absolutely continuous with respect to $\tilde{p}(x_{[0:T]}, y^n_{[T]} )$. Given that the KL is well-defined we can now proceed with the proof.
	
Start by applying Lemma \ref{lemma:data_proc_ineq} and get:
\begin{equation}
    \begin{split}
        \mathbf{KL}(p(\mathbf{y}_{[T]}^n)||\tilde{p}(\mathbf{y}_{[T]}^n))
        &\leq
        \mathbf{KL}({p}(\mathbf{x}_{[0:T]}, \mathbf{y}^n_{[T]} )||\tilde{p}(\mathbf{x}_{[0:T]}, \mathbf{y}^n_{[T]} ))
        \\&=
        \sum_{y^n_{[T]}} \sum_{x_{[0:T]}} {p}(x_{[0:T]}, y^n_{[T]} ) \log \left ( \frac{{p}(x_{[0:T]}, y^n_{[T]} )}{\tilde{p}(x_{[0:T]}, y^n_{[T]} )} \right ).
    \end{split}
\end{equation}
Use now the definition of $\tilde{p}(x_{[0:T]}, y^n_{[T]} )$:
\begin{equation}
    \begin{split}
        &\mathbf{KL}(p(\mathbf{y}_{[T]}^n)||\tilde{p}(\mathbf{y}_{[T]}^n))\leq
        \sum_{y^n_{[T]}} \sum_{x_{[0:T]}} {p}(x_{[0:T]}, y^n_{[T]} ) \log \left ( \frac{{p}(x_{[0:T]}, y^n_{[T]} )}{\tilde{p}(x_{[0:T]}, y^n_{[T]} )} \right )\\
        &=
        \sum_{y^n_{[T]}} \sum_{x_{[0:T]}} p(x_{[0:T]}^{\setminus n} ) p(x_{[0:T]}^{n}|x_{[0:T]}^{\setminus n} ) \prod_{t \in [T]} p(y_t^n|x_t^n )\\
        &\qquad \qquad \log \left ( \frac{p(x_{[0:T]}^{\setminus n} ) p(x_{[0:T]}^{n}|x_{[0:T]}^{\setminus n} ) \prod_{t \in [T]} p(y_t^n|x_t^n )}{p(x_{[0:T]}^{\setminus n} ) p(x_0^{n} ) \prod_{t \in [T]} p(x_t^n|x_{t-1} ) p(y_t^n|x_t^n )} \right )\\
        &=
         \sum_{x_{[0:T]}} p(x_{[0:T]}^{\setminus n} ) p(x_{[0:T]}^{n}|x_{[0:T]}^{\setminus n} ) \sum_{y^n_{[T]}} \prod_{t \in [T]} p(y_t^n|x_t^n ) \log \left ( \frac{p(x_{[0:T]}^{\setminus n} ) p(x_{[0:T]}^{n}|x_{[0:T]}^{\setminus n} )}{p(x_{[0:T]}^{\setminus n} ) p(x_0^{n} ) \prod_{t \in [T]} p(x_t^n|x_{t-1} )} \right )
    \end{split}
\end{equation}
from which we can also simplify $p(x_{[0:T]}^{\setminus n}$ and given that $\sum_{y^n_{[T]}} \prod_{t \in [T]} p(y_t^n|x_t^n )=1$ we conclude:
\begin{equation}
    \begin{split}
        &\mathbf{KL}(p(\mathbf{y}_{[T]}^n)||\tilde{p}(\mathbf{y}_{[T]}^n))\\
        &\leq
         \sum_{x_{[0:T]}} p(x_{[0:T]}^{\setminus n} ) p(x_{[0:T]}^{n}|x_{[0:T]}^{\setminus n} ) \sum_{y^n_{[T]}} \prod_{t \in [T]} p(y_t^n|x_t^n ) \log \left ( \frac{p(x_{[0:T]}^{\setminus n} ) p(x_{[0:T]}^{n}|x_{[0:T]}^{\setminus n} )}{p(x_{[0:T]}^{\setminus n} ) p(x_0^{n} ) \prod_{t \in [T]} p(x_t^n|x_{t-1} )} \right )\\
        &=
        \sum_{x_{[0:T]}} p(x_{[0:T]}^{\setminus n} ) p(x_{[0:T]}^{n}|x_{[0:T]}^{\setminus n} ) \log \left ( \frac{p(x_{[0:T]}^{n}|x_{[0:T]}^{\setminus n} )}{ p(x_0^{n} ) \prod_{t \in [T]} p(x_t^n|x_{t-1} )} \right ).
    \end{split}
\end{equation}
Now we can use the recursive definition of $p(x_{[0:T]}^{n}|x_{[0:T]}^{\setminus n} )$ in \eqref{eq:recursive_conditional_joint}:
\begin{equation}
    \begin{split}
        &\mathbf{KL}(p(\mathbf{y}_{[T]}^n)||\tilde{p}(\mathbf{y}_{[T]}^n))\leq
        \sum_{x_{[0:T]}} p(x_{[0:T]}^{\setminus n} ) p(x_{[0:T]}^{n}|x_{[0:T]}^{\setminus n} ) \log \left ( \frac{p(x_{[0:T]}^{n}|x_{[0:T]}^{\setminus n} ) }{p(x_0^{n} ) \prod_{t \in [T]} p(x_t^n|x_{t-1} )} \right )\\
        &=
        \sum_{x_{[0:T]}} p(x_{[0:T]} ) \log \left ( \frac{p(x_{[0:T]}^{n}|x_{[0:T]}^{\setminus n} ) }{p(x_0^{n} ) \prod_{t \in [T]} p(x_t^n|x_{t-1} )} \right )\\
        &=
        \sum_{x_{[0:T]}} p(x_{[0:T]} ) \log \left ( \frac{ p(x_0^{n} ) \prod_{t \in [T]} p(x_t^n|x_{t-1} ) f \left (x_{t-1}^n, x_{[0:t]}^{\setminus n} \right )}{p(x_0^{n} ) \prod_{t \in [T]} p(x_t^n|x_{t-1} )} \right )\\
        &=
        \sum_{x_{[0:T]}} p(x_{[0:T]} ) \log \left ( \prod_{t \in [T]} f \left (x_{t-1}^n, x_{[0:t]}^{\setminus n} \right ) \right )=
        \sum_{x_{[0:T]}} p(x_{[0:T]} ) \sum_{t \in [T]}\log \left ( f \left (x_{t-1}^n, x_{[0:t]}^{\setminus n} \right ) \right ).
    \end{split}
\end{equation}
Note now that we can move the sum over time steps in front and, given that the argument of the logarithm depends only on $x_{[0:t-1]},x_t^{\setminus n}$ we can also simplify $p(x_{[0:T]} )$:
\begin{equation}\label{eq:before_jensen}
    \begin{split}
        &\mathbf{KL}(p(\mathbf{y}_{[T]}^n)||\tilde{p}(\mathbf{y}_{[T]}^n))\leq
        \sum_{x_{[0:T]}} p(x_{[0:T]} ) \sum_{t \in [T]}\log \left ( f \left (x_{t-1}^n, x_{[0:t]}^{\setminus n} \right ) \right )\\
        &=
        \sum_{t \in [T]} \sum_{x_{[0:t-1]}, x_t^{\setminus n}} p(x_{[0:t-1]} )\prod_{\bar{n} \in [N], \bar{n} \neq n} p(x_{t}^{\bar{n}}|x_{t-1} )\\
        & \qquad \qquad \log \left ( \frac{ \prod_{\bar{n} \in [N], \bar{n} \neq n} p(x_{t}^{\bar{n}}|x_{t-1} )}{\sum_{\bar{x}_{t-1}^n} \prod_{\bar{n} \in [N], \bar{n} \neq n} p(x_{t}^{\bar{n}}|\bar{x}_{t-1} ) p(\bar{x}_{t-1}^n|x_{[0:t-1]}^{\setminus n} )} \right ),\\
    \end{split}
\end{equation}
where $\bar{x}_{t-1}$ is such that $\bar{x}_{t-1}^{\setminus n}= x_{t-1}^{\setminus n}$ and $\bar{x}_{t-1}^n \neq x_{t-1}^n$.

Given that $-log(x)$ is a convex function and the denominator of our logarithm is an expectation we can apply Jensen inequality:
\begin{equation}
    \begin{split}
        &\log \left ( \frac{ \prod_{\bar{n} \in [N], \bar{n} \neq n} p(x_{t}^{\bar{n}}|x_{t-1} )}{\sum_{\bar{x}_{t-1}^n} \prod_{\bar{n} \in [N], \bar{n} \neq n} p(x_{t}^{\bar{n}}|\bar{x}_{t-1} ) p(\bar{x}_{t-1}^n|x_{[0:t-1]}^{\setminus n} )} \right )\\
        &=
        -\log \left ( \sum_{\bar{x}_{t-1}^n} p(\bar{x}_{t-1}^n|x_{[0:t-1]}^{\setminus n} )	\frac{ \prod_{\bar{n} \in [N], \bar{n} \neq n} p(x_{t}^{\bar{n}}|\bar{x}_{t-1} ) }{ \prod_{\bar{n} \in [N], \bar{n} \neq n} p(x_{t}^{\bar{n}}|x_{t-1} )} \right )\\
        &\leq  -\sum_{\bar{x}_{t-1}^n} p(\bar{x}_{t-1}^n|x_{[0:t-1]}^{\setminus n} ) 
        \log \left ( \prod_{\bar{n} \in [N], \bar{n} \neq n} \frac{  p(x_{t}^{\bar{n}}|\bar{x}_{t-1} ) }{ p(x_{t}^{\bar{n}}|x_{t-1} )} \right ) \\
        &=
        \sum_{\bar{x}_{t-1}^n} p(\bar{x}_{t-1}^n|x_{[0:t-1]}^{\setminus n} )  
        \log \left ( \prod_{\bar{n} \in [N], \bar{n} \neq n} \frac{ p(x_{t}^{\bar{n}}|x_{t-1} )}{  p(x_{t}^{\bar{n}}|\bar{x}_{t-1} ) } \right ).
    \end{split}
\end{equation}
Joining \eqref{eq:before_jensen} with the above yields to:
\begin{equation}
    \begin{split}
        &\mathbf{KL}(p(\mathbf{y}_{[T]}^n)||\tilde{p}(\mathbf{y}_{[T]}^n)) \leq
        \sum_{t \in [T]} \sum_{x_{[0:t-1]},x_t^{\setminus n}} p(x_{[0:t-1]},x_t^{\setminus n} ) \left [  \sum_{\bar{x}_{t-1}^n} p(\bar{x}_{t-1}^n|x_{[0:t-1]}^{\setminus n} )  \right.\\
        & \left. \qquad \qquad \qquad \qquad \qquad \qquad \qquad \qquad
        \log \left ( \prod_{\bar{n} \in [N], \bar{n} \neq n} \frac{ p(x_{t}^{\bar{n}}|x_{t-1} )}{  p(x_{t}^{\bar{n}}|\bar{x}_{t-1} ) } \right ) \right ]\\
        &=
        \sum_{t \in [T]} \sum_{x_{[0:t-1]},x_t^{\setminus n}} p(x_{[0:t-1]},x_t^{\setminus n} ) \left [  \sum_{\bar{x}_{t-1}^n} p(\bar{x}_{t-1}^n|x_{[0:t-1]}^{\setminus n} ) 
        \sum_{\bar{n} \in [N], \bar{n} \neq n} \log \left (  \frac{ p(x_{t}^{\bar{n}}|x_{t-1} )}{  p(x_{t}^{\bar{n}}|\bar{x}_{t-1} ) } \right ) \right ].\\
    \end{split}
\end{equation}
Similarly to what we did with the time steps sum, we can move the sum over the dimensions $\bar{n}$ and exchange the order with $\sum_{x_t^{\setminus n}}$, this will allow us to have a dependence on $\bar{n}$ and so simplify $p(x_t^{\setminus n}|x_{t-1})$ to $p(x_t^{\bar{n}}|x_{t-1})$:
\begin{equation}
    \begin{split}
    &\mathbf{KL}(p(\mathbf{y}_{[T]}^n)||\tilde{p}(\mathbf{y}_{[T]}^n))\\
    &\leq
    \sum_{t \in [T]} \sum_{x_{[0:t-1]},x_t^{\setminus n}} p(x_{[0:t-1]},x_t^{\setminus n} ) \left [  \sum_{\bar{x}_{t-1}^n} p(\bar{x}_{t-1}^n|x_{[0:t-1]}^{\setminus n} ) \sum_{\bar{n} \in [N], \bar{n} \neq n} \log \left (  \frac{ p(x_{t}^{\bar{n}}|x_{t-1} )}{  p(x_{t}^{\bar{n}}|\bar{x}_{t-1} ) } \right ) \right ]\\
    &=
    \sum_{t \in [T]} \sum_{x_{[0:t-1]},x_t^{\setminus n}} p(x_{[0:t-1]} ) p(x_t^{\setminus n}|x_{t-1} ) \left [  \sum_{\bar{x}_{t-1}^n} p(\bar{x}_{t-1}^n|x_{[0:t-1]}^{\setminus n} ) \sum_{\bar{n} \in [N], \bar{n} \neq n} \log \left (  \frac{ p(x_{t}^{\bar{n}}|x_{t-1} )}{  p(x_{t}^{\bar{n}}|\bar{x}_{t-1} ) } \right ) \right ]\\
    &=
    \sum_{t \in [T]} \sum_{x_{[0:t-1]}} \sum_{\bar{x}_{t-1}^n} p(\bar{x}_{t-1}^n|x_{[0:t-1]}^{\setminus n} ) p(x_{[0:t-1]} )  \left [ \sum_{x_t^{\setminus n}} p(x_t^{\setminus n}|x_{t-1} ) \sum_{\bar{n} \in [N], \bar{n} \neq n} \log \left (  \frac{ p(x_{t}^{\bar{n}}|x_{t-1} )}{  p(x_{t}^{\bar{n}}|\bar{x}_{t-1} ) } \right ) \right ]\\
    &=
    \sum_{t \in [T]} \sum_{x_{[0:t-1]}} \sum_{\bar{x}_{t-1}^n} p(\bar{x}_{t-1}^n|x_{[0:t-1]}^{\setminus n} ) p(x_{[0:t-1]} ) \sum_{\bar{n} \in [N], \bar{n} \neq n} \left [ \sum_{x_t^{\bar{n}}} p(x_t^{\bar{n}}|x_{t-1} )
    \log \left (  \frac{ p(x_{t}^{\bar{n}}|x_{t-1} )}{  p(x_{t}^{\bar{n}}|\bar{x}_{t-1} ) } \right ) \right ]\\
    &=
    \sum_{t \in [T]} \sum_{x_{[0:t-1]}} \sum_{\bar{x}_{t-1}^n} p(\bar{x}_{t-1}^n|x_{[0:t-1]}^{\setminus n} ) p(x_{[0:t-1]} ) \sum_{\bar{n} \in [N], \bar{n} \neq n} \mathbf{KL} \left ( p(\mathbf{x}_t^{\bar{n}}|\mathbf{x}_{t-1} ) ||
    p(\mathbf{x}_t^{\bar{n}}|\bar{\mathbf{x}}_{t-1} ) \right ).
    \end{split}
\end{equation}

We can now use Corollary \ref{corol:general_KL_bound} and get:
\begin{equation}
    \begin{split}
        &\mathbf{KL}(p(\mathbf{y}_{[T]}^n)||\tilde{p}(\mathbf{y}_{[T]}^n))\\
        &\leq
        \sum_{t \in [T]} \sum_{x_{[0:t-1]}} \sum_{\bar{x}_{t-1}^n} p(\bar{x}_{t-1}^n|x_{[0:t-1]}^{\setminus n} ) p(x_{[0:t-1]} ) \sum_{\bar{n} \in [N], \bar{n} \neq n} \mathbf{KL} \left ( p(\mathbf{x}_t^{\bar{n}}|\mathbf{x}_{t-1} ) ||
        p(\mathbf{x}_t^{\bar{n}}|\bar{\mathbf{x}}_{t-1} ) \right )\\
        &\leq
        \sum_{t \in [T]} \sum_{x_{[0:t-1]}} \sum_{\bar{x}_{t-1}^n} p(\bar{x}_{t-1}^n|x_{[0:t-1]}^{\setminus n} ) p(x_{[0:t-1]} ) \\
        & \qquad
        \sum_{\bar{n} \in [N], \bar{n} \neq n} \frac{1}{2}\left [ \sum_{x_t^{\bar{n}}} p(x_t^{\bar{n}}|x_{t-1} ) \frac{1}{p(x_t^{\bar{n}}|x_{t-1} )^2} \left ( p(x_t^{\bar{n}}|x_{t-1} )  - p(x_t^{\bar{n}}|\bar{x}_{t-1} ) \right )^2 \right ]\\
        & \qquad\qquad\qquad
        + \frac{1}{3}\left [ \sum_{x_t^{\bar{n}}} p(x_t^{\bar{n}}|x_{t-1} ) \frac{1}{p(x_t^{\bar{n}}|\bar{x}_{t-1} )^3} \left | p(x_t^{\bar{n}}|x_{t-1} )  - p(x_t^{\bar{n}}|\bar{x}_{t-1} ) \right |^3 \right ].
    \end{split}
\end{equation}

We can then bound the differences between the kernels with our Assumption \ref{ass:kernel_app}:
\begin{equation}
    \begin{split}
    &\mathbf{KL}(p(\mathbf{y}_{[T]}^n)||\tilde{p}(\mathbf{y}_{[T]}^n))\\
    &\leq
    \sum_{t \in [T]} \sum_{x_{[0:t-1]}} \sum_{\bar{x}_{t-1}^n} p(\bar{x}_{t-1}^n|x_{[0:t-1]}^{\setminus n} ) p(x_{[0:t-1]} ) \\
    & \qquad
    \sum_{\bar{n} \in [N], \bar{n} \neq n} \frac{1}{2}\left [ \sum_{x_t^{\bar{n}}} p(x_t^{\bar{n}}|x_{t-1} ) \frac{1}{p(x_t^{\bar{n}}|x_{t-1} )^2} \left ( p(x_t^{\bar{n}}|x_{t-1} )  - p(x_t^{\bar{n}}|\bar{x}_{t-1} ) \right )^2 \right ]\\
    & \qquad\qquad\qquad
    + \frac{1}{3}\left [ \sum_{x_t^{\bar{n}}} p(x_t^{\bar{n}}|x_{t-1} ) \frac{1}{p(x_t^{\bar{n}}|\bar{x}_{t-1} )^3} \left | p(x_t^{\bar{n}}|x_{t-1} )  - p(x_t^{\bar{n}}|\bar{x}_{t-1} ) \right |^3 \right ]\\
    &\leq
    \sum_{t \in [T]} \sum_{x_{[0:t-1]}} \sum_{\bar{x}_{t-1}^n} p(\bar{x}_{t-1}^n|x_{[0:t-1]}^{\setminus n} ) p(x_{[0:t-1]} ) \\
    & \qquad
    \sum_{\bar{n} \in [N], \bar{n} \neq n} \frac{1}{2}\left [ \sum_{x_t^{\bar{n}}} p(x_t^{\bar{n}}|x_{t-1} ) \frac{1}{p(x_t^{\bar{n}}|x_{t-1} )^2} \frac{1}{N^2} \left | d_{n,\bar{n}} \left ( x_{t-1}^{n} \right )  - d_{n,\bar{n}} \left ( \bar{x}_{t-1}^{n} \right ) \right |^2\right ]\\
    & \qquad
    + \frac{1}{3}\left [ \sum_{x_t^{\bar{n}}} p(x_t^{\bar{n}}|x_{t-1} ) \frac{1}{p(x_t^{\bar{n}}|\bar{x}_{t-1} )^3} \frac{1}{N^3} \left | d_{n,\bar{n}} \left ( x_{t-1}^{n} \right )  - d_{n,\bar{n}} \left ( \bar{x}_{t-1}^{n} \right ) \right |^3 \right ],
    \end{split}
\end{equation} 
and Assumption \ref{ass:boundeness_app}:
\begin{equation}
    \begin{split}
        &\mathbf{KL}(p(\mathbf{y}_{[T]}^n)||\tilde{p}(\mathbf{y}_{[T]}^n))\\
        &\leq
        \sum_{t \in [T]} \sum_{x_{[0:t-1]}} \sum_{\bar{x}_{t-1}^n} p(\bar{x}_{t-1}^n|x_{[0:t-1]}^{\setminus n} ) p(x_{[0:t-1]} ) \\
        & \qquad
        \sum_{\bar{n} \in [N], \bar{n} \neq n} \frac{1}{2} \frac{1}{N^2} \left | d_{n,\bar{n}} \left ( x_{t-1}^{n} \right )  - d_{n,\bar{n}} \left ( \bar{x}_{t-1}^{n} \right ) \right |^2 \left [ \sum_{x_t^{\bar{n}}} p(x_t^{\bar{n}}|x_{t-1} ) \frac{1}{p(x_t^{\bar{n}}|x_{t-1} )^2} \right ]\\
        & \qquad
        + \frac{1}{3}\frac{1}{N^3} \left | d_{n,\bar{n}} \left ( x_{t-1}^{n} \right )  - d_{n,\bar{n}} \left ( \bar{x}_{t-1}^{n} \right ) \right |^3 \left [ \sum_{x_t^{\bar{n}}} p(x_t^{\bar{n}}|x_{t-1} ) \frac{1}{p(x_t^{\bar{n}}|\bar{x}_{t-1} )^3}  \right ]\\
        &\leq
        \sum_{t \in [T]} \sum_{x_{[0:t-1]}} \sum_{\bar{x}_{t-1}^n} p(\bar{x}_{t-1}^n|x_{[0:t-1]}^{\setminus n} ) p(x_{[0:t-1]} ) \\
        & \qquad
        \sum_{\bar{n} \in [N], \bar{n} \neq n} \frac{1}{2} \frac{1}{\epsilon^2 N^2} \left | d_{n,\bar{n}} \left ( x_{t-1}^{n} \right )  - d_{n,\bar{n}} \left ( \bar{x}_{t-1}^{n} \right ) \right |^2
        + \frac{1}{3}\frac{1}{\epsilon^3 N^3} \left | d_{n,\bar{n}} \left ( x_{t-1}^{n} \right )  - d_{n,\bar{n}} \left ( \bar{x}_{t-1}^{n} \right ) \right |^3\\
        &\leq
        \sum_{t \in [T]} \sum_{x_{[0:t-1]}} \sum_{\bar{x}_{t-1}^n} p(\bar{x}_{t-1}^n|x_{[0:t-1]}^{\setminus n} ) p(x_{[0:t-1]} ) \\
        & \qquad
        \sum_{\bar{n} \in [N], \bar{n} \neq n} \left [ \frac{1}{2} \frac{1}{\epsilon^2 N} 
        + \frac{1}{3}\frac{1}{\epsilon^3 N} \right ] \frac{1}{N}\left | d_{n,\bar{n}} \left ( x_{t-1}^{n} \right )  - d_{n,\bar{n}} \left ( \bar{x}_{t-1}^{n} \right ) \right |^2\\
    \end{split}
\end{equation} 
where in the last step we used $\left | d_{n,\bar{n}} \left ( x_{t-1}^{n} \right )  - d_{n,\bar{n}} \left ( \bar{x}_{t-1}^{n} \right ) \right | < N$. We can then reformulate:
\begin{equation}
    \begin{split}
        &\mathbf{KL}(p(\mathbf{y}_{[T]}^n)||\tilde{p}(\mathbf{y}_{[T]}^n))\\
        &\leq
        \sum_{t \in [T]} \sum_{x_{[0:t-1]}} \sum_{\bar{x}_{t-1}^n} p(\bar{x}_{t-1}^n|x_{[0:t-1]}^{\setminus n} ) p(x_{[0:t-1]} ) \\
        & \qquad
        \sum_{\bar{n} \in [N], \bar{n} \neq n} \left [ \frac{1}{2} \frac{1}{\epsilon^2 N} 
        + \frac{1}{3}\frac{1}{\epsilon^3 N} \right ] \frac{1}{N}\left | d_{n,\bar{n}} \left ( x_{t-1}^{n} \right )  - d_{n,\bar{n}} \left ( \bar{x}_{t-1}^{n} \right ) \right |^2\\
        &\leq
        \left [ \frac{1}{2} \frac{1}{\epsilon^2 N} 
        + \frac{1}{3}\frac{1}{\epsilon^3 N} \right ] \sum_{t \in [T]} \sum_{\bar{n} \in [N], \bar{n} \neq n} \frac{1}{N} \\
        & \qquad
        \sum_{x_{[0:t-1]}} \sum_{\bar{x}_{t-1}^n} p(\bar{x}_{t-1}^n|x_{[0:t-1]}^{\setminus n} ) p(x_{[0:t-1]} )\left | d_{n,\bar{n}} \left ( x_{t-1}^{n} \right )  - d_{n,\bar{n}} \left ( \bar{x}_{t-1}^{n} \right ) \right |^2.\\
    \end{split}
\end{equation}  
Now we need to work on:
\begin{equation}
    \sum_{x_{[0:t-1]}} \sum_{\bar{x}_{t-1}^n} p(\bar{x}_{t-1}^n|x_{[0:t-1]}^{\setminus n} ) p(x_{[0:t-1]} ) \left | d_{n,\bar{n}} \left ( x_{t-1}^{n} \right )  - d_{n,\bar{n}} \left ( \bar{x}_{t-1}^{n} \right ) \right |^2,
\end{equation}
which can be reformulated by expanding the square:
\begin{equation}
    \sum_{x_{[0:t-1]}} \sum_{\bar{x}_{t-1}^n} p(\bar{x}_{t-1}^n|x_{[0:t-1]}^{\setminus n} ) p(x_{[0:t-1]} ) \left ( d_{n,\bar{n}} \left ( x_{t-1}^{n} \right )^2 - 2 d_{n,\bar{n}} \left ( x_{t-1}^{n} \right ) d_{n,\bar{n}} \left ( \bar{x}_{t-1}^{n} \right ) + d_{n,\bar{n}} \left ( \bar{x}_{t-1}^{n} \right )^2 \right ),
\end{equation}
and we can work on three terms inside the parenthesis separately. Firstly, we have:
\begin{equation}
    \begin{split}
        \sum_{x_{[0:t-1]}} \sum_{\bar{x}_{t-1}^n} p(\bar{x}_{t-1}^n|x_{[0:t-1]}^{\setminus n} ) p(x_{[0:t-1]} ) d_{n,\bar{n}} \left ( x_{t-1}^{n} \right )^2&= \mathbb{E}\left [ d_{n,\bar{n}} \left ( \mathbf{x}_{t-1}^{n} \right )^2 \right ]
        \\
        &= \mathbb{E} \left \{ \mathbb{E}\left [ d_{n,\bar{n}} \left ( \mathbf{x}_{t-1}^{n} \right )^2 \Big{|} \mathbf{x}_{[0:t-1]}^{\setminus n} \right ] \right \},
    \end{split}
\end{equation}
as we have no dependence on $\bar{x}_{t-1}^n$ and where the last step follows from the law of total expectations. Now we look at the other squared term:
\begin{equation}
    \begin{split}
        &\sum_{x_{[0:t-1]}} \sum_{\bar{x}_{t-1}^n} p(\bar{x}_{t-1}^n|x_{[0:t-1]}^{\setminus n} ) p(x_{[0:t-1]} ) d_{n,\bar{n}} \left ( \bar{x}_{t-1}^{n} \right )^2\\
        &=
        \sum_{x_{[0:t-1]}^{\setminus n}} \sum_{\bar{x}_{t-1}^n} p(\bar{x}_{t-1}^n|x_{[0:t-1]}^{\setminus n} ) p(x_{[0:t-1]}^{\setminus n} ) d_{n,\bar{n}} \left ( \bar{x}_{t-1}^{n} \right )^2 = \mathbb{E} \left \{ \mathbb{E} \left [ d_{n,\bar{n}} \left ( \mathbf{x}_{t-1}^{n} \right )^2 \Big{|} \mathbf{x}_{[0:t-1]}^{\setminus n} \right ] \right \},
    \end{split}    
\end{equation}
as again we have no dependence on $x_{[0:t-1]}^{n}$ and so we can marginilise them out from $p(x_{[0:t-1]} )$. There is only the cross term left, and try to reformulate it as an expectation:
\begin{equation}
    \begin{split}
        &\sum_{x_{[0:t-1]}} \sum_{\bar{x}_{t-1}^n} p(\bar{x}_{t-1}^n|x_{[0:t-1]}^{\setminus n} ) p(x_{[0:t-1]} ) d_{n,\bar{n}} \left ( x_{t-1}^{n} \right ) d_{n,\bar{n}} \left ( \bar{x}_{t-1}^{n} \right )\\
        &=
        \sum_{x_{[0:t-2]}} \sum_{x_{t-1}^{\setminus n}} \sum_{x_{t-1}^n} \sum_{\bar{x}_{t-1}^n} p(x_{[0:t-1]} ) p(\bar{x}_{t-1}^n|x_{[0:t-1]}^{\setminus n} ) d_{n,\bar{n}} \left ( x_{t-1}^{n} \right ) d_{n,\bar{n}} \left ( \bar{x}_{t-1}^{n} \right )\\
        &=
        \sum_{x_{[0:t-2]}^{\setminus n}} \sum_{x_{t-1}^{\setminus n}} \sum_{x_{t-1}^n} \sum_{\bar{x}_{t-1}^n} p(x_{[0:t-2]}^{\setminus n}, x_{t-1} ) p(\bar{x}_{t-1}^n|x_{[0:t-1]}^{\setminus n} ) d_{n,\bar{n}} \left ( x_{t-1}^{n} \right ) d_{n,\bar{n}} \left ( \bar{x}_{t-1}^{n} \right )\\
        &=
        \sum_{x_{[0:t-1]}^{\setminus n}} p(x_{[0:t-1]}^{\setminus n} ) \sum_{x_{t-1}^n} \sum_{\bar{x}_{t-1}^n} p(x_{t-1}^n | x_{[0:t-1]}^{\setminus n} ) p(\bar{x}_{t-1}^n|x_{[0:t-1]}^{\setminus n} ) d_{n,\bar{n}} \left ( x_{t-1}^{n} \right ) d_{n,\bar{n}} \left ( \bar{x}_{t-1}^{n} \right )\\
        &=
        \sum_{x_{[0:t-1]}^{\setminus n}} p(x_{[0:t-1]}^{\setminus n} ) \left [ \sum_{x_{t-1}^n} p(x_{t-1}^n | x_{[0:t-1]}^{\setminus n} ) d_{n,\bar{n}} \left ( x_{t-1}^{n} \right ) \right ] \left [ \sum_{\bar{x}_{t-1}^n} p(\bar{x}_{t-1}^n|x_{[0:t-1]}^{\setminus n} )  d_{n,\bar{n}} \left ( \bar{x}_{t-1}^{n} \right ) \right ]\\
        &=
        \sum_{x_{[0:t-1]}^{\setminus n}} p(x_{[0:t-1]}^{\setminus n} ) \mathbb{E} \left [ d_{n,\bar{n}} \left ( \mathbf{x}_{t-1}^{n} \right ) \Big{|} \mathbf{x}_{[0:t-1]}^{\setminus n} \right ]^2
        =
        \mathbb{E} \left \{
        \mathbb{E} \left [d_{n,\bar{n}} \left ( \mathbf{x}_{t-1}^{n} \right ) \Big{|} \mathbf{x}_{[0:t-1]}^{\setminus n} \right ]^2 \right \}.
    \end{split}
\end{equation}
From the three reformulations above it follows that:
\begin{equation}
    \begin{split}
        &\mathbf{KL}(p(\mathbf{y}_{[T]}^n)||\tilde{p}(\mathbf{y}_{[T]}^n))\\
        &\leq
        \left [ \frac{1}{2 \epsilon^2 N} 
        + \frac{1}{3 \epsilon^3 N} \right ]\\
        & \qquad\sum_{t \in [T]} 
        \sum_{\bar{n} \in [N], \bar{n} \neq n} \frac{1}{N} \sum_{x_{[0:t-1]}} \sum_{\bar{x}_{t-1}^n} p(\bar{x}_{t-1}^n|x_{[0:t-1]}^{\setminus n} ) p(x_{[0:t-1]} ) \left | d_{n,\bar{n}} \left ( x_{t-1}^{n} \right )  - d_{n,\bar{n}} \left ( \bar{x}_{t-1}^{n} \right ) \right |^2\\
        &\leq
        \left [ \frac{1}{2 \epsilon^2 N} 
        + \frac{1}{3 \epsilon^3 N} \right ]\\
        & \qquad\sum_{t \in [T]} 
        \sum_{\bar{n} \in [N], \bar{n} \neq n} \frac{1}{N} \left [ 
        2 \mathbb{E} \left \{ \mathbb{E} \left [ d_{n,\bar{n}} \left ( \mathbf{x}_{t-1}^{n} \right )^2 \Big{|} \mathbf{x}_{[0:t-1]}^{\setminus n} \right ] \right \}
        - 2 \mathbb{E} \left \{
        \mathbb{E} \left [d_{n,\bar{n}} \left ( \mathbf{x}_{t-1}^{n} \right ) \Big{|} \mathbf{x}_{[0:t-1]}^{\setminus n} \right ]^2 \right \} \right ]\\
        &\leq
        2\left [ \frac{1}{2 \epsilon^2 N} 
        + \frac{1}{3 \epsilon^3 N} \right ] \sum_{\bar{n} \in [N], \bar{n} \neq n} \frac{1}{N} \left [ \sum_{t \in [T]} 
        \mathbb{E} \left \{ \mathbb{V}ar \left [ d_{n,\bar{n}} \left ( \mathbf{x}_{t-1}^{n} \right ) \Big{|} \mathbf{x}_{[0:t-1]}^{\setminus n} \right ] \right \} \right ]\\
        &\leq
        2\left [ \frac{1}{2 \epsilon^2 N} 
        + \frac{1}{3 \epsilon^3 N} \right ] \sum_{t \in [T]} 
        \mathbb{E} \left \{ \frac{1}{N} \sum_{\bar{n} \in [N], \bar{n} \neq n} \mathbb{V}ar \left [ d_{n,\bar{n}} \left ( \mathbf{x}_{t-1}^{n} \right ) \Big{|} \mathbf{x}_{[0:t-1]}^{\setminus n} \right ] \right \},\\
    \end{split}
\end{equation} 
which completes the proof.

\end{proof}

\subsection{KL bounds for SimBa-CL for general partitions}

As explained in the main paper we can generalize SimBa-CL with and without feedback to work on different factorizations of the dimensions. Precisely, if we have a general partition $\mathcal{K}$ of $[N]$, then we can define for each $K \in \mathcal{K}$:
\begin{align}
    & p\left (y_{[T]}^K\right ) \coloneqq  \sum_{x_{[0:T]}} p\left (x_{[0:T]}^{\setminus K}\right ) p\left (x_{[0:T]}^{K}|x_{[0:T]}^{\setminus K}\right ) \prod_{t \in [T]} \prod_{n \in K} p\left (y_t^n|x_t^n\right );  \label{eq:general_part_with_feed}\\
    & \tilde{p}_\mathcal{K}\left (y_{[T]}^K\right ) \coloneqq \sum_{x_{[0:T]}} p\left (x_{[0:T]}^{\setminus K}\right ) \prod_{n \in K} p\left (x_0^{n}\right ) \prod_{t \in [T]} p\left (x_t^n|x_{t-1}\right ) p\left (y_t^n|x_t^n\right ),\label{eq:general_part_without_feed_app}
\end{align}
and then combine them as $p_\mathcal{K}(y_{[T]} ) \coloneqq  \prod_{K \in \mathcal{K}} p\left (y_{[T]}^K\right )$ and as $\tilde{p}_\mathcal{K}(y_{[T]} ) \coloneqq \prod_{K \in \mathcal{K}} \tilde{p}_\mathcal{K}\left (y_{[T]}^K\right )$ to provide an approximation of our likelihood $p(y_{[T]})$. In this case, we are computing the marginals on subsets of the state-space, which can be particularly suited when we expect strong dependence inside the elements of the partition, e.g. an agent-based model with individuals partitioned in households. 

Again $p_\mathcal{K}(y_{[T]} )$ can be seen as a SimBa-CL with feedback, where the likelihood is approximated with the product of the true marginals, and $\tilde{p}_\mathcal{K}(y_{[T]} )$ can be seen as a SimBa-CL without feedback. As for the fully factorize case, we have a recursive formula for $p(x_{[0:T]}^K|x_{[0:T]}^{\setminus K} )$ given by:
\begin{equation}\label{eq:recursive_conditional_joint_partition}
    \begin{split}
        p(x_{[0:T]}^K|x_{[0:T]}^{\setminus K} ) 
        &=
        \prod_{n \in K} p(x_{0}^n ) \prod_{t \in [T]} p(x_{t}^n|x_{t-1} ) f_K (x_{t-1}^K, x_{[0:t]}^{\setminus K}),
    \end{split}
\end{equation}
where the simulation feedback is now:
\begin{equation}\label{eq:sim_feed_gen_part}
    \begin{split}
        f_K (x_{t-1}^K, x_{[0:t]}^{\setminus K})
        \coloneqq 
        \frac{ \prod_{\bar{n} \in [N] \setminus K} p(x_{t}^{\bar{n}}|x_{t-1} )}{\sum_{\bar{x}_{t-1}^K} \prod_{\bar{n} \in [N] \setminus K} p(x_{t}^{\bar{n}}|\bar{x}_{t-1}^K, x_{[t-1]}^{\setminus K} ) p(\bar{x}_{t-1}^K|x_{0}^{\setminus K}, x_{[t-1]}^{\setminus K} )},
    \end{split}
\end{equation}

As for SimBa-CL with and without feedback, with these general terms referring to the fully factorized case, we can assess the degradation of approximation quality when excluding the feedback from the procedure. Naturally, we need a more general version of the previous assumptions, given that we are now dealing with changes across multiple dimensions. The assumptions have already been reported before and can be found in Assumption \ref{ass:general_kernel_ass} and Assumption \ref{ass:general_boundeness_ass}, however, we report them again here and comment on the interpretation.

\begin{assumption}\label{ass:general_kernel}
For any $K \in \mathcal{K}$ and for any $x_t^{\bar{n}} \in \mathcal{X}$ with ${\bar{n}} \notin K$, if $x_{t-1}, \bar{x}_{t-1} \in \mathcal{X}^N$ are such that $x_{t-1}^{\setminus {K}}= \bar{x}_{t-1}^{\setminus {K}}$ then:
    \begin{equation}
        \left | p\left (x_t^{\bar{n}}|x_{t-1}\right ) - p\left (x_t^{\bar{n}}|\bar{x}_{t-1}\right ) \right | \leq \frac{1}{N} \left |d_{K,\bar{n}}\left (x_{t-1}^{n}\right ) - d_{K,\bar{n}}\left (\bar{x}_{t-1}^{n}\right ) \right |.
    \end{equation}
where $d_{K,\bar{n}}:\mathcal{X}^K \to \mathbb{R}_+$.
\end{assumption}

\begin{assumption}\label{ass:general_boundeness}
    For any $K, \bar{K} \in \mathcal{K}$, if $x_{t-1}, \bar{x}_{t-1} \in \mathcal{X}^N$ are such that $x_{t-1}^{\setminus \bar{K}}= \bar{x}_{t-1}^{\setminus \bar{K}}$ then there exists $ 0 < \epsilon < 1$ such that:
    \begin{equation}
        \sum_{x^{{n}}_t} p\left (x_t^{{n}}|x_{t-1}\right ) \frac{1}{p\left (x_t^{n}|\bar{x}_{t-1}\right )^2} \leq \frac{1}{\epsilon^2}, \quad \text{and} \quad
        \sum_{x^{{n}}_t} p\left (x_t^{{n}}|x_{t-1}\right ) \frac{1}{p\left (x_t^{n}|\bar{x}_{t-1}\right )^3} \leq \frac{1}{\epsilon^3}.
    \end{equation}
\end{assumption}

Assumption \ref{ass:general_kernel} and Assumption \ref{ass:general_boundeness} provide the required generalizations. Assumption \ref{ass:general_kernel}  is measuring the interactions on $K \in \mathcal{K}$ and ensuring that changes on $K$ have a limited impact on the dynamic outside $K$. Meanwhile, Assumption \ref{ass:general_boundeness} formulates a stronger version of Assumption \ref{ass:boundeness_app}. 

As for the fully factorized case, given the assumptions, we can bound the $\mathbf{KL}$-divergence between our SimBa-CL with and without feedback on general partitions. As it can be noticed from the statement of Theorem \ref{thm:general_KL_bound}, there are a couple of deviations from the fully factorize case.

\begin{theorem}\label{thm:general_KL_bound}
If $\left | d_{K,\bar{n}} \left ( x^{K} \right )  - d_{K,\bar{n}} \left ( \bar{x}^{K} \right ) \right | < N$ for any $x^K,\bar{x}^K \in \mathcal{X}^K$ and for any $K \in \mathcal{K}$, and assumptions \ref{ass:general_boundeness}-\ref{ass:general_kernel} hold, then for any $K \in \mathcal{K}$:
    \begin{equation}
        \mathbf{KL} \left [ p\left (\mathbf{y}_{[T]}^K\right )||\tilde{p}_{\mathcal{K}}\left (\mathbf{y}_{[T]}^K\right ) \right ] \leq 
        \frac{a(\epsilon)}{N} \sum_{t \in [T]} 
        \mathbb{E} \left \{ \frac{1}{N} \sum_{\bar{n} \in [N] \setminus K } \mathbb{V}ar \left [ d_{K,\bar{n}} \left ( \mathbf{x}_{t-1}^{K} \right ) \Big{|} \mathbf{x}_{[0:t-1]}^{\setminus K} \right ] \right \},
    \end{equation}
 where $a(\epsilon) \coloneqq 2\left [ \frac{1}{2 \epsilon^2} + \frac{1}{3 \epsilon^3} \right ]$.
\end{theorem}

Once more, we can interpret the bound as the one from Theorem 1, suggesting to exclude the feedback whenever the dimension $N$ is large and when the process noise is contained.

Given assumptions \ref{ass:general_kernel_ass}-\ref{ass:general_boundeness_ass} we can prove bounds on the KL-divergence between the with feedback case and the without feedback case.

\begin{proof}[proof of Theorem 5]
The proof of Theorem 2 follows the same steps as the proof of Theorem 1, with $K$ instead of $n$. However, the last steps are slightly different due to a different assumption to apply. We report them here for completeness.

Precisely after proving that our KL is well-defined, applying the data processing inequality, exploiting the factorizations in the model, and applying Jensen, we reach the point of using Corollary \ref{corol:general_KL_bound} and get:
\begin{equation}
    \begin{split}
        &\mathbf{KL}(p(\mathbf{y}_{[T]}^K)||\tilde{p}(\mathbf{y}_{[T]}^K))\\
        &\leq
        \sum_{t \in [T]} \sum_{x_{[0:t-1]}} \sum_{\bar{x}_{t-1}^K} p(\bar{x}_{t-1}^K|x_{[0:t-1]}^{\setminus K} ) p(x_{[0:t-1]} ) \sum_{\bar{n} \in [N] \setminus K} \mathbf{KL} \left ( p(\mathbf{x}_t^{\bar{n}}|\mathbf{x}_{t-1} ) ||
        p(\mathbf{x}_t^{\bar{n}}|\bar{\mathbf{x}}_{t-1} ) \right )\\
        &\leq
        \sum_{t \in [T]} \sum_{x_{[0:t-1]}} \sum_{\bar{x}_{t-1}^K} p(\bar{x}_{t-1}^K|x_{[0:t-1]}^{\setminus K} ) p(x_{[0:t-1]} ) \\
        & \qquad
        \sum_{\bar{n} \in [N] \setminus K} \frac{1}{2}\left [ \sum_{x_t^{\bar{n}}} p(x_t^{\bar{n}}|x_{t-1} ) \frac{1}{p(x_t^{\bar{n}}|x_{t-1} )^2} \left ( p(x_t^{\bar{n}}|x_{t-1} )  - p(x_t^{\bar{n}}|\bar{x}_{t-1} ) \right )^2 \right ]\\
        & \qquad\qquad\qquad
        + \frac{1}{3}\left [ \sum_{x_t^{\bar{n}}} p(x_t^{\bar{n}}|x_{t-1} ) \frac{1}{p(x_t^{\bar{n}}|\bar{x}_{t-1} )^3} \left | p(x_t^{\bar{n}}|x_{t-1} )  - p(x_t^{\bar{n}}|\bar{x}_{t-1} ) \right |^3 \right ].
    \end{split}
\end{equation}

We can then bound the differences between the kernels with our Assumption \ref{ass:general_kernel_ass}:
\begin{equation}
    \begin{split}
    &\mathbf{KL}(p(\mathbf{y}_{[T]}^K)||\tilde{p}(\mathbf{y}_{[T]}^K))\\
    &\leq
    \sum_{t \in [T]} \sum_{x_{[0:t-1]}} \sum_{\bar{x}_{t-1}^K} p(\bar{x}_{t-1}^K|x_{[0:t-1]}^{\setminus K} ) p(x_{[0:t-1]} ) \\
    & \qquad
    \sum_{\bar{n} \in [N] \setminus K} \frac{1}{2}\left [ \sum_{x_t^{\bar{n}}} p(x_t^{\bar{n}}|x_{t-1} ) \frac{1}{p(x_t^{\bar{n}}|x_{t-1} )^2} \frac{1}{N^2} \left | d_{K,\bar{n}} \left ( x_{t-1}^{K} \right )  - d_{K,\bar{n}} \left ( \bar{x}_{t-1}^{K} \right ) \right |^2\right ]\\
    & \qquad
    + \frac{1}{3}\left [ \sum_{x_t^{\bar{n}}} p(x_t^{\bar{n}}|x_{t-1} ) \frac{1}{p(x_t^{\bar{n}}|\bar{x}_{t-1} )^3} \frac{1}{N^3} \left | d_{K,\bar{n}} \left ( x_{t-1}^{K} \right )  - d_{K,\bar{n}} \left ( \bar{x}_{t-1}^{K} \right ) \right |^3 \right ],
    \end{split}
\end{equation} 
and Assumption \ref{ass:general_boundeness_ass}:
\begin{equation}
    \begin{split}
        &\mathbf{KL}(p(\mathbf{y}_{[T]}^K)||\tilde{p}(\mathbf{y}_{[T]}^K))\\
        &\leq
        \sum_{t \in [T]} \sum_{x_{[0:t-1]}} \sum_{\bar{x}_{t-1}^K} p(\bar{x}_{t-1}^K|x_{[0:t-1]}^{\setminus K} ) p(x_{[0:t-1]} ) \\
        & \qquad
        \sum_{\bar{n} \in [N] \setminus K} \frac{1}{2} \frac{1}{\epsilon^2 N^2} \left | d_{K,\bar{n}} \left ( x_{t-1}^{K} \right )  - d_{K,\bar{n}} \left ( \bar{x}_{t-1}^{K} \right ) \right |^2
        + \frac{1}{3}\frac{1}{\epsilon^3 N^3} \left | d_{K,\bar{n}} \left ( x_{t-1}^{K} \right )  - d_{K,\bar{n}} \left ( \bar{x}_{t-1}^{K} \right ) \right |^3\\
        &\leq
        \sum_{t \in [T]} \sum_{x_{[0:t-1]}} \sum_{\bar{x}_{t-1}^K} p(\bar{x}_{t-1}^K|x_{[0:t-1]}^{\setminus K} ) p(x_{[0:t-1]} ) \\
        & \qquad
        \sum_{\bar{n} \in [N] \setminus K} \left [ \frac{1}{2} \frac{1}{\epsilon^2 N} 
        + \frac{1}{3}\frac{1}{\epsilon^3 N} \right ] \frac{1}{N}\left | d_{K,\bar{n}} \left ( x_{t-1}^{K} \right )  - d_{K,\bar{n}} \left ( \bar{x}_{t-1}^{K} \right ) \right |^2\\
    \end{split}
\end{equation} 
where in the last step we used $\left | d_{K,\bar{n}} \left ( x_{t-1}^{K} \right )  - d_{K,\bar{n}} \left ( \bar{x}_{t-1}^{K} \right ) \right | < N$. Similarly to the proof of Theorem 1, we can reformulate in terms of expected variance but this time we will work on $d_{K, \bar{n}}(\mathbf{x}_{t-1}^K)$ and so:
\begin{equation}
    \begin{split}
        &\mathbf{KL}(p(\mathbf{y}_{[T]}^K)||\tilde{p}(\mathbf{y}_{[T]}^K))\\
        &\leq
        \left [ \frac{1}{2 \epsilon^2 N} 
        + \frac{1}{3 \epsilon^3 N} \right ]\\
        & \qquad\sum_{t \in [T]} 
        \sum_{\bar{n} \in [N] \setminus K} \frac{1}{N} \sum_{x_{[0:t-1]}} \sum_{\bar{x}_{t-1}^K} p(\bar{x}_{t-1}^K|x_{[0:t-1]}^{\setminus K} ) p(x_{[0:t-1]} ) \left | d_{K,\bar{n}} \left ( x_{t-1}^{K} \right )  - d_{K,\bar{n}} \left ( \bar{x}_{t-1}^{K} \right ) \right |^2\\
        &\leq
        \left [ \frac{1}{2 \epsilon^2 N} 
        + \frac{1}{3 \epsilon^3 N} \right ]\\
        & \qquad\sum_{t \in [T]} 
        \sum_{\bar{n} \in [N] \setminus K} \frac{1}{N} \left [ 
        2 \mathbb{E} \left \{ \mathbb{E} \left [ d_{K,\bar{n}} \left ( \mathbf{x}_{t-1}^{K} \right )^2 \Big{|} \mathbf{x}_{[0:t-1]}^{\setminus K} \right ] \right \}
        - 2 \mathbb{E} \left \{
        \mathbb{E} \left [d_{K,\bar{n}} \left ( \mathbf{x}_{t-1}^{K} \right ) \Big{|} \mathbf{x}_{[0:t-1]}^{\setminus K} \right ]^2 \right \} \right ]\\
        &\leq
        2\left [ \frac{1}{2 \epsilon^2 N} 
        + \frac{1}{3 \epsilon^3 N} \right ] \sum_{\bar{n} \in [N] \setminus K} \frac{1}{N} \left [ \sum_{t \in [T]} 
        \mathbb{E} \left \{ \mathbb{V}ar \left [ d_{K,\bar{n}} \left ( \mathbf{x}_{t-1}^{K} \right ) \Big{|} \mathbf{x}_{[0:t-1]}^{\setminus K} \right ] \right \} \right ]\\
        &\leq
        2\left [ \frac{1}{2 \epsilon^2 N} 
        + \frac{1}{3 \epsilon^3 N} \right ] \sum_{t \in [T]} 
        \mathbb{E} \left \{ \frac{1}{N} \sum_{\bar{n} \in [N] \setminus K} \mathbb{V}ar \left [ d_{K,\bar{n}} \left ( \mathbf{x}_{t-1}^{K} \right ) \Big{|} \mathbf{x}_{[0:t-1]}^{\setminus K} \right ] \right \},\\
    \end{split}
\end{equation} 
which completes the proof.

\end{proof}

\section{SimBa-CL as composite likelihood}

Here we discuss further SimBa-CL as a composite likelihood method, with a particular focus on the estimation of the sensitivity matrix and variability matrix and the observed information.

\subsection{Observed information}

In the main paper, we reported $S(\theta)$ and $V\left (\theta\right )$ as our sensitivity matrix and variability matrix, which we defined as:
\begin{equation}
    S\left (\theta\right ) = \mathbb{E}_{\theta} \left \{ -\Hessian_{\theta} \left [ \log \tilde{p}_{\mathcal{K}}\left (\mathbf{y}_{[T]}^K| {\theta} \right ) \right ] \right \} \text{ and }
    V\left (\theta\right ) = \mathbb{V}\text{ar}_{\theta} \left \{ \nabla_{\theta} \left [ \log \tilde{p}_{\mathcal{K}}\left (\mathbf{y}_{[T]}^K| {\theta} \right ) \right ] \right \}.
\end{equation}
This formulation suggests a bootstrap approach where the expectations are estimated by simulating from the model and computing the Hessian and gradient of the composite likelihood in the simulated data accordingly. However, it is often common to use an observed information approach, where the observed data are plugged-in our instead of simulated data. It is worth noting that for the above formulation of variability and specificity matrix, this is not ideal as the observed data are not independent and identically distributed, meaning that we do not have a way to estimate the mean and the variance in the equations above. 

In the main paper we also reported another formulation of the sensitivity matrix and variability matrix, where approximate Bartlett identities are used:
\begin{equation}
\begin{split}
    &S\left (\theta\right ) \approx 
    \sum_{K \in \mathcal{K}} \mathbb{E}_{\theta} \left \{ \nabla_{\theta} \left [ \log \tilde{p}_{\mathcal{K}}\left (\mathbf{y}_{[T]}^K| \theta \right ) \right ] \nabla_{\theta} \left [ \log \tilde{p}_{\mathcal{K}}\left (\mathbf{y}_{[T]}^K| \theta \right ) \right ]^\top \right \},\\
    &V\left (\theta\right ) \approx \sum_{K, \tilde{K} \in \mathcal{K}} \mathbb{E}_{\theta} \left \{ \nabla_{\theta} \left [ 
    \log \tilde{p}_{\mathcal{K}}\left (\mathbf{y}_{[T]}^K| \theta \right ) \right ] \nabla_{\theta} \left [ 
     \log \tilde{p}_{\mathcal{K}}\left (\mathbf{y}_{[T]}^{\tilde{K}}| \theta \right ) \right ]^\top \right \}.
\end{split}
\end{equation}
These alternative formulations are more suited to the observed information approach as we could assume independence across blocks and identical distribution of the blocks and get:
\begin{equation}
\begin{split}
    &S\left (\theta\right ) \approx 
    \sum_{K \in \mathcal{K}} \left \{ \nabla_{\theta} \left [ \log \tilde{p}_{\mathcal{K}}\left (y_{[T]}^K| \theta \right ) \right ] \nabla_{\theta} \left [ \log \tilde{p}_{\mathcal{K}}\left (y_{[T]}^K| \theta \right ) \right ]^\top \right \},\\
    &V\left (\theta\right ) \approx \sum_{K, \tilde{K} \in \mathcal{K}} \left \{ \nabla_{\theta} \left [ 
    \log \tilde{p}_{\mathcal{K}}\left (y_{[T]}^K| \theta \right ) \right ] \nabla_{\theta} \left [ 
     \log \tilde{p}_{\mathcal{K}}\left (y_{[T]}^{\tilde{K}}| \theta \right ) \right ]^\top \right \}.
\end{split}
\end{equation}

It is surely not surprising that this approach will not perform well when the blocks are highly correlated and different in distribution. 

For completeness, we report the experiment table from the main paper where empirical coverage when learning the full model is measured. Here we add an extra row showing the resulting coverage across parameters when using the observed information. As clear from the bad coverage, the observed information might lead to overconfident sets, this is probably due to the challenge of identifying the parameters and so the optimizer gets stuck in local maxima of the likelihood.

\begin{table}[http!]
    \centering
    \begin{tabular}{|l|c|c|c|c|c|}
    \hline
     Parameter & $\beta_0$ & $\beta_\lambda$ & $\beta_\gamma$ & $q$ & $\iota$  \\
    \hline
    Without Bartlett & $0.17 \text{ and }0.05$ & $0.61 \text{ and }0.87$ & $0.8 \text{ and }1.$ & $0.87 \text{ and } 0.5$ & $0.02$\\ 
    With Bartlett & $0.98 \text{ and } 0.89$ & $0.99 \text{ and } 0.75$ & $0.97 \text{ and } 0.97$ & $1. \text{ and } 0.98$ & $0.92$\\
    With Bartlett OI & $0.48 \text{ and } 0.30$ & $0.51 \text{ and } 0.31$ & $0.01 \text{ and } 0.06$ & $0.62 \text{ and } 0.02$ & $0.03$\\
    \hline
    \end{tabular}
    \caption{Empirical coverage per each parameter when computing the Godambe information matrix with and without the approximate Bartlett identities and when using the observed information. Whenever the parameter is bi-dimensional the coverage per each component is reported in the same cell separated by ``$\text{and}$''. ``OI'' refers to observed information.}
    \label{tab:9dim_coverage_app}
\end{table}

\subsection{Estimating the sensitivity and variability matrix}

Bootstrap sampling seems like a good strategy to compute an estimate of the variability and sensitivity matrix, precisely given $P$ samples from the model:
\begin{equation}
    y^{i,}_{[T]} \sim p(y_{[T]}|\hat{\theta}_{CL}) \text{ for } i \in [P],
\end{equation}
where $\hat{\theta}_{CL}$ is our maximum SimBa-CL estimator, can be used to get:
\begin{equation}
\begin{split}
    & S\left (\hat{\theta}_{CL}\right ) \approx \frac{1}{P} \sum_{i \in [P]} \left \{ -\Hessian_{\theta} \left [ \log \tilde{p}_{\mathcal{K}}\left (y_{[T]}^K| \hat{\theta}_{CL} \right ) \right ] \right \} \text{ and }\\
    & V\left (\hat{\theta}_{CL}\right ) \approx \frac{1}{P-1} \sum_{i \in [P]} \left [ \nabla_{\theta} \left [ \log \tilde{p}_{\mathcal{K}}\left (y_{[T]}^K| \hat{\theta}_{CL} \right ) \right ] -  \frac{1}{P} \sum_{i \in [P]} \left \{ \nabla_{\theta} \left [ \log \tilde{p}_{\mathcal{K}}\left (y_{[T]}^K| \hat{\theta}_{CL} \right ) \right ] \right \} \right ]^2.
\end{split}
\end{equation}
Similarly when invoking the approximate first and second Bartlett identities we have:
\begin{equation}
\begin{split}
    &S\left (\hat{\theta}_{CL} \right ) \approx 
    \sum_{K \in \mathcal{K}} \frac{1}{P} \sum_{i \in [P]} \left \{ \nabla_{\theta} \left [ \log \tilde{p}_{\mathcal{K}}\left (y_{[T]}^K| \hat{\theta}_{CL} \right ) \right ] \nabla_{\theta} \left [ \log \tilde{p}_{\mathcal{K}}\left (y_{[T]}^K| \hat{\theta}_{CL} \right ) \right ]^\top \right \},\\
    &V\left (\hat{\theta}_{CL} \right ) \approx \sum_{K, \tilde{K} \in \mathcal{K}} \frac{1}{P} \sum_{i \in [P]} \left \{ \nabla_{\theta} \left [ 
    \log \tilde{p}_{\mathcal{K}}\left (y_{[T]}^K| \hat{\theta}_{CL} \right ) \right ] \nabla_{\theta} \left [ 
     \log \tilde{p}_{\mathcal{K}}\left (y_{[T]}^{\tilde{K}}| \hat{\theta}_{CL} \right ) \right ]^\top \right \}.
\end{split}
\end{equation}
Note that all the differentiation operations can be run in parallel on $P$, making the bootstrap approach relatively cheap when consistent computational resources are available.

\subsection{Approximate first and second Bartlett identities}

In this section, we motivate the use of the Bartlett identities in our computations. The first Bartlett identity states:
\begin{equation}
    \mathbb{E}_{\theta} \left [ \nabla_{\theta} \log \left ( p\left (\mathbf{y}_{[T]}| \theta \right ) \right ) \right ] = 0,
\end{equation}
which we want to discuss for the SimBa-CL case.

For SimBa-CL we have:
\begin{equation}
\begin{split}
    \mathbb{E}_{\theta} \left [ \nabla_{\theta} \log \left ( \tilde{p}_{\mathcal{K}}\left (\mathbf{y}_{[T]}| \theta \right ) \right ) \right ] &= \sum_{K \in \mathcal{K}} \mathbb{E}_{\theta} \left [ \nabla_{\theta} \log \left ( \tilde{p}_{\mathcal{K}}\left (\mathbf{y}^K_{[T]}| \theta \right ) \right ) \right ] = \sum_{K \in \mathcal{K}} \mathbb{E}_{\theta} \left [ \frac{\nabla_{\theta} \left ( \tilde{p}_{\mathcal{K}}\left (\mathbf{y}^K_{[T]}| \theta \right ) \right )}{\tilde{p}_{\mathcal{K}}\left (\mathbf{y}^K_{[T]}| \theta \right )} \right ]\\ 
    &= \sum_{K \in \mathcal{K}} \sum_{y^K_{[T]}} \left [ \nabla_{\theta} \left ( \tilde{p}_{\mathcal{K}}\left ({y}^K_{[T]}| \theta \right ) \right ) \frac{{p}\left ({y}^K_{[T]}| \theta \right )}{\tilde{p}_{\mathcal{K}}\left ({y}^K_{[T]}| \theta \right )} \right ]\\ 
    &\approx \sum_{K \in \mathcal{K}} \sum_{y^K_{[T]}} \left [ \nabla_{\theta} \left ( \tilde{p}_{\mathcal{K}}\left ({y}^K_{[T]}| \theta \right ) \right ) \right ] = 0,\\ 
\end{split}
\end{equation}
where we use the $\tilde{p}_{\mathcal{K}}\left ({y}^K_{[T]}| \theta \right ) \approx {p}\left ({y}^K_{[T]}| \theta \right )$, i.e. a simulation feedback close to $1$, and the fact that $\tilde{p}_{\mathcal{K}}$ is still a proper probability distribution and so sum up to $1$. Obviously, if we consider SimBa-CL with feedback the first Bartlett identity holds exactly cause we are targeting the true marginals.

We now discuss the second Bartlett identity:
\begin{equation}
    \mathbb{E}_{\theta} \left [ \Hessian_{\theta} \log \left ( p\left (\mathbf{y}_{[T]}| \theta \right ) \right ) \right ] = -\mathbb{E}_{\theta}  \left \{ \nabla_{\theta} \left [\log p\left (\mathbf{y}_{[T]}| \theta \right ) \right ] \nabla_{\theta} \left [\log p\left (\mathbf{y}_{[T]}| \theta \right ) \right ]^\top \right \}.
\end{equation}
Note that for SimBa-CL we have:
\begin{equation}
\begin{split}
    &\mathbb{E}_{\theta} \left [ \Hessian_{\theta} \log \left ( \tilde{p}_{\mathcal{K}}\left (\mathbf{y}_{[T]}| \theta \right ) \right ) \right ] = \sum_{K \in \mathcal{K}} \mathbb{E}_{\theta} \left [ \Hessian_{\theta} \log \left ( \tilde{p}_{\mathcal{K}}\left (\mathbf{y}^K_{[T]}| \theta \right ) \right ) \right ] \\
    &= \sum_{K \in \mathcal{K}} \mathbb{E}_{\theta} \left [ \nabla_{\theta} \nabla_{\theta} \log \left ( \tilde{p}_{\mathcal{K}}\left (\mathbf{y}^K_{[T]}| \theta \right ) \right ) \right ] = \sum_{K \in \mathcal{K}} \mathbb{E}_{\theta} \left [  \nabla_{\theta} \left ( \frac{\nabla_{\theta} \tilde{p}_{\mathcal{K}}\left (\mathbf{y}^K_{[T]}| \theta \right )}{\tilde{p}_{\mathcal{K}}\left (\mathbf{y}^K_{[T]}| \theta \right )} \right ) \right ] = \\
    &= \sum_{K \in \mathcal{K}} \mathbb{E}_{\theta} \left [   \left ( \frac{\Hessian_{\theta} \tilde{p}_{\mathcal{K}}\left (\mathbf{y}^K_{[T]}| \theta \right )}{\tilde{p}_{\mathcal{K}}\left (\mathbf{y}^K_{[T]}| \theta \right )} \right ) - \left ( \frac{\nabla_{\theta} \tilde{p}_{\mathcal{K}}\left (\mathbf{y}^K_{[T]}| \theta \right ) \nabla_{\theta} \tilde{p}_{\mathcal{K}}\left (\mathbf{y}^K_{[T]}| \theta \right )^{\top}}{\tilde{p}_{\mathcal{K}}\left (\mathbf{y}^K_{[T]}| \theta \right )^2} \right ) \right ]\\
    &= \sum_{K \in \mathcal{K}} \sum_{y_{[T]}^K} \left [   \left ( \frac{\Hessian_{\theta} \tilde{p}_{\mathcal{K}}\left ({y}^K_{[T]}| \theta \right )}{\tilde{p}_{\mathcal{K}}\left ({y}^K_{[T]}| \theta \right )} \right ) - \left ( \nabla_{\theta} \log \tilde{p}_{\mathcal{K}}\left ({y}^K_{[T]}| \theta \right ) \nabla_{\theta} \log \tilde{p}_{\mathcal{K}}\left ({y}^K_{[T]}| \theta \right )^{\top} \right ) \right ] p(y^K_{[T]}|\theta) \\
    &\approx - \sum_{K \in \mathcal{K}} \sum_{y_{[T]}^K} \nabla_{\theta} \left [ \log \tilde{p}_{\mathcal{K}}\left ({y}^K_{[T]}| \theta \right ) \right ] \nabla_{\theta} \left [ \log \tilde{p}_{\mathcal{K}}\left ({y}^K_{[T]}| \theta \right )\right ]^{\top} p(y^K_{[T]}|\theta)\\
    &= - \sum_{K \in \mathcal{K}} \mathbb{E}_{\theta} \left \{ \nabla_{\theta} \left [ \log \tilde{p}_{\mathcal{K}}\left ({y}^K_{[T]}| \theta \right ) \right ] \nabla_{\theta} \left [ \log \tilde{p}_{\mathcal{K}}\left ({y}^K_{[T]}| \theta \right )\right ]^{\top} \right \},
\end{split}
\end{equation}
where again we used our simulation feedback being approximatively $1$.

Combining the definitions of mean and variance with these approximate Bartlett identities leads to our proposed estimates of the sensitivity and variability matrix.

\section{Experiments}

Here, we provide additional details on the experiments reported in the main paper, as well as some extra experiments that were excluded to keep the paper concise. For experiments from the main paper, the model, model parameters, SimBa-CL versions, and $P$ remain the same as described in the main text.

\subsection{Empirical evaluation of the KL-divergence}

In this section, we expand more on the empirical KL-divergence, which is used as an evaluation metric to measure the similarities between different SimBa-CL.

It is essential to recognise that any distribution $q\left (x\right )$ can be approximated empirically as $q\left (x\right ) \approx \sum_{z \in [E]} q(z) \delta_{z} (x)$, where $E$ corresponds to the number of evaluations of $x$. While this approximation provides a sparse representation of $p\left (x\right )$, it offers nonetheless a means to estimate the KL-divergence as follows:
\begin{equation}\label{eq:empirical_KL}
    \mathbf{KL}\left [ p\left (\mathbf{x}\right ) || q\left (\mathbf{x}\right )\right ] \approx \sum_{e \in [E]} p\left (x^e\right ) \log \left ( \frac{p\left (x^e\right )}{q\left (x^e\right )} \right ).
\end{equation}
We can then simulate multiple times from the model to have some approximate evaluation of the chosen SimBa-CL, and we can then compare different SimBa-CL using the above evaluation metric. 

Note that after simulating multiple times from the model we have a sample per each SimBa-CL and we can create multiple samples by leaving one simulation out at a time. This will allow us to produce a sample of KL-divergences from which we can compute mean and variance. The table with mean and variance for our experiments is reported in the main paper. However, as we are comparing distributions another valid comparison is to graphically look at the boxplots of the log SimBa-CL, which we report here.

Figure \ref{fig:log_like_feeback} depicts a comparison between the logarithm of the fully factorized SimBa-CL with and without feedback. The distribution of the two log-likelihoods appears similar, indicating that excluding the feedback does not result in a significant loss of precision. Figure \ref{fig:log_like_partition} compares the logarithm of the fully factorized SimBa-CL without feedback with the logarithm of the coupled SimBa-CL without feedback. Again we can observe that the two distributions become more and more similar when we increase $N$ and reduce the variance in the system. 

\begin{figure}[http!]
    \centering
    \includegraphics[width=0.9\textwidth]{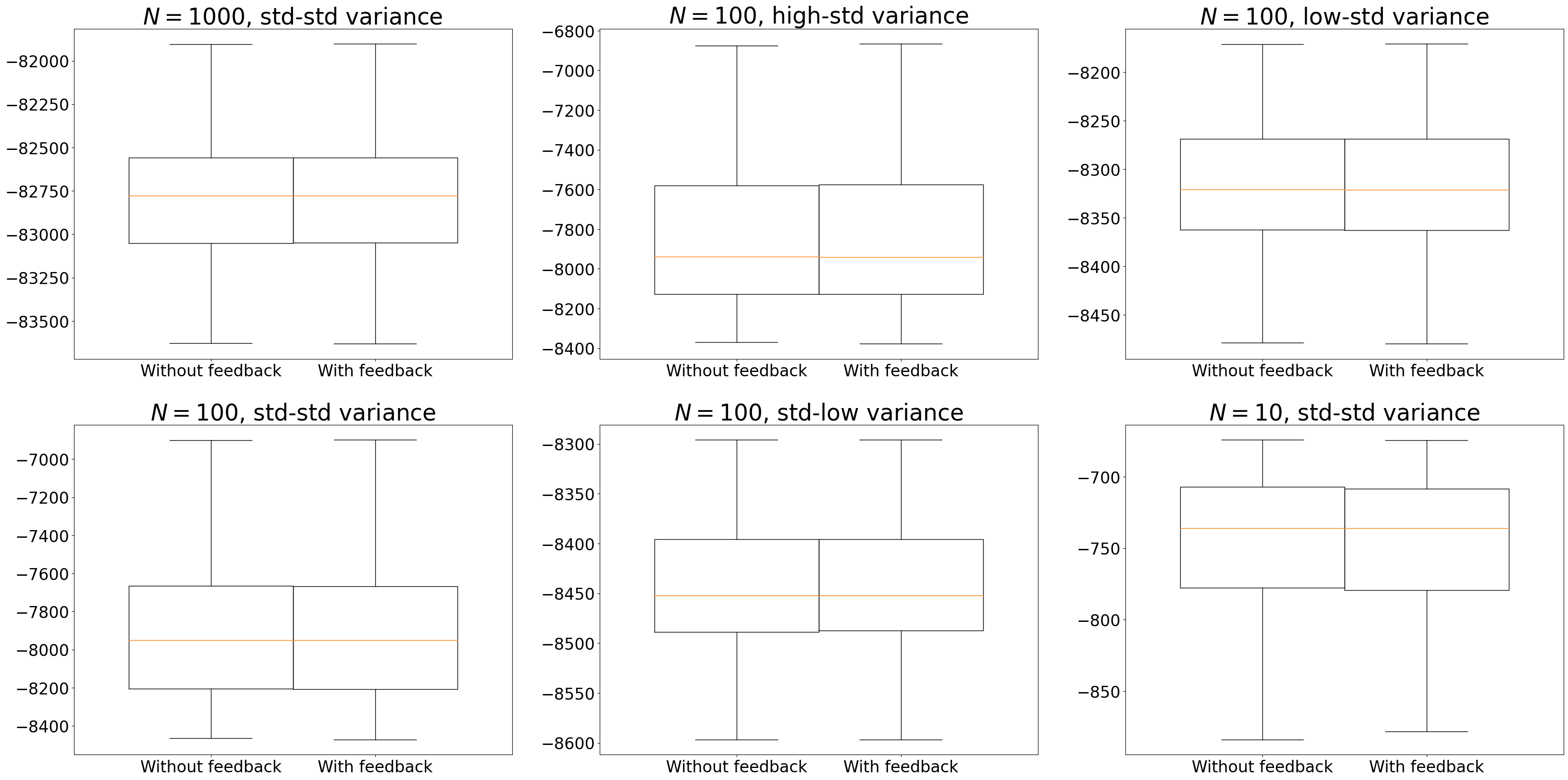}
    \caption{Comparing empirical KL between fully factorized SimBa-CL with feedback and fully factorized SimBa-CL without feedback under different scenarios.}
    \label{fig:log_like_feeback}
\end{figure}

\begin{figure}[http!]
    \centering
    \includegraphics[width=0.9\textwidth]{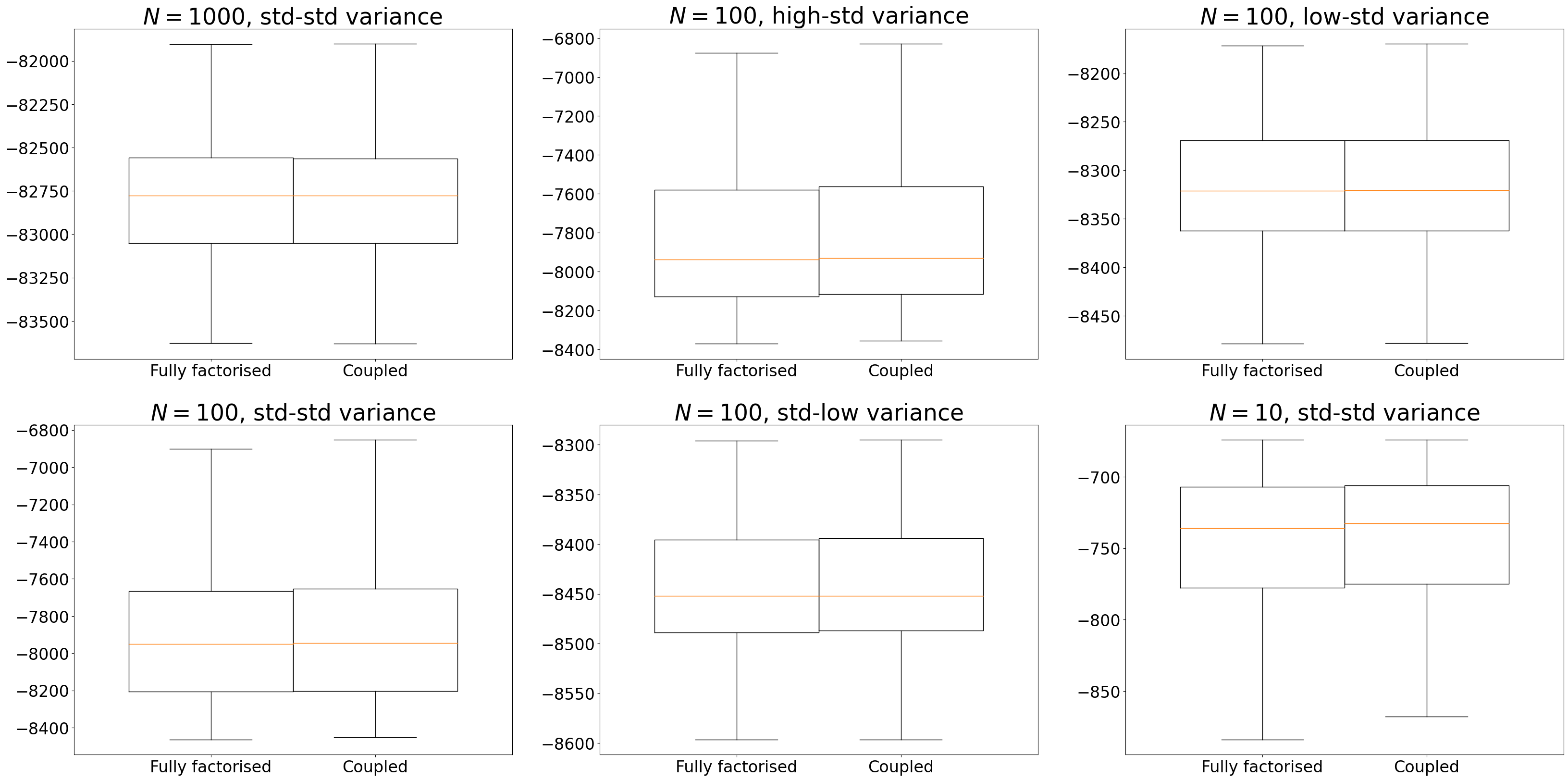}
    \caption{Comparing empirical KL between fully factorized SimBa-CL without feedback and coupled SimBa-CL without feedback under different scenarios.}
    \label{fig:log_like_partition}
\end{figure}

\subsection{Optimization of the parameters}

To optimize the parameters of our models and compute the maximum SimBa-CL estimator we generally used Adam optimizer \citep{kingma2014adam}. Precisely, we provide gradient via automatic differentiation in TensorFlow and use the built-in optimizer. As a loss function we used $-\sum_{n \in [N]} \log(\tilde{p}(y_{[T]}^n|\theta)) \slash (N T)$ with initial learning rate for Adam given by $0.1$. We run $500$ optimization steps for the 9-dimensional synthetic data scenario and $3000$ for the FM data. We check convergence via the stability of the loss function. We also evaluate SimBa-CL using $500$ simulations, but just $100$ for the FM data to meet the memory constraints. For the 2-dimensional synthetic data scenario, we use vanilla gradient descent with learning rate $100$ and $200$ optimization steps, again we check convergence via the stability of the loss function.

Convergence plots are reported below for the different scenarios. We can always notice that the optimization process can find a local minimum of the loss function and that we stop our optimization after convergence is reached.

It is important to comment on the optimization for the FM experiment. In the figure we can spot multiple drops while performing the optimization, these drops refer to the failure of the optimization procedure where the gradient has pushed the parameters in regions with zero likelihood. To solve this issue we reset the optimization by sampling a new random initial condition and restart the optimization from there, resulting in the vertical drops in Figure \ref{fig:FMD_optim}. Even though this might lead to some optimizations not converging, remark that we run 100 optimization procedures in parallel and choose the best one in terms of SimBa-CL score.

\begin{figure}[httb!]
    \centering
    \includegraphics[width=\textwidth]{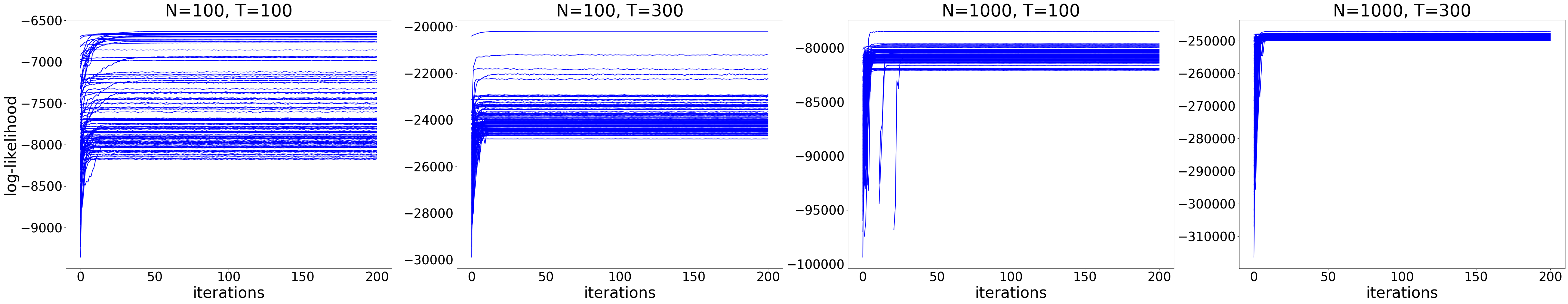}
    \caption{Log-likelihood convergence during the optimization of $\beta_\lambda$, when studying the asymptotic properties of SimBa-CL.}\label{fig:FMD_optim}
\end{figure}

\begin{figure}[httb!]
    \centering
    \includegraphics[width=\textwidth]{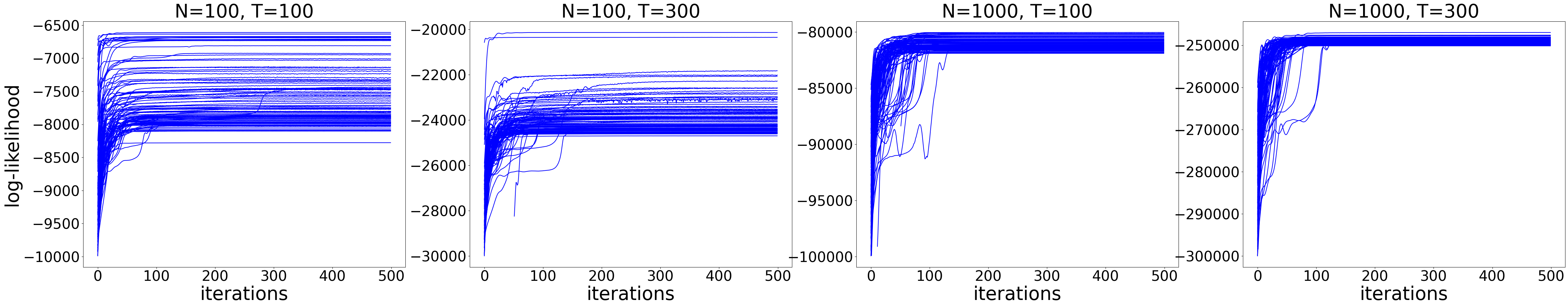}
    \caption{Log-likelihood convergence during the optimization of all the nine parameters, when studying the asymptotic properties of SimBa-CL.}
\end{figure}

\begin{figure}[httb!]
    \centering
    \includegraphics[width=0.75\textwidth]{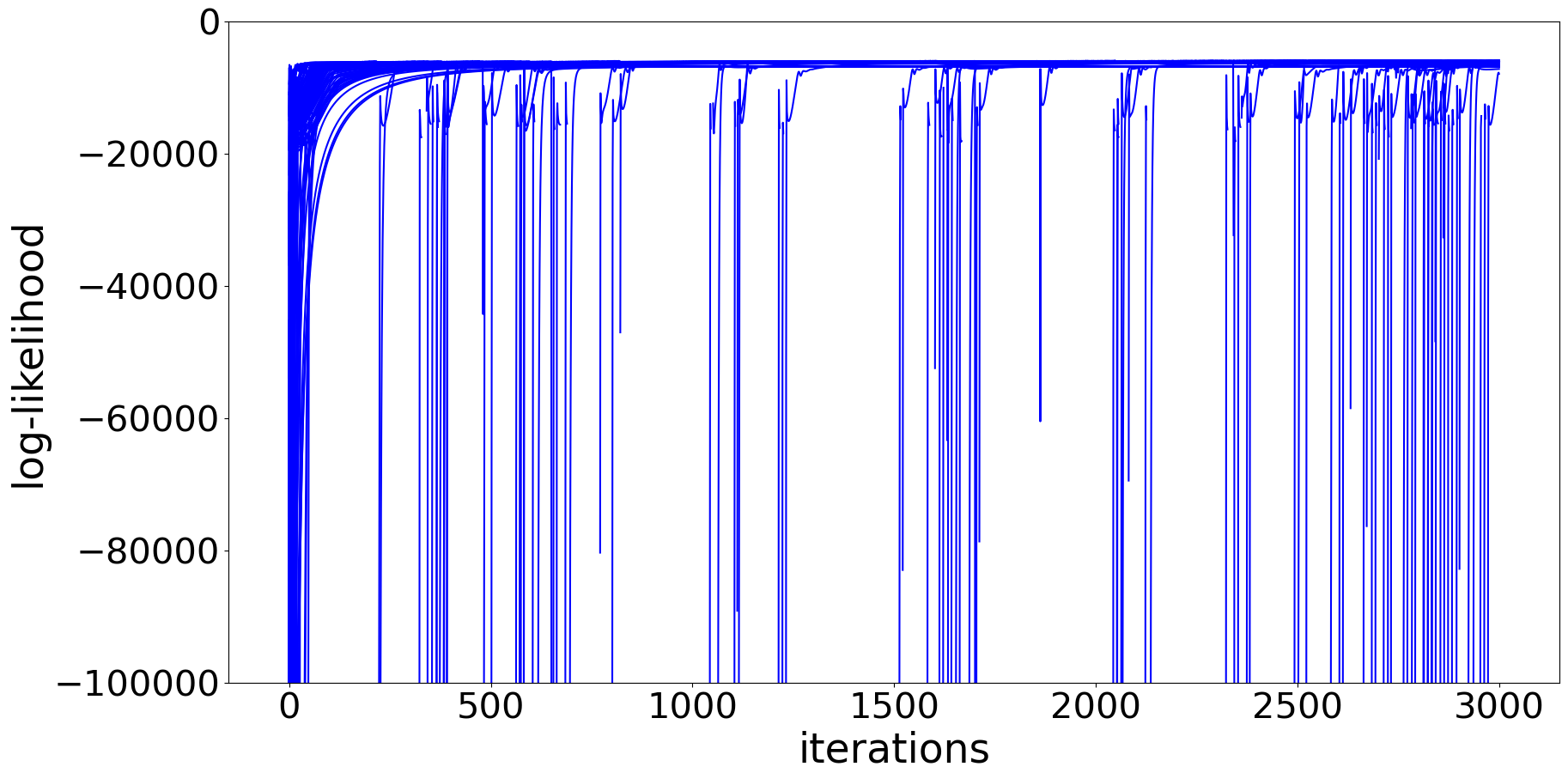}
    \caption{Log-likelihood convergence during the optimization for the FM experiment.}
\end{figure}

\newpage
\subsection{Asymptotic properties of SimBa-CL}

In this section, we provide some additional details on the asymptotic properties of SimBa-CL. Specifically, we provide a graphical representation of the empirical coverage for both the 2-dimensional and 9-dimensional case. For the 2 dimensional case, we also provide a graphical representation of the confidence sets, which are ellipsoids.  

It is obvious from Figure \ref{fig:9dim_coverage} that the noise is dominating our computations, meaning that getting good estimates of the expectation and variance in $S(\theta)$ and $V(\theta)$ is not easy, as they refer to the full space $\mathcal{Y}^{T N}$. However, using the Bartlett identities lead to less noisy estimates as now $S(\theta)$ and $V(\theta)$ are expressed as sums of expectation on the space $\mathcal{Y}^{T}$. This computational trick along with SimBa-CL without feedback being close in the large population limit to SimBa-CL with feedback, where the Bartlett identities hold exactly, can explain the better coverage experienced in Figure \ref{fig:9dim_coverage_Bartlett}. Further theoretical studies are needed to provide a more formal justification.

\begin{figure}[httb!]
    \centering
    \includegraphics[width=0.75\textwidth]{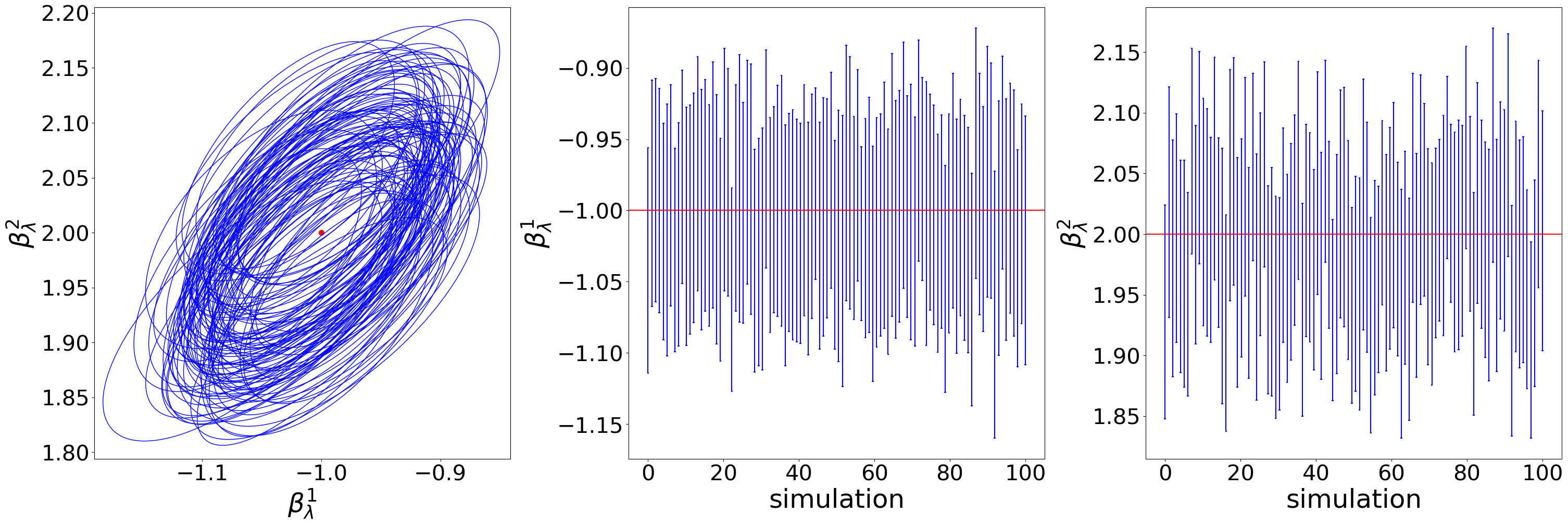}
    \includegraphics[width=0.75\textwidth]{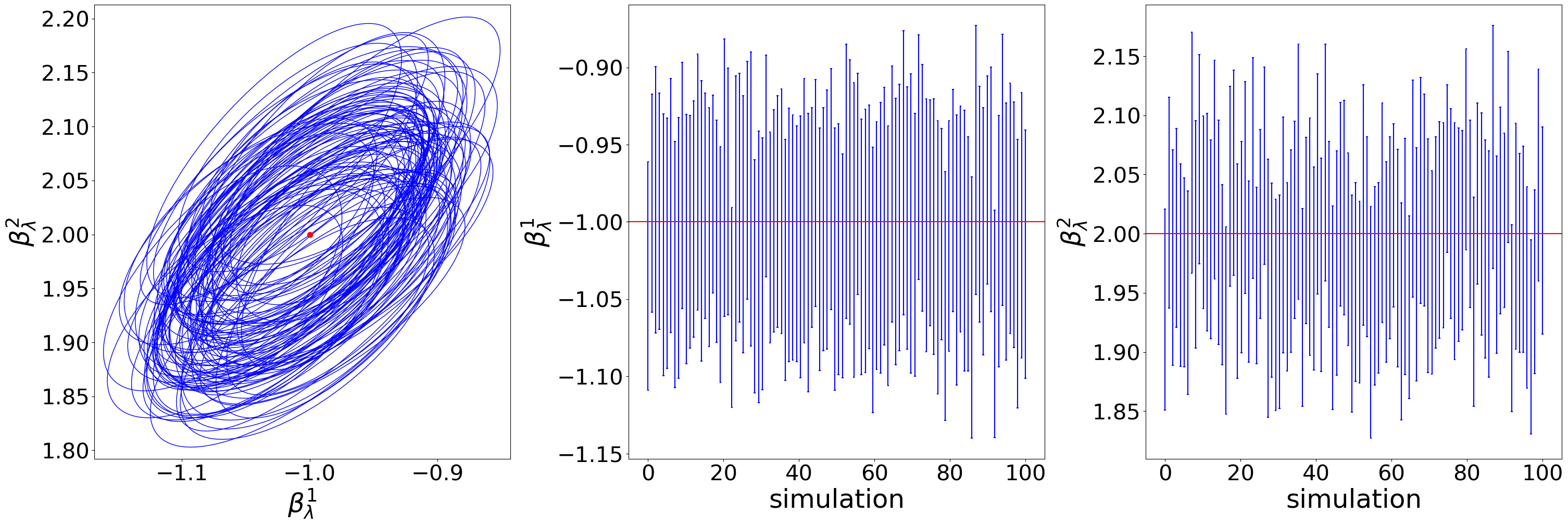}
    \caption{Coverage of the confidence sets with and without Bartlett identities. The first row is for the latter and the second row is for the former. From left to right, graphical coverage of $\beta_\lambda$, $\beta_\lambda^1$ (marginally) and $\beta_\lambda^2$ (marginally). Dots or solid lines are used for the true parameters.}
    \label{fig:coverage_2dim}
\end{figure}
\begin{figure}[httb!]
    \centering
    \includegraphics[width=0.75\textwidth]{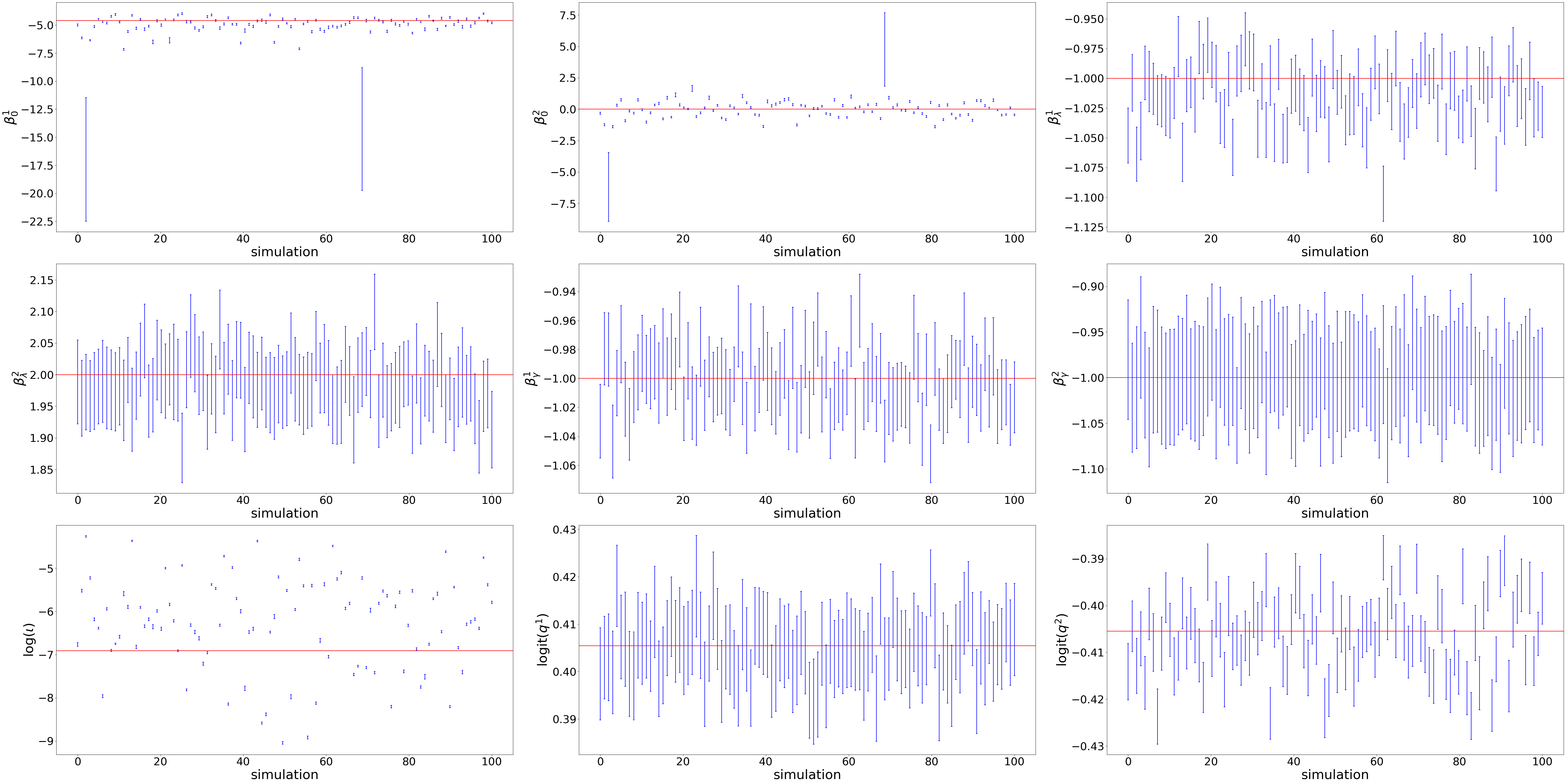}
    \caption{Coverage of the confidence intervals from the diagonal of the Godambe information matrix estimated without the approximate Bartlett identities. Parameter labels are displayed on the x-axis. Solid lines are used for the true parameters.}
    \label{fig:9dim_coverage}
\end{figure}

\begin{figure}[httb!]
    \centering
    \includegraphics[width=0.75\textwidth]{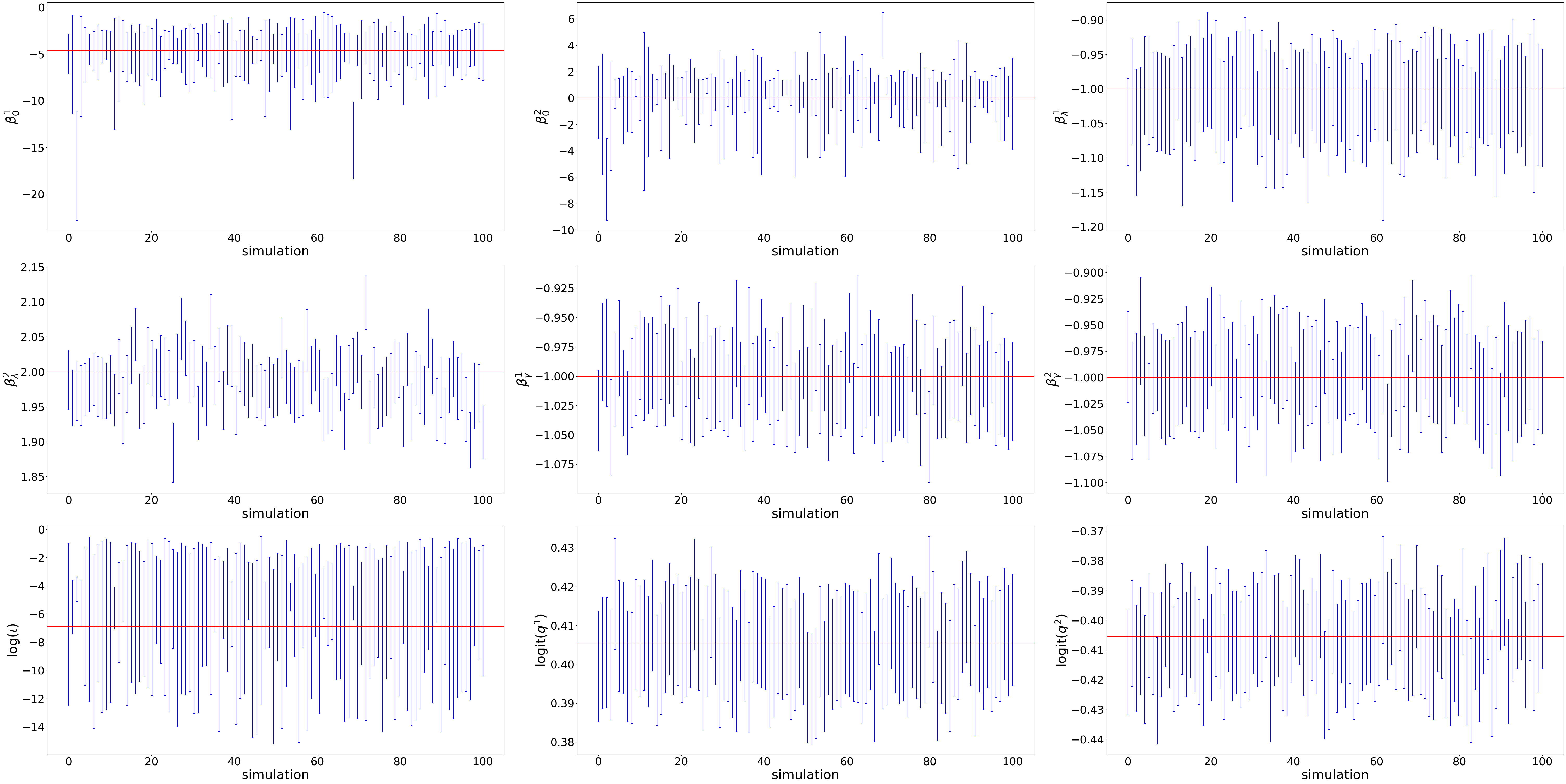}
    \caption{Coverage of the confidence intervals from the diagonal of the Godambe information matrix estimated with the approximate Bartlett identities. Parameter labels are displayed on the x-axis. Solid lines are used for the true parameters.}
    \label{fig:9dim_coverage_Bartlett}
\end{figure}

\newpage
\subsection{Spatial SIS}

Consider an individual-based susceptible-infected-susceptible (SIS) model as the one from previous experiments. Suppose however that the spatial interaction is not homogeneous, and that a spatial kernel is measuring the infection pressure from one individual to the other. We then have an initial probability of infection $\left ( {1} \slash {1+\exp{\left (-\beta_0^\top w_n\right )}} \right )$ and a transition kernel with a probability of transitioning from $S$ to $I$ of:
$$
1- e^{ - \lambda_n \left ( \frac{\sum_{\bar{n} \in [N]}  s(n,\bar{n},\psi) \mathbb{I}\left (x_{t-1}^{\bar{n}}=2\right ) }{N} + \iota \right ) },
$$
and from $I$ to $S$ of $1-e^{- \gamma_n }$, where $\lambda_n = \left ({1}\slash{1+\exp{\left (-\beta_{\lambda}^\top w_n\right )}} \right )$ and $\gamma_n = \left ({1} \slash {1+\exp{\left (-\beta_{\gamma }^\top w_n\right )}} \right )$ and $s(n,\bar{n},\psi)\coloneqq \exp \left ( - E_{n,\bar{n}}^2 \slash (2 \psi^2) \right )$ with $E_{n,\bar{n}}$ euclidean distance between $n$ and $\bar{n}$ and $\psi$ positive parameters. We then consider the same emission distribution of our baseline SIS. For this model we set our baseline to $N=1000$, $T= 100$, $w_n$ to be such that $w_n^1 = 1$ and $w_n^2 \sim \mathbf{Normal}\left (0,1\right )$, and the data generating parameters $\beta_0 = [0.1, 0]^\top$, $\beta_{\lambda} = [-1, 2]^\top$, $\beta_{\gamma} = [-1, -1]^\top$, $q = [0.6, 0.4]^\top$, $\iota=0.01$ and $\psi=1$. 

Given the above model, we can repeat a similar experiment to the one from the main paper and learn the parameters on a grid to study the shape of the SimBa-CL surfaces. In this study, we consider only fully factorized SimBa-CL with and without feedback and we have excluded the general partition case. As for the main paper we set $P=1024$.

The comments on $\beta_0,\beta_\lambda,\beta_\gamma,q$ are the same as for the homogeneous scenario. $\iota$ and $\psi$ needs some additional attention. Firstly notice that the more we increase $\psi$ the more the spatial effect is strong and we highly penalise infected that are far away. This automatically tells us that over a certain threshold, it will be useless to increase $\psi$ as the penalization is already very strong. Secondly, there is an obvious identifiability issue when looking at $\iota$ and $\psi$ together. Indeed, increasing $\psi$ and decreasing $\iota$ will lead to similar epidemics, where the difference is that most of the infections are either coming from the spatial interaction or the environment. This ``banana'' shape makes these two parameters hard to learn and we have to be careful when reporting uncertainty around them to avoid being overconfident on a local maxima.

\begin{figure}[httb!]
    \centering
    \includegraphics[width=0.65\textwidth]{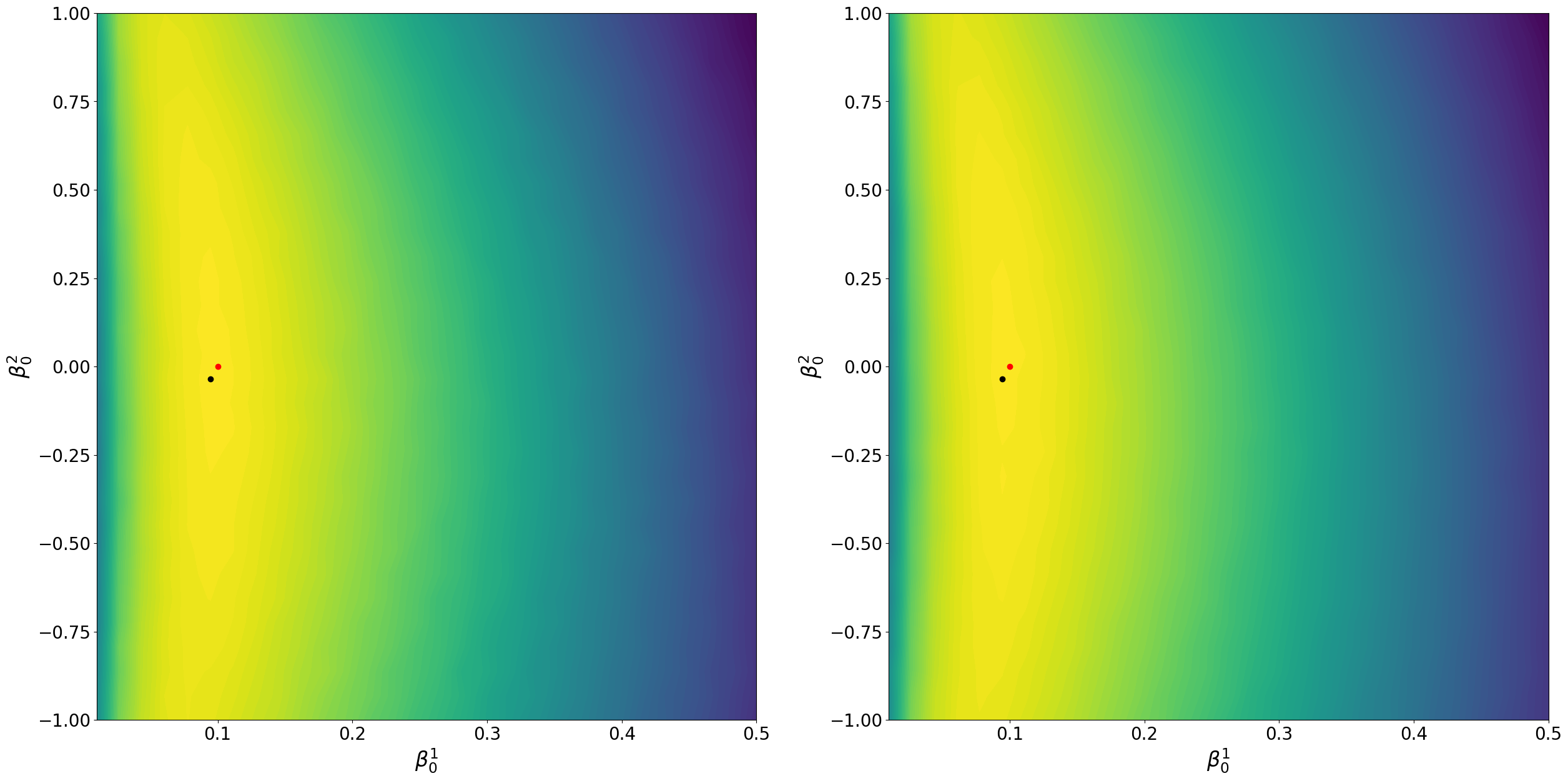}
    \caption{Profile likelihood for $\beta_0$ in spatial SIS.}
\end{figure}

\begin{figure}[httb!]
    \centering
    \includegraphics[width=0.65\textwidth]{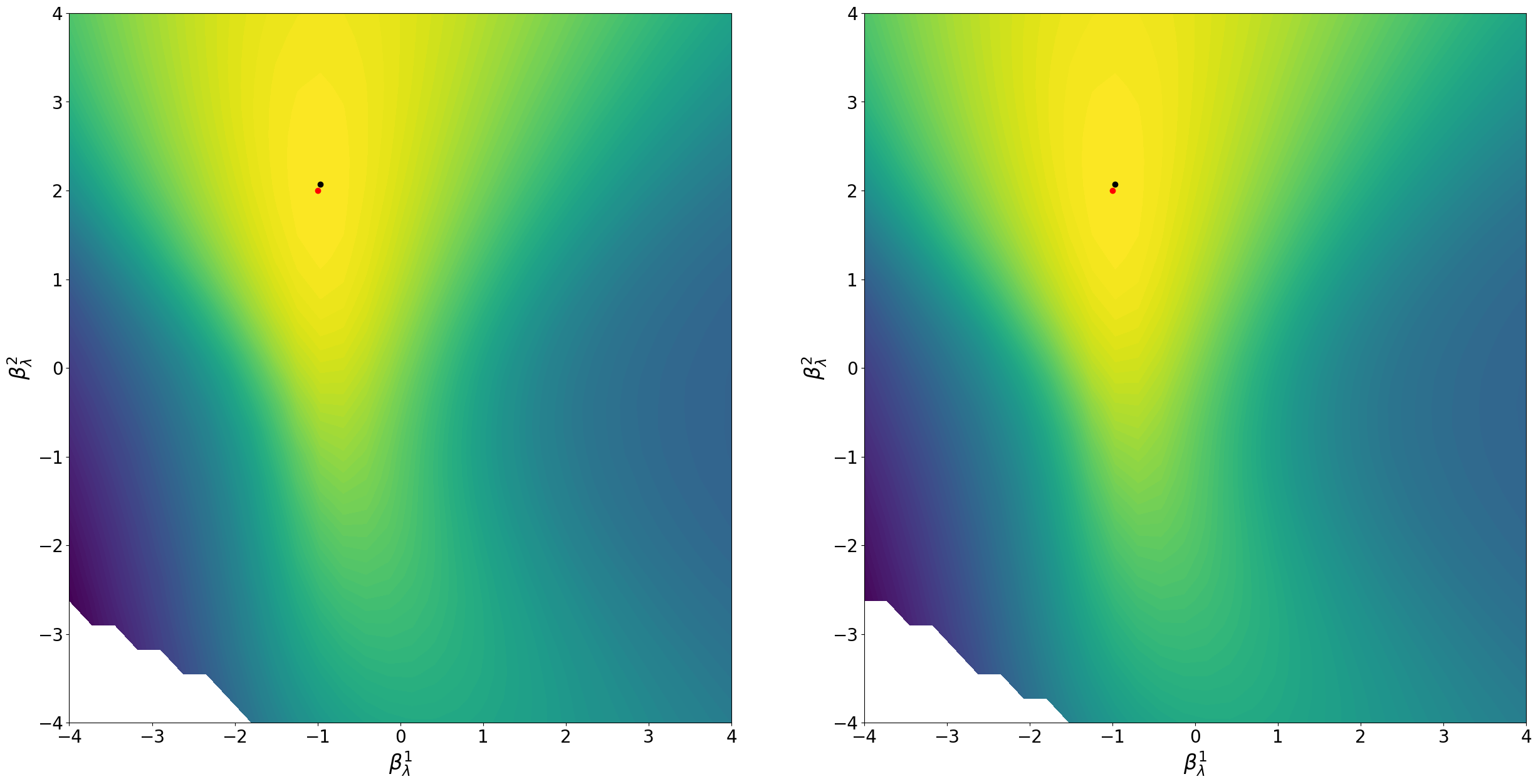}
    \caption{Profile likelihood for $\beta_\lambda$ in spatial SIS.}
\end{figure}

\begin{figure}[httb!]
    \centering
    \includegraphics[width=0.65\textwidth]{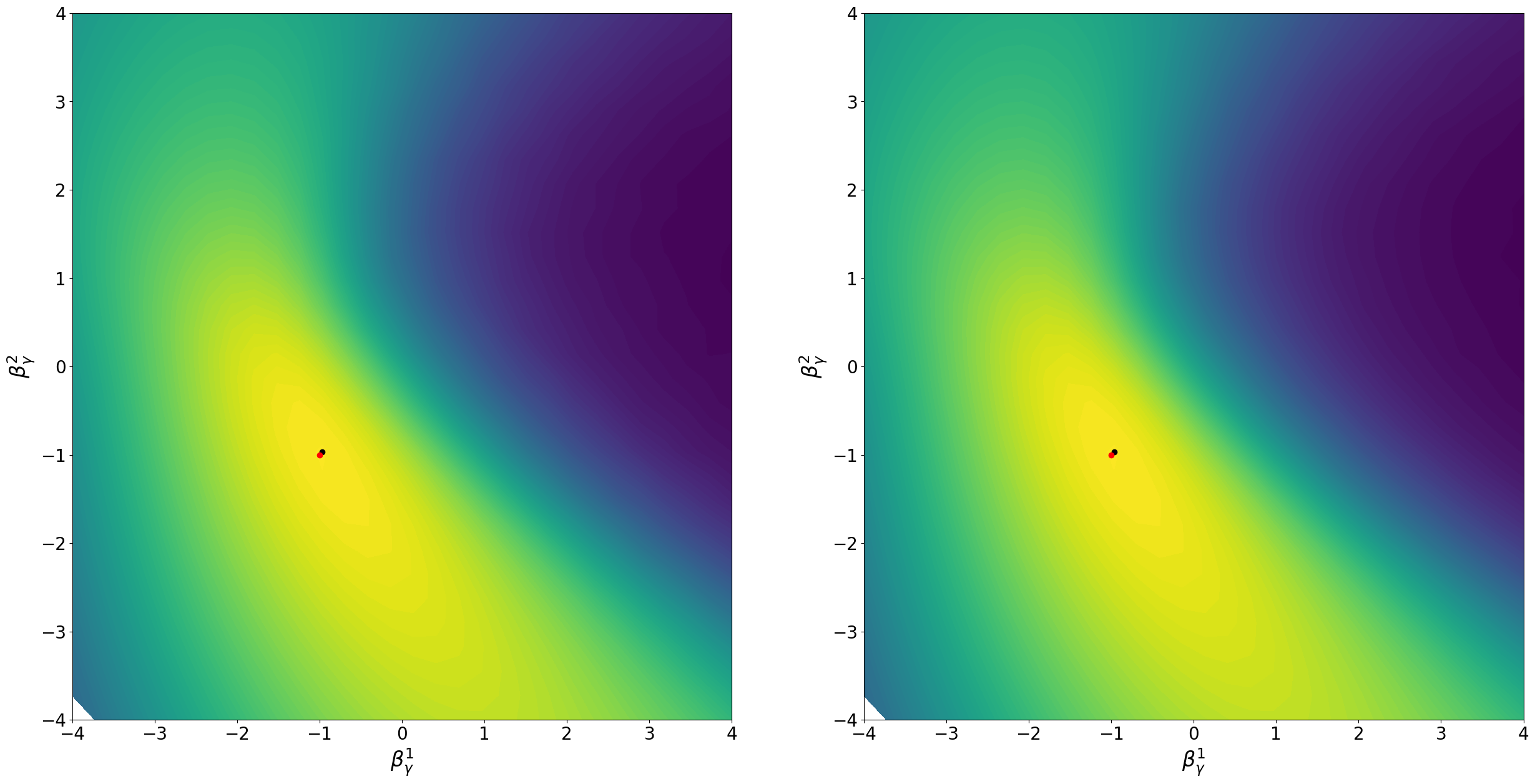}
    \caption{Profile likelihood for $\beta_\gamma$ in spatial SIS.}
\end{figure}

\begin{figure}[httb!]
    \centering
    \includegraphics[width=0.65\textwidth]{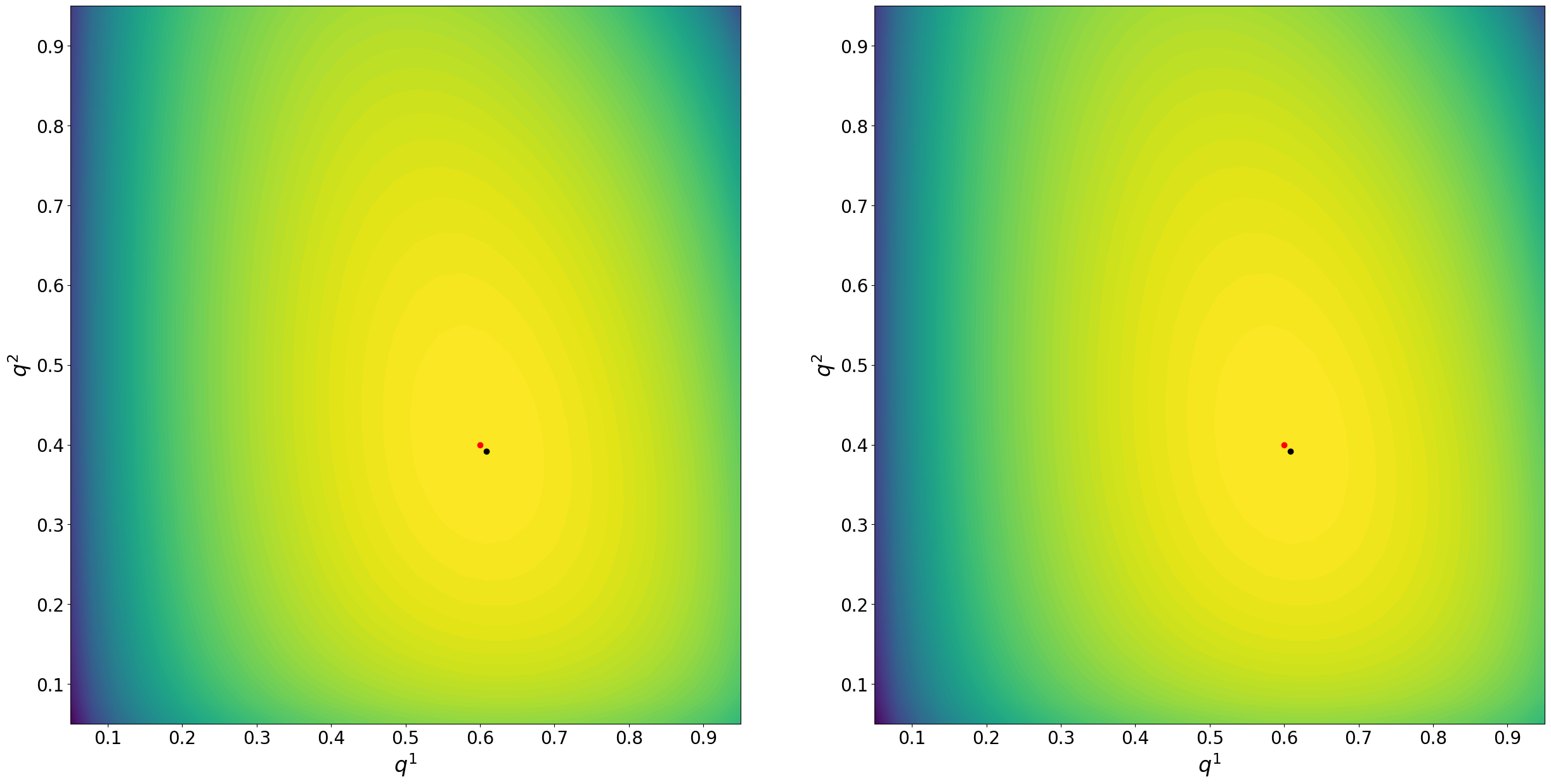}
    \caption{Profile likelihood for $q$ in spatial SIS.}
\end{figure}

\begin{figure}[httb!]
    \centering
    \includegraphics[width=0.65\textwidth]{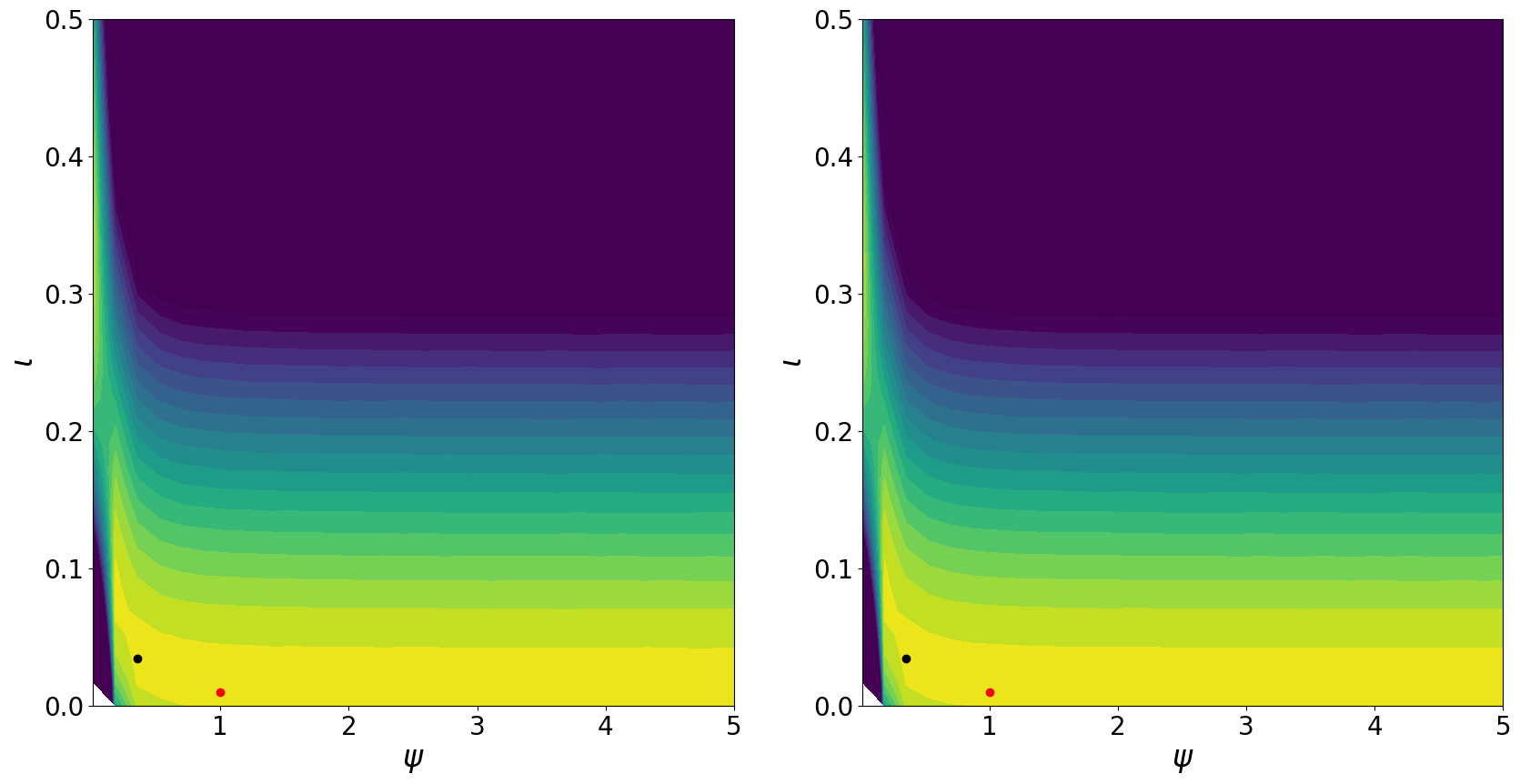}
    \caption{Profile likelihood for $\iota$ and $\psi$ in spatial SIS.}
    \label{fig:log_like_surfaces_spatial_app}
\end{figure}

\newpage
\subsection{FM modelling choices}

When working with the foot and mouth data we can have different modelling choices. One option is to learn the initial probability of infection as a common parameter across infected farms: 
\begin{align}
    &p(x_0^n|\theta) = 
    \begin{bmatrix}
        1-p_0 \mathbb{I}(n \in F_I)\\ 
        p_0 \mathbb{I}(n \in F_I)	\\
        0\\
        0
    \end{bmatrix},\\
\end{align}
where $F_I$ is the set of infected farms over time and so $n$ belongs to $F_I$ if it will be observed infected at some time $t$. This leads to satisfying likelihood scores, but to unreal epidemics, as it sets $p_0\approx 1$ meaning that all the epidemics observed in the future are already infected at time $0$ and no spatial interaction is learned. 

Another option is to reverse engineer the infection process and use a geometric distribution for the infection time:
\begin{align}
    &p(x_0^n|\theta) = 
    \begin{bmatrix}
        1-\left (1-e^{-\gamma} \right ) e^{-\gamma \tau_n}\\ 
        \left (1-e^{-\gamma} \right ) e^{-\gamma \tau_n}	\\
        0\\
        0
    \end{bmatrix},\\
\end{align}
where $\tau_n$ is the infection time of individual $n$. This again leads to satisfying likelihood scores, but it is susceptible to criticism given that we are informing the initial distribution with future observation.

More discussion can follow regarding the choices of the spatial kernel and covariates. We decided to not normalize the covariates and follow the approach of \cite{jewell2009bayesian}, even though normalization of the covariates might lead to better results as explained in \cite{jewell2013bayesian}. In terms of the spatial kernel, an obvious choice is the use of a Gaussian spatial kernel, where the decay is modeled through a Gaussian function.

Even though optimal modeling choices for the FM data is an interesting problem, this is beyond the scope of this experiment, which is included as a pure real data application.

\subsection{Limitations and extensions}

In this section, we provide some experiments on the limitations and extensions of SimBa-CL.

\subsubsection{Conditional SimBa-CL}

As unconditional simulations from the model could lead to inconsistencies with the data, we try to test a conditional version of SimBa-CL without feedback. Specifically, at each time step, we simulate from $\tilde{p}\left (x^n_{t}|y^n_{[t]}, x_{[0:t-1]}^{\setminus n}\right )$ instead from the model and use these samples for the Monte Carlo estimate. 

We compare the unconditional and the conditional SimBa-CL in terms of log-likelihood mean and variance, and average filtering performance. The former is self-explanatory, the latter is considering the ground truth for the latent process and computing the average prediction probability of the correct state over time and population. We consider our baseline individual-based SIS model with a varying population size $N=100, 1000, 2000$, a varying environmental effect $\iota=0.00001,0.001,0.1$ corresponding to the low, medium, high scenarios, and a varying $q=[0.2, 0.1], [0.6,0.4], [0.9,0.9]$ corresponding to the low, medium, high scenarios. We also set $P=500$. Table \ref{tab:condsimba} shows that the conditional SimBa-CL is generally associated with a lower variance in both the log-likelihood and the filtering performance. This effect is less evident for large population sizes and large environmental effects.

\begin{table}[]
    \centering
    \begin{tabular}{lrllll}
    \hline
     Simba-CL      &   $N$ & $\epsilon$   & $q$    & log-likelihood     & filtering performance   \\
    \hline
     Unconditional &   100 & low          & low    & -5131.96(2.6794)   & 0.84(0.0005827)         \\
     Conditional   &   100 & low          & low    & -5100.07(0.0432)   & 0.83(6.01e-05)          \\
     Unconditional &   100 & low          & medium & -6785.45(3.4968)   & 0.97(0.0018434)         \\
     Conditional   &   100 & low          & medium & -6738.07(0.0425)   & 1.0(6.9e-06)            \\
     Unconditional &   100 & low          & high   & -3340.12(3.4506)   & 1.0(0.00021)            \\
     Conditional   &   100 & low          & high   & -3295.73(0.0156)   & 1.0(3e-07)              \\
     Unconditional &   100 & medium       & low    & -5179.42(1.8061)   & 0.85(0.0010506)         \\
     Conditional   &   100 & medium       & low    & -5143.67(0.2772)   & 0.87(0.0002244)         \\
     Unconditional &   100 & medium       & medium & -8313.9(2.3432)    & 0.91(0.0001687)         \\
     Conditional   &   100 & medium       & medium & -8276.77(0.0822)   & 0.91(6.2e-06)           \\
     Unconditional &   100 & medium       & high   & -5117.49(4.3901)   & 0.99(8.89e-05)          \\
     Conditional   &   100 & medium       & high   & -5013.39(0.0461)   & 0.99(5e-07)             \\
     Unconditional &   100 & high         & low    & -4913.78(0.0408)   & 0.76(2.04e-05)          \\
     Conditional   &   100 & high         & low    & -4913.82(0.0245)   & 0.76(9e-06)             \\
     Unconditional &   100 & high         & medium & -8461.96(0.1054)   & 0.89(2.59e-05)          \\
     Conditional   &   100 & high         & medium & -8462.15(0.0443)   & 0.89(5.8e-06)           \\
     Unconditional &   100 & high         & high   & -6001.89(0.1919)   & 0.98(4.7e-06)           \\
     Conditional   &   100 & high         & high   & -5997.51(0.024)    & 0.98(5e-07)             \\
     Unconditional &  1000 & low          & low    & -50951.32(1.1928)  & 0.82(0.0001415)         \\
     Conditional   &  1000 & low          & low    & -50947.66(0.0539)  & 0.82(2.23e-05)          \\
     Unconditional &  1000 & low          & medium & -82876.23(0.6251)  & 0.91(7.08e-05)          \\
     Conditional   &  1000 & low          & medium & -82876.2(0.0689)   & 0.91(2.9e-06)           \\
     Unconditional &  1000 & low          & high   & -53233.93(1.8217)  & 0.99(1.39e-05)          \\
     Conditional   &  1000 & low          & high   & -53204.11(0.0351)  & 0.99(2e-07)             \\
     Unconditional &  1000 & medium       & low    & -50693.04(1.0425)  & 0.8(7.96e-05)           \\
     Conditional   &  1000 & medium       & low    & -50681.28(0.064)   & 0.8(1.34e-05)           \\
     Unconditional &  1000 & medium       & medium & -81994.59(0.4083)  & 0.91(4.73e-05)          \\
     Conditional   &  1000 & medium       & medium & -81992.01(0.0737)  & 0.91(3e-06)             \\
     Unconditional &  1000 & medium       & high   & -55111.02(1.5721)  & 0.99(9.2e-06)           \\
     Conditional   &  1000 & medium       & high   & -55098.72(0.0323)  & 0.99(2e-07)             \\
     Unconditional &  1000 & high         & low    & -50277.48(0.0492)  & 0.77(7.6e-06)           \\
     Conditional   &  1000 & high         & low    & -50277.68(0.0271)  & 0.77(3.1e-06)           \\
     Unconditional &  1000 & high         & medium & -84675.37(0.0926)  & 0.89(4.8e-06)           \\
     Conditional   &  1000 & high         & medium & -84674.12(0.0404)  & 0.89(1.2e-06)           \\
     Unconditional &  1000 & high         & high   & -61529.04(0.3132)  & 0.98(1.5e-06)           \\
     Conditional   &  1000 & high         & high   & -61511.64(0.0268)  & 0.98(1e-07)             \\
     Unconditional &  2000 & low          & low    & -101230.21(1.987)  & 0.81(0.0001159)         \\
     Conditional   &  2000 & low          & low    & -101203.51(0.1243) & 0.81(2e-05)             \\
     Unconditional &  2000 & low          & medium & -163284.12(2.081)  & 0.91(3.94e-05)          \\
     Conditional   &  2000 & low          & medium & -163266.38(0.0707) & 0.91(2.1e-06)           \\
     Unconditional &  2000 & low          & high   & -110632.57(3.8001) & 0.99(5.5e-06)           \\
     Conditional   &  2000 & low          & high   & -110567.61(0.0306) & 0.99(1e-07)             \\
     Unconditional &  2000 & medium       & low    & -101850.23(0.3624) & 0.81(7.66e-05)          \\
     Conditional   &  2000 & medium       & low    & -101849.97(0.057)  & 0.82(1.44e-05)          \\
     Unconditional &  2000 & medium       & medium & -163018.38(0.1773) & 0.91(3.05e-05)          \\
     Conditional   &  2000 & medium       & medium & -163021.54(0.0658) & 0.91(2e-06)             \\
     Unconditional &  2000 & medium       & high   & -107468.3(7.3799)  & 0.99(7.3e-06)           \\
     Conditional   &  2000 & medium       & high   & -107230.03(0.0386) & 0.99(1e-07)             \\
     Unconditional &  2000 & high         & low    & -99940.46(0.0389)  & 0.76(5e-06)             \\
     Conditional   &  2000 & high         & low    & -99940.57(0.0199)  & 0.76(2.2e-06)           \\
     Unconditional &  2000 & high         & medium & -172386.18(0.1524) & 0.89(4.1e-06)           \\
     Conditional   &  2000 & high         & medium & -172383.19(0.0434) & 0.89(1e-06)             \\
     Unconditional &  2000 & high         & high   & -121700.89(0.124)  & 0.98(1e-06)             \\
     Conditional   &  2000 & high         & high   & -121698.41(0.02)   & 0.98(1e-07)             \\
    \hline
    \end{tabular}
    \caption{Log-likelihood and filtering performance for unconditional and conditional SimBa-CL across different scenarios.}
    \label{tab:condsimba}
\end{table}

\subsubsection{Gradient considerations}

We are differentiate quantities of the form $\sum_{n \in [N]} \log \left ( p(y_{[T]}^n|\theta)\right )$, with $\tilde{p}$ for SimBa-CL without feedback. This reduces to differentiating $p(y_{[T]}^n|\theta)$:
\begin{equation}
    \begin{split}
        \frac{\partial p(y_{[T]}^n|\theta)}{\partial \theta}
        &= \sum_{x_{[0:T-1]}^{\setminus n}} p\left ( x_{[0:T-1]}^{\setminus n}|\theta \right )  \frac{\partial p\left (y_{[T]}^n|x_{[0:T-1]}^{\setminus n},\theta\right )}{\partial \theta}\\
        &+
        \sum_{x_{[0:T-1]}^{\setminus n}}  p\left (y_{[T]}^n|x_{[0:T-1]}^{\setminus n},\theta\right ) \frac{\partial p\left ( x_{[0:T-1]}^{\setminus n}|\theta \right ) }{\partial \theta}.
    \end{split}
\end{equation}
In our experiments, we consider TensorFlow's default automatic differentiation, which consists of the approximation:
$$
\frac{\partial p(y_{[T]}^n|\theta)}{\partial \theta}
        \approx \frac{1}{P} \sum_{i \in [P]} \frac{\partial p\left (y_{[T]}^n|x_{[0:T-1]}^{(i),\setminus n},\theta\right )}{\partial \theta}.
$$

As discussed in the main paper, we could also use the Gumbel-Softmax trick \citep{gumbel1954statistical,maddison2016concrete,jang2016categorical}, which approximates the one-hot encoding representation of categorical random variable with a probability vector that converges to the one-hot encoding as the temperature parameter goes to zero. Indeed, suppose we have an $M$-dimensional categorical distribution $X \sim Cat(\pi)$, with $X$ one-hot encoding representation. We can simulate it using the Gumbel distribution as follows:
$$
X =  \text{one-hot}\left( \arg \max_{k=1,\dots,M} \{ \log(\pi^{(k)}) + G_k \}_k\right ), \quad \text{ with } G_k \sim Gumb(0,1) \text{ i.i.d.},
$$
with $\text{one-hot}(k)$ being the vector of all zeros with $1$ in the $k$th position. Now, the $\arg \max$ can be substituted with a Softmax to produce a continuous function:
$$
X^{(k)} \approx  \frac{ \exp  \left (\frac{\log(\pi^{(k)}) + G_k}{\tau}\right )}{\sum_{j} \exp  \left (\frac{\log(\pi^{(j)}) + G_j}{\tau}\right )} \quad \text{for } j=1,\dots,M,
$$
which converges to the $\arg \max$, i.e. a one-hot encoding vector, when $\tau \to 0$. We are now considering a continuous function of i.i.d. random variables that are independent of the parameters of the categorical distribution, meaning that we can apply the reparameterization trick \citep{blundell15}, and exchange expectations and derivatives. Applying this to our scenario we get something along the lines of:
\begin{equation}
    \begin{split}
        \frac{\partial p(y_{[T]}^n|\theta)}{\partial \theta}
        &= 
        \frac{\partial }{\partial \theta} \mathbb{E}_{x_{[0:T-1]}} \left [ p\left (y_{[T]}^n|x_{[0:T-1]}^{\setminus n},\theta\right ) \right ]\\
        &\approx 
        \frac{\partial }{\partial \theta} \mathbb{E}_{G_{[N],[0:T-1]}} \left [ p\left (y_{[T]}^n|s(G_{[N]\setminus n,[0:T-1]},\theta,\tau),\theta\right ) \right ]\\
        &= 
        \mathbb{E}_{G_{[N],[0:T-1]}} \left [ \frac{\partial }{\partial \theta}p\left (y_{[T]}^n|s(G_{[N]\setminus n,[0:T-1]},\theta,\tau),\theta\right ) \right ]\\
        &= 
        \mathbb{E}_{G_{[N],[0:T-1]}} \left [ \left \{ \frac{\partial s(G_{[N]\setminus n,[0:T-1]},\theta,\tau)}{\partial \theta} \frac{\partial  p\left (y_{[T]}^n|s,\theta\right )}{\partial s} \right \}_{s=s(G_{[N]\setminus n,[0:T-1]},\theta,\tau)} \right.\\
        &\left. \qquad \qquad \qquad \qquad+ \left \{ \frac{\partial p\left (y_{[T]}^n|s,\theta\right )}{\partial \theta} \right \}_{s=s(G_{[N]\setminus n,[0:T-1]},\theta,\tau)} \right ],
    \end{split}
\end{equation}
where $G_{[N],[0:T-1]} \coloneqq (G_{n,t})_{n \in [N],t \in [0:T-1]}$ are i.i.d. Gumbel random variables, and $s$ is a function of the Gumbel random variables, the parameter $\theta$, and the temperature parameter $\tau$, which is a composition of continuous functions. The Gumbel-Softmax relaxation and the reparameterization trick are implemented in tensorflow\_probability via the distribution \textit{RelaxedOneHotCategorical}.

In this section, we compare the default \textit{OneHotCategorical}, which we used in our experiments, with the \textit{RelaxedOneHotCategorical}, in terms of gradients estimate and computational cost for our baseline SIS model.

Starting from gradient estimates, we consider our baseline SIS model for a population size $N=100,1000$ and we compute the gradient of $\beta_{\lambda}$ both at the data-generating parameter DGP and in $\beta_{\lambda}=[0,0]^\top$. To compute the gradient we consider both \textit{OneHotCategorical} which we refer to as the ``Without $\tau$'' scenario and \textit{RelaxedOneHotCategorical} for a $\tau=0.1,0.01,0.0001$. We also set $P=500$. Table \ref{tab:gradients_relax_vs} shows that all gradient estimates have the same sign in all scenarios, with these estimates being particularly close for large populations. It is also important to observe that TensorFlow's default implementation underestimates the variance of the gradient.

\begin{table}[]
    \centering
    \begin{tabular}{lrlllll}
    \hline
                    &    N & Without $\tau$   & $\tau=0.1$       & $\tau=0.01$      & $\tau=0.0001$    \\
    \hline
     At DGP (0)     &  100 & -0.0023(8e-05)   & -0.0042(0.00436) & -0.0038(0.00941) & -0.0026(0.00044) \\
     Not at DGP (0) &  100 & 0.0093(0.00021)  & 0.0183(0.00278)  & 0.0163(0.00938)  & 0.012(0.00324)   \\
     At DGP (0)     & 1000 & -0.0006(8e-05)   & -0.0011(0.00116) & -0.0019(0.00794) & -0.0009(0.00113) \\
     Not at DGP (0) & 1000 & 0.0289(0.00027)  & 0.057(0.0051)    & 0.0521(0.01858)  & 0.0394(0.01097)  \\
     At DGP (1)     &  100 & 0.0005(2e-05)    & 1e-04(0.00179)   & 0.0013(0.00679)  & 0.0005(0.00033)  \\
     Not at DGP (1) &  100 & -0.0011(0.00017) & -0.0024(0.00125) & -0.0028(0.00369) & -0.0022(0.00426) \\
     At DGP (1)     & 1000 & -0.0002(2e-05)   & -0.0003(0.00054) & -0.0005(0.00223) & -0.0003(0.00055) \\
     Not at DGP (1) & 1000 & -0.0188(9e-05)   & -0.0187(0.00179) & -0.0197(0.01062) & -0.0204(0.00806) \\
    \hline
    \end{tabular}
    \caption{Gradients estimates for $\beta_{\lambda}$ in different scenarios. $(0)$ refers to the first component. $(1)$ refers to the second component.}
    \label{tab:gradients_relax_vs}
\end{table}

Moving to the computational considerations, we again consider our baseline SIS model for a population size $N=10, 100, 500, 1000, 1500, 2000, 2500$, a time horizon $T=10,50,100$, and a number of particles $P=100,300,500$. In Table \ref{tab:comp_gradients_relax_vs} we report the time in seconds of running a gradient computation of TensorFlow's default implementation, ``Without $\tau$'', and the continuous relaxation, ``With $\tau$''. To fit within the margin we have reported only a few columns of all the simulations. The standard deviation of the computational time is reported in brackets and computed across 100 replicates. Overall we observe that the ``Without $\tau$'' is 1.5 times faster than the ``With $\tau$'' approach. Moreover, it is more frequent for the ``With $\tau$'' approach to go out-of-memory (OOM). For instance, we were not able to run a single gradient computation with $N=2500$. This is due to the heavier computational graph required by the ``With $\tau$'' approach.

\begin{table}[]
    \centering
    \begin{tabular}{ll|ll|lll|l}
    \hline
                    & N    & T=10, P=300   & P=500   & T=50, P=100   & P=300   & P=500   & T=100, P=100   \\
    \hline
     Without $\tau$ & 10   & 0.16(0.07)    & 0.16(0.06)    & 0.15(0.07)    & 0.28(0.06)    & 0.28(0.06)    & 0.3(0.07)      \\
     With $\tau$    &      & 0.23(0.08)    & 0.23(0.22)    & 0.22(0.08)    & 0.38(0.08)    & 0.38(0.08)    & 0.39(0.08)     \\
     Without $\tau$ & 100  & 0.16(0.07)    & 0.15(0.07)    & 0.15(0.07)    & 0.3(0.07)     & 0.3(0.07)     & 0.29(0.07)     \\
     With $\tau$    &      & 0.21(0.08)    & 0.21(0.08)    & 0.21(0.08)    & 0.39(0.08)    & 0.39(0.08)    & 0.39(0.08)     \\
     Without $\tau$ & 500  & 0.16(0.14)    & 0.17(0.07)    & 0.2(0.07)     & 0.31(0.07)    & 0.33(0.07)    & 0.39(0.07)     \\
     With $\tau$    &      & 0.2(0.08)     & 0.24(0.08)    & 0.3(0.08)     & 0.39(0.08)    & 0.46(0.08)    & 0.58(0.08)     \\
     Without $\tau$ & 1000 & 0.16(0.07)    & 0.22(0.07)    & 0.32(0.07)    & 0.29(0.07)    & 0.42(0.07)    & 0.62(0.08)     \\
     With $\tau$    &      & 0.21(0.08)    & 0.33(0.09)    & 0.48(0.09)    & 0.41(0.08)    & 0.65(0.09)    & OOM            \\
     Without $\tau$ & 1500 & 0.18(0.08)    & 0.3(0.09)     & 0.46(0.09)    & OOM           & OOM           & OOM            \\
     With $\tau$    &      & 0.25(0.08)    & 0.45(0.09)    & 0.7(0.1)      & OOM           & OOM           & OOM            \\
     Without $\tau$ & 2000 & 0.19(0.07)    & 0.39(0.2)     & 0.59(0.09)    & OOM           & OOM           & OOM            \\
     With $\tau$    &      & 0.27(0.08)    & 0.57(0.08)    & OOM           & OOM           & OOM           & OOM            \\
     Without $\tau$ & 2500 & 0.2(0.07)     & 0.45(0.08)    & OOM           & OOM           & OOM           & OOM            \\
     With $\tau$    &      & OOM           & OOM           & OOM           & OOM           & OOM           & OOM            \\
    \hline
    \end{tabular}
    \caption{Time in seconds in different scenarios. ``OOM'' stands for out-of-memory. Columns that are grouped together have the same $T$. }
    \label{tab:comp_gradients_relax_vs}
\end{table}

\subsection{The effect of $P$}

In this section, we graphically explore the effect of increasing $P$. We consider our baseline SIS model with $P = 10, 100, 1000$ and simulate data from it. Next, we define a grid for the parameter $\beta_\lambda$ and compute SimBa-CL without feedback for each element of the grid. Specifically, we set $\beta_\lambda$ to each grid value while keeping the true values for the other parameters. The resulting log-likelihood surfaces are shown in Figure \ref{fig:simba_P}.

We observe that the primary effect of $P$ is on the noisiness of the log-likelihood surface. For small values of $P$ (left-hand side of the figure), the log-likelihood surface is less smooth compared to that for larger values of $P$ (right-hand side of the figure). This effect is particularly noticeable in regions with low log-likelihood values. However, the recovery of the parameter does not appear to be significantly affected.

\begin{figure*}[httb!]
    \centering
    \includegraphics[width=\textwidth]{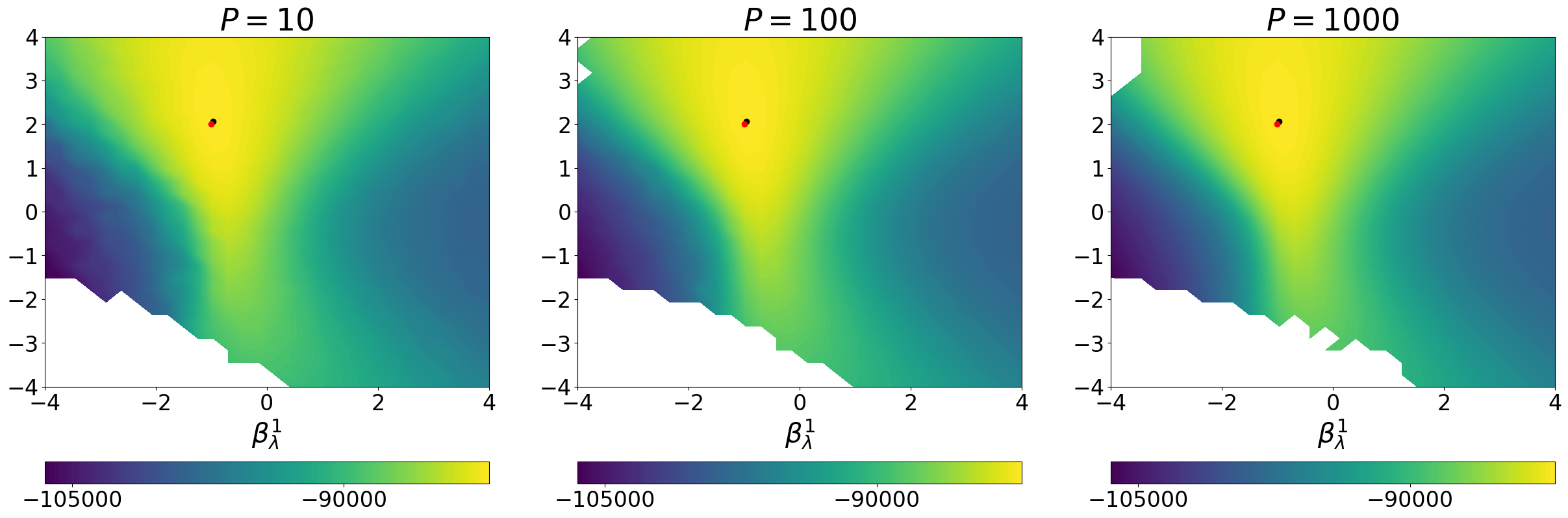}
    \caption{Profile log-likelihood surfaces for $\beta_\lambda$ obtained from fully factorized SimBa-CL without feedback for different values of $P$. The red dot indicates the data-generating parameter, and the black dot denotes the maximum on the grid. }
    \label{fig:simba_P}
\end{figure*}


\end{document}